\def \IsDraft{} 

\documentclass[11pt]{article}

\usepackage[
pagebackref,backref=true,
hyperfootnotes=false,
colorlinks=true,
urlcolor=externallinkcolor,
linkcolor=internallinkcolor,
filecolor=externallinkcolor,
citecolor=internallinkcolor,
breaklinks=true,
pdfstartview=FitH,
pdfpagelayout=OneColumn]{hyperref}

\usepackage[numbers]{natbib} 

\usepackage{titlesec}
\usepackage{mathtools}
\usepackage{bbm}
\usepackage{amsmath,amsfonts,amsmath,amssymb,color,amsthm,boxedminipage,url,xparse,fullpage}
\definecolor{internallinkcolor}{rgb}{0,.5,0}
\usepackage{amsthm} 

\usepackage{xspace}
\usepackage{nicefrac}

\usepackage{cleveref}

\usepackage{tikz}

\providecommand{\remove}[1]{}

\ifdefined\IsDraft
    \newcommand{\authnote}[2]{{\bf [{\color{red} #1's Note:} {\color{blue} #2}]}}
\else
    \newcommand{\authnote}[2]{}
\fi

\titleclass{\subsubsubsection}{straight}[\subsection]

\newcounter{subsubsubsection}[subsubsection]
\renewcommand\thesubsubsubsection{\thesubsubsection.\arabic{subsubsubsection}}

\titleformat{\subsubsubsection}
{\normalfont\normalsize\bfseries}{\thesubsubsubsection}{1em}{}
\titlespacing*{\subsubsubsection}
{0pt}{3.25ex plus 1ex minus .2ex}{1.5ex plus .2ex}

\makeatletter
\renewcommand\paragraph{\@startsection{paragraph}{5}{\z@}%
	{3.25ex \@plus1ex \@minus.2ex}%
	{-1em}%
	{\normalfont\normalsize\bfseries}}
\renewcommand\subparagraph{\@startsection{subparagraph}{6}{\parindent}%
	{3.25ex \@plus1ex \@minus .2ex}%
	{-1em}%
	{\normalfont\normalsize\bfseries}}
\def\toclevel@subsubsubsection{4}
\def\toclevel@paragraph{5}
\def\toclevel@paragraph{6}
\def\l@subsubsubsection{\@dottedtocline{4}{7em}{4em}}
\def\l@paragraph{\@dottedtocline{5}{10em}{5em}}
\def\l@subparagraph{\@dottedtocline{6}{14em}{6em}}
\makeatother

\newenvironment{mybox}{\begin{center}\begin{tabular}{|p{0.97\linewidth}|c|}   \hline} {  \\ \hline \end{tabular} \end{center}}

\let\originalleft\left
\let\originalright\right
\renewcommand{\left}{\mathopen{}\mathclose\bgroup\originalleft}
\renewcommand{\right}{\aftergroup\egroup\originalright}

\newcommand{\ie}  {i.e.,\xspace}
\newcommand{\eg}  {e.g.,\xspace}

\newcommand{\wrt} {with respect to\xspace}
\newcommand{\wlg} {without loss of generality\xspace}
\newcommand{\cf}{cf.,\xspace}
\newcommand{\vecc}[1]{\left\lVert #1 \right\rVert}

\newcommand{\set}[1]{\left\{#1\right\}}
\newcommand{\eqdef}{:=}

\newcommand{\N}{{\mathbb{N}}}

\newcommand{\zo}{\set{0,1}}

\newcommand{\zn}{{\zo^n}}
\newcommand{\zs}{{\zo^\ast}}
\newcommand{\zl}{{\zo^\ell}}

\newcommand{\condition}{\;\ifnum\currentgrouptype=16 \middle\fi|\;}

\newcommand{\getsr}{\mathbin{\stackrel{\mbox{\tiny R}}{\gets}}}

\newcommand{\eps}{\varepsilon}

\newcommand{\la}{\gets}

\makeatletter
\@ifpackageloaded{complexity}{}{%
\usepackage[classfont=roman,langfont=roman,funcfont=roman]{complexity}}
\makeatother

\newcommand{\Exp}{\operatorname*{E}}

\newcommand{\Hall}{\operatorname{H}}
\newcommand{\Hmin}{\operatorname{H_{\infty}}}

\newcommand{\negl}{\operatorname{neg}}

\newcommand{\Supp}{\operatorname{Supp}}

\newcommand{\MathFam}[1]{\mathcal{#1}}

\newcommand{\Dom}{\MathFam{D}}

\newcommand{\MathAlg}[1]{\mathsf{#1}}

\newcommand{\Ensuremath}[1]{\ensuremath{#1}\xspace}

\newcommand{\MathAlgX}[1]{\Ensuremath{\MathAlg{#1}}}

\renewcommand{\cref}{\Cref}
\usepackage{aliascnt}

\newtheorem{theorem}{Theorem}[section]

\newaliascnt{lemma}{theorem}
\newtheorem{lemma}[lemma]{Lemma}
\aliascntresetthe{lemma}
\crefname{lemma}{Lemma}{Lemmas}

\newaliascnt{claima}{theorem}
\newtheorem{claima}[claima]{Claim}
\aliascntresetthe{claima}
\crefname{claima}{Claim}{Claims}

\newaliascnt{constructiona}{theorem}
\newtheorem{constructiona}[constructiona]{Construction}
\aliascntresetthe{constructiona}
\crefname{constructiona}{Construction}{Constructions}

\newaliascnt{claim}{theorem}

\aliascntresetthe{claim}
\crefname{claim}{Claim}{Claims}

\newaliascnt{corollary}{theorem}

\aliascntresetthe{corollary}
\crefname{corollary}{Corollary}{Corollaries}

\newaliascnt{construction}{theorem}

\aliascntresetthe{construction}
\crefname{construction}{Construction}{Constructions}

\newaliascnt{fact}{theorem}

\aliascntresetthe{fact}
\crefname{fact}{Fact}{Facts}

\newaliascnt{proposition}{theorem}
\newtheorem{proposition}[proposition]{Proposition}
\aliascntresetthe{proposition}
\crefname{proposition}{Proposition}{Propositions}

\newaliascnt{conjecture}{theorem}

\aliascntresetthe{conjecture}
\crefname{conjecture}{Conjecture}{Conjectures}

\newaliascnt{definition}{theorem}
\newtheorem{definition}[definition]{Definition}
\aliascntresetthe{definition}
\crefname{definition}{Definition}{Definitions}

\newaliascnt{notation}{theorem}

\aliascntresetthe{notation}
\crefname{notation}{Notation}{Notation}

\newaliascnt{assertion}{theorem}

\aliascntresetthe{assertion}
\crefname{assertion}{Assertion}{Assertion}

\newaliascnt{assumption}{theorem}

\aliascntresetthe{assumption}
\crefname{assumption}{Assumption}{Assumption}

\newaliascnt{remark}{theorem}

\aliascntresetthe{remark}
\crefname{remark}{Remark}{Remarks}

\newaliascnt{question}{theorem}

\aliascntresetthe{question}
\crefname{question}{Question}{Questions}

\newaliascnt{example}{theorem}
\ifdefined\excludeexample
\else
   
\fi

\aliascntresetthe{example}
\crefname{exmaple}{Example}{Examples}

\crefname{equation}{Equation}{Equations}

\newaliascnt{proto}{theorem}

\newtheorem{proto}[proto]{Protocol}

\aliascntresetthe{proto}
\crefname{proto}{protocol}{protocols}

\newaliascnt{algo}{theorem}
\newtheorem{algo}[algo]{Algorithm}
\aliascntresetthe{algo}
\crefname{algo}{algorithm}{algorithms}

\newaliascnt{expr}{theorem}
\newtheorem{expr}[expr]{Experiment}
\aliascntresetthe{expr}
\crefname{experiment}{experiment}{experiments}

\def\FullBox{$\Box$}
\def\qed{\ifmmode\qquad\FullBox\else{\unskip\nobreak\hfil
\penalty50\hskip1em\null\nobreak\hfil\FullBox
\parfillskip=0pt\finalhyphendemerits=0\endgraf}\fi}

\def\qedsketch{\ifmmode\Box\else{\unskip\nobreak\hfil
\penalty50\hskip1em\null\nobreak\hfil$\Box$
\parfillskip=0pt\finalhyphendemerits=0\endgraf}\fi}

{\proof[#1]\list{}{\leftmargin=1cm\rightmargin=1cm}}
{\endproof%
  \vspace{10pt}
}

\newcommand{\ex}[1]{\Ex\left[#1\right]}
\newcommand{\Ex}{{\mathrm E}}
\renewcommand{\Pr}{{\mathrm {Pr}}}
\newcommand{\pr}[1]{\Pr\left[#1\right]}
\newcommand{\ppr}[2]{\Pr_{#1}\left[#2\right]}

\newcommand{\Ac}{\MathAlgX{A}}

\newcommand{\Bc}{\MathAlgX{B}}

\newcommand{\Fc}{\MathAlgX{F}}
\newcommand{\Dc}{\MathAlgX{D}}

\newcommand{\size}[1]{\left|#1\right|}

\newcommand{\Uni}{{\mathord{\mathcal{U}}}}

\newcommand{\prob}[1]{\mathsf{\textsc{#1}}}

\newcommand{\SD}{\prob{SD}}

\providecommand{\cL}{{\cal{L}}}

\newcommand{\ppt}{{\sc ppt}\xspace}

\def\cD{{\cal D}}
\def\cE{{\cal E}}

\def\cL{{\cal L}}

\def\cR{{\cal R}}
\def\cS{{\cal S}}

\def\cU{{\cal U}}

\newcommand{\Tableofcontents}{
\thispagestyle{empty}
\pagenumbering{gobble}
\clearpage
\tableofcontents
\thispagestyle{empty}
\clearpage
\pagenumbering{arabic}
}

\DeclareMathOperator{\Good}{good}

\makeatletter
\let\xx@thm\@thm
\AtBeginDocument{\let\@thm\xx@thm}
\makeatother

\DeclareMathSymbol{\shortminus}{\mathbin}{AMSa}{"39}

\newcommand{\g}{\mathcal{G}}

\newcommand{\Gc}{G}
\newcommand{\tO}{\widetilde{O}}
\newcommand{\calF}{\mathcal{F}}

\newcommand{\Inv}{\MathAlg{Inv}}

\newcommand{\MyAtop}[2]{\genfrac{}{}{0pt}{}{#1}{#2}}

\def\state{{\sf state}}
\newcommand{\cLv}{\overline{\cL}}
\newcommand{\Hsam}{\operatorname{H}}
\newcommand{\Hmax}{\operatorname{H_0}}

\newcommand{\Jt}{\widetilde{J}}
\newcommand{\XT}{\widetilde{X}}
\newcommand{\At}{\widetilde{\Ac}}
\newcommand{\ExtendOnex}{\MathAlg{ExtendOne}} 
\newcommand{\Lbar}{\overline{\cL}}

\newcommand{\RNum}[1]{\uppercase\expandafter{\romannumeral #1\relax}}

\newcommand{\SearchHeurBPP}{\class{SearchHeurBPP}}

\newcommand{\polySamp}{\class{PSamp}}

\newcommand{\random}{r}
\newcommand{\Random}{R}
\newcommand{\RelR}{{\mathcal{R}}}
\newcommand{\RelQ}{{\mathcal{Q}}}
\newcommand{\RelV}{{\mathcal{V}}}
\newcommand{\ioS}{i.o.\ }
\newcommand{\YB}{\mathcal{Y}}
\newcommand{\ys}{{y}}

\newcommand{\Canon}{\operatorname{Canon}}
\newcommand{\Hreal}{\Hall^{\operatorname{Real}}}
\newcommand{\Lang}{\mathsf{L}}
\newcommand{\rDelta}{\delta}
\newcommand{\qDelta}{\varepsilon}
\newcommand{\vDelta}{\delta}
\newcommand{\secParR}{n}
\newcommand{\secParQ}{m}
\newcommand{\secParV}{n}

\newcommand{\HSh}{\operatorname{H}}

\newcommand{\mapstocat}{@}
\newcommand{\mapstocdot}{.}

\def\nolimits{} 
\newcommand{\CF}{\MathAlg{CF}}
\newcommand{\eee}{e}
\newcommand{\Image}{\operatorname{Im}}

\newcommand{\kreal}{{k_{\textsc{real}}}}

\title{Inaccessible Entropy II:\\
	IE Functions and Universal One-Way Hashing\thanks{This is the final draft of this paper. The full version was published in the Theory of Computing \cite{HaitnerHolReVaWe20}. An extended abstract of this work appeared  appeared as ``Universal One-Way Hash Functions via Inaccessible Entropy'' in Eurocrypt 2010 \cite{HaitnerHolReVaWe10}.}}

\author{
	Iftach Haitner\thanks{Tel-Aviv University, Tel Aviv, Israel. E-mail:\texttt{iftachh@tau.ac.il}. Initial research and conference version was done while at Microsoft Research New England. Preparation of this journal version supported by ERC starting grant 638121 and Israel Science Foundation grant   666/19.}	
	\and Thomas Holenstein\thanks{Google, Zurich, Switzerland.  E-mail:\texttt{tholenst\mapstocat{}google\mapstocdot{}com}.}
	\and Omer Reingold\thanks{Stanford University, Stanford, California. E-mail:\texttt{reingold\mapstocat{}stanford\mapstocdot{}edu}. Initial research and conference version supported by US-Israel BSF grant 2006060. Preparation of this journal version supported by NSF grant  CCF-1763311.}
\and Salil Vadhan\thanks{Harvard University, Cambridge, Massachusetts. E-mail:\texttt{salil$\_$vadhan\mapstocat{}harvard\mapstocdot{}edu}. Initial research and conference version supported by NSF grants CNS-0831289 and US-Israel BSF grants 2006060 and 2010196. Preparation of this journal version supported by NSF grant CCF-1763299 and a Simons Investigator Award.}
\and Hoeteck Wee\thanks{ENS, Paris, France E-mail:\texttt{wee\mapstocat{}di\mapstocdot{}ens\mapstocdot{}fr}. Most of
	the 
	work was done while at Queens College, CUNY. Supported in part by PSC-CUNY Award \#6014939~40 and NSF CAREER Award CNS-0953626.}
}

\begin{document}
\sloppy
\maketitle

\begin{abstract}
This paper uses a variant of  the notion
of   \emph{inaccessible entropy} (Haitner, Reingold, Vadhan and Wee, STOC 2009), to give an alternative construction and proof for the fundamental result, first proved by Rompel (STOC 1990), that   \emph{Universal One-Way Hash Functions (UOWHFs)} can be based on any one-way functions. We observe that a small tweak of any one-way function $f$ is already a weak form of a UOWHF: consider
the function $F(x,i)$ that returns the   $i$-bit-long   prefix of $f(x)$. If $F$ were a UOWHF then given a random $x$ and $i$ it would be hard to come up with $x'\neq x$ such that $F(x,i)=F(x',i)$. While this may not be the case, we show (rather easily) that it is hard to sample $x'$ with almost full entropy among all the possible such values of $x'$. The rest of our construction simply amplifies and exploits this basic property.Combined with other recent work, the construction of three fundamental cryptographic primitives (Pseudorandom Generators, Statistically Hiding Commitments and UOWHFs) out of one-way functions is now  to a large extent unified. In particular, all three constructions rely on and manipulate computational notions of entropy in similar ways. Pseudorandom Generators rely on the well-established notion of pseudoentropy, whereas Statistically Hiding Commitments and UOWHFs rely on the newer notion of inaccessible entropy. In an additional result we reprove  the seminal result of Impagliazzo and Levin (FOCS 1989): a reduction from ``uniform distribution'' average-case   complexity problems to ones with arbitrary  (polynomial-time samplable)  distributions. We do that using  techniques similar to those   we use to construct UOWHFs from one-way functions,  where the source of this similarity is the use of a notion similar to    inaccessible entropy. This draws an interesting connection between two seemingly separate lines of research: average-case   complexity and universal one-way hash-functions. 
\end{abstract}

\Tableofcontents

\section{Introduction}
Unlike the more common notions of computational entropy, \eg
\emph{pseudoentropy}  \cite{HastadImLeLu99}, that are only useful
as a \emph{lower} bound on
the ``computational entropy'' of a distribution, \emph{accessible entropy}
is an \emph{upper} bound on computational entropy.  
In particular, it measures the entropy
(Shannon, or other types)    
of the distribution
computed  
by a
resource-bounded machine.   

The inaccessible entropy of the
distribution     
is  the gap between its (real)
entropy  and its accessible entropy.
Inaccessible entropy was introduced by \citet*{HaitnerReVaWe09} 
as a means  
to give a simpler construction and proof of statistically hiding commitment from one-way functions (reproving the result
of \cite{HaitnerNgOnReVa09}),  
and  to construct constant-round statistically hiding commitment    from constant-round zero-knowledge proof  for \NP. In this
article  
introduce simpler variant of their  notion to give an alternative construction and proof for the fundamental result,  first proved by \citet{Rompel90}, that  
\emph{Universal One-Way Hash Functions (UOWHFs)}  
can be based on  one-way functions.  In an additional result, we reprove the seminal result of  \citet{ImpLev90}: a reduction from ``uniform distribution'' average-case  
complexity problems to ones with arbitrary (though polynomial samplable one) distributions. The latter is proved  using  similar techniques to the ones we use  to construct UOWHFs from  one-way functions, where the source of this similarity is the use of a similar notion of inaccessible entropy. This draws an interesting connection between two seemingly separate lines of research: 
average-case  
complexity and universal one-way hash-functions. 

We start by discussing our construction of universal one-way hash functions, where the result about 
average-case  
complexity is described in \cref{subsec:averagecase}. 
\emph{Universal one-way hash functions} (UOWHFs), as introduced by 
\citet{NaorYu89}, are a weaker form of 
\emph{collision-resistant hash functions};  a function family $\calF$ is
collision resistant     
if given a randomly chosen function
$f\in \calF$,
it is infeasible to find any pair of distinct inputs $x,x'$ such that $f(x)=f(x')$. UOWHFs only require
\emph{target collision resistance}, where the adversary must specify one
of the inputs $x$ before seeing the description of the function $f$.
We give a formal definition.

Let $(x,\state)\getsr \Ac(1^k)$

\begin{definition}\label{def:UOWHF}
	A family of functions
	$\calF_k = \set{ \Fc_z \colon \zo^{n(k)} \mapsto \zo^{m(k)}}_{z\in \zo^k}$, 
	for $n(k),m(k)\in\poly(k)$, is a family of 
	{\sf universal one-way hash functions (UOWHFs)}
	if it satisfies:

	\begin{description}
		\item[Efficiency:] given $z\in \zo^k$ and $x\in \zo^{n(k)}$, $\Fc_z(x)$ can be evaluated in time $\poly(k)$.
		\item[Shrinking:] $m(k)<n(k)$.
		
		\item[Target Collision Resistance:] the probability  that a \ppt (probabilistic polynomial-time) adversary $\Ac$  succeeds
		in the following game is negligible in $k$:
		
				\begin{enumerate}
					\item Let $(x,\state)\getsr \Ac(1^k)$.
					\item[-] Abort if $(x,\state) \notin \zo^{n(k)}\times\zo^\ast$.
					\item Let $z\getsr \zo^k$.
					\item Let $x' \getsr \Ac(\state,z)$.
					\item[-] Abort if $x'  \notin \zo^{n(k)}$.
					\item $\Ac$ \emph{succeeds} if $x\neq x'$ and $\Fc_z(x)=\Fc_z(x')$. \footnote{The $\leftarrow$ notation is explained in \cref{subsec:notation}.
						Namely, on security parameter $1^k$, algorithm $\Ac$ first samples an element $x$ in the function family input domain.  Then given (the description of)  a function  $\Fc_z$ uniformly drawn from the family, algorithm $\Ac$ has to find a collision with $x$: an element $x'\ne x$ that   $\Fc_z$ maps to the same output value. To avoid  discussing \emph{stateful} algorithms,   we do not allow  $\Ac$ to keep a ``state'' between the game stages. Rather, we  enable it to transfer information between the stages using the auxiliary string $\state$.}
			\end{enumerate}
	\end{description}~
\end{definition}

It turns out that this weaker security property suffices for many applications. The most immediate application given in \cite{NaorYu89} is 
\emph{secure fingerprinting}, whereby the pair $(f,f(x))$ can  
be  
taken as a compact ``fingerprint'' of a large file $x$, such that it is infeasible
for an adversary, seeing the fingerprint, to change the file $x$ to $x'$ without being detected.  More dramatically, \cite{NaorYu89} also showed that UOWHFs can be used to construct secure digital signature schemes, whereas all previous constructions (with proofs of security in the standard model) 
were based on trapdoor functions (as might have been expected to be
necessary due to the public-key nature of signature schemes).
More recently, UOWHFs have been used in the Cramer-Shoup encryption
scheme \cite{CramerSh03} and in the 
construction, from one-way functions,  of statistically hiding commitment schemes \cite{HaitnerNgOnReVa09,HaitnerReVaWe09}.


\citet{NaorYu89} gave a simple and elegant construction of UOWHFs 
from any one-way \emph{permutation}. 
\cite{SanYun90a} generalized the construction of \cite{NaorYu89} to get UOWHFs from \emph{regular} one-way 
functions.  
\citet{Rompel90} gave a 
more
involved construction to prove that UOWHFs can be constructed from an
arbitrary one-way function, thereby resolving the complexity
of UOWHFs (as one-way functions are the minimal complexity assumption
for complexity-based cryptography, and are easily implied by
UOWHFs); this remains the state of the art for arbitrary one-way functions.\footnote{More details of \cite{Rompel90}'s proof are worked out, with some corrections, in \cite{Rompel90-thesis,KatzKo05}.}
While complications may be expected for constructions from arbitrary one-way functions (due to their lack of structure), Rompel's analysis
also feels quite ad hoc. In contrast, the construction of pseudorandom generators from one-way functions
in    
\cite{HastadImLeLu99}, while
also somewhat complex, involves natural abstractions (\eg pseudoentropy) that allow for modularity and measure for what is being achieved at each stage of the construction.

In this paper, we give simpler constructions of UOWHFs from one-way functions,
based on (a variant of) the recently introduced notion of
\emph{inaccessible entropy} \citet{HaitnerReVaWe09}. In addition, one of the
constructions obtains slightly better efficiency and security than Rompel's
original construction.

\subsection{Inaccessible entropy}\label{sec:intr:IAE}
For describing our construction, it will be cleaner to work with a
variant of UOWHFs where there is a \emph{single} shrinking function
$\Fc : \zn\mapsto \zo^m$ (for each setting of the security parameter $k$)
such that it is infeasible to find collisions with \emph{random inputs}.
So in our model  
an adversary $\Ac$ is given a uniformly random $x\getsr \zn$, 
outputs an $x'$ such that $\Fc(x')=\Fc(x)$, and
succeeds\footnote{It   
	is easy to convert any such
	function $\Fc$ into a standard UOWHF family by defining $\Fc_z(x)=\Fc(z+x)$.}
if $x'\neq x$.  
Note that we can assume without loss of generality that $x'=\Ac(x)$ is 
always a preimage of $\Fc(x)$ ($\Ac$ has the option of 
returning   
$x$ in case it does not find a different preimage); we refer to an
algorithm $\Ac$ with this property as
an \emph{$\Fc$-collision finder}.

Our construction is based on an
information-theoretic   
view of UOWHFs. The fact that $\Fc$ is shrinking implies that
there are many preimages $x'$ of $F(x)$ available to $\Ac$. Indeed, if we consider an (inefficient) adversary $\Ac(x)$ that outputs a uniformly random preimage $x'\getsr \Fc^{-1}(\Fc(x))$ and let $X$ be a random variable uniformly distributed on $\zn$, then
$$\HSh(\Ac(X)\mid X) = \HSh(X\mid \Fc(X)) \geq n-m,$$ where $\HSh(\cdot\mid \cdot)$ denotes conditional Shannon entropy. (See \cref{sec:prelim} for more definitional details.)
We refer to the quantity $\HSh(X\mid \Fc(X))$ as the 
\emph{real entropy of $\Fc^{-1}$.}

On the other hand, 
\emph{target collision resistance} 
means that effectively only one of the preimages is accessible to
$\Ac$. That is for every probabilistic polynomial-time $\Fc$-collision finder $\Ac$, we have $\Pr[\Ac(X)\neq X]=\negl(n)$, which is equivalent to requiring that:
$$\HSh(\Ac(X)\mid X) = \negl(n)$$
for all probabilistic polynomial-time $\Fc$-collision finders $\Ac$.  (If $\Ac$ can find a collision $X'$ with non-negligible probability, then it can achieve non-negligible conditional entropy by returning $X'$ with probability 1/2 and returning $X$ with probability $1/2$.)
We refer to the maximum of $\HSh(\Ac(X)\mid X)$ over all efficient $\Fc$-collision finders as the \emph{accessible entropy of $\Fc^{-1}$.}
We emphasize that accessible entropy refers to an \emph{upper bound} on a form
of computational entropy, in contrast to the \citet{HastadImLeLu99} notion of
\emph{pseudoentropy}.

Thus, a natural weakening of the UOWHF property is to simply require a
noticeable gap between the real and
accessible entropies of $\Fc^{-1}$.
That is, for every probabilistic polynomial-time $\Fc$-collision 
finder $\Ac$, we have
$\HSh(\Ac(X)\mid X) < \HSh(X\mid \Fc(X))-\Delta$, for some 
noticeable $\Delta$, which we refer to as the
\emph{inaccessible entropy of $\Fc$}.

\subsection{Our UOWHF constructions}     
\label{sec:intr:constructions}
Our constructions of UOWHFs have two parts. First, we show how to obtain a function with noticeable inaccessible
entropy from any one-way function. Second, we show how to build a UOWHF from any function with inaccessible entropy.

\paragraph{OWFs $\implies$ inaccessible entropy.} Given a one-way function $f \colon \zn\mapsto \zo^m$, we show that a random truncation of $f$ has inaccessible entropy. Specifically, we define $\Fc(x,i)$ to be the first $i$ bits
of $f(x)$.

To see that this works, suppose for contradiction that $\Fc$ does not have noticeable inaccessible entropy. That is, we have an efficient adversary $\Ac$ that on input $(x,i)$ can sample from the set
$S(x,i) = \set{x' \colon f(x')_{1\ldots i} = f(x)_{1\ldots i}}$ with almost-maximal entropy, which is equivalent to sampling
according to a distribution that is statistically close to the uniform distribution on $S(x,i)$. We can now use $\Ac$
to construct an inverter $\Inv$ for $f$ that works as follows on input $y$: choose $x_0\getsr \zn$, and then for $i=1,\ldots,n$ generate a random
$x_i\getsr A(x_{i-1},i-1)$ subject to the constraint that $f(x_i)_{1,\cdots,i}=y_{1,\cdots,i}$. The latter step is feasible, since we are guaranteed that $f(x_i)_{1,\ldots,i-1} = y_{1,\cdots,i-1}$ by the fact that $\Ac$ is
an $\Fc$-collision finder, and the expected number of trials needed get agreement with $y_i$ is at most 2 (since $y_i\in \zo$, and $y$ and $f(x_i)$ are statistically close). It is not difficult to show that when run on a random output $Y$ of $f$, $\Inv$ produces an almost-uniform
preimage of $Y$. This contradicts the one-wayness of $f$. Indeed, we only need $f$ to be a \emph{distributional} one-way function \cite{ImpagliazzoLu89}, whereby it is infeasible to generate almost-uniform preimages under $f$.

\paragraph{Inaccessible entropy $\implies$ UOWHFs.}
Once we have a non-negligible amount of inaccessible entropy, we can construct a UOWHF via a series of standard
transformations.
\begin{enumerate}
	\item Repetition: By evaluating $\Fc$ on many inputs, we can increase the amount of inaccessible entropy from $1/\poly(n)$ to $\poly(n)$. Specifically, we take $\Fc^t(x_1,\ldots,x_t)=(\Fc(x_1),\ldots,\Fc(x_t))$ where
	$t=\poly(n)$. This transformation also has the useful effect of converting the real entropy of $\Fc^{-1}$ to  
	\emph{real min-entropy}: with very high probability  $x\getsr X^t$,  $\Fc^t(x)$ has large  number of pre-images.
	
	\item Hashing Inputs: By hashing the input to $\Fc$ (namely taking $\Fc'(x,g)=(\Fc(x),g(x))$ for a universal hash function $g$), we can reduce both the
	real (min-)entropy and the accessible entropy so that $(\Fc')^{-1}$ still has a significant amount of real entropy,
	but has (weak) target collision resistance (on random inputs). \label{itm:hash-input}
	
	\item Hashing Outputs: By hashing the output to $\Fc$ (namely taking $\Fc'(x,g)=g(\Fc(x))$), we can reduce the
	output length of $\Fc$ to obtain a shrinking function that still has (weak)
	target collision resistance.
\end{enumerate}

There are two technicalities that occur in the above steps. First, hashing the inputs only yields \emph{weak} target collision resistance; this is due to the fact that accessible Shannon entropy is an average-case measure and thus allows for the possibility that
the adversary can achieve high accessible entropy most of the time. Fortunately, this weak form of target collision resistance can be amplified
to full target collision resistance using another application of repetition and hashing (similar to
\cite{CRSTVW07}).

Second, the hashing steps require having a fairly accurate estimate of the real entropy. This can
be handled similarly to \cite{HastadImLeLu99,Rompel90}, by trying all (polynomially many) possibilities and concatenating
the resulting UOWHFs, at least one of which will be target collision resistant.

\paragraph{A more efficient construction.} We obtain a more efficient construction of UOWHFs by hashing
the output of the one-way function $f$ before truncating. That is, we define
$\Fc(x,g,i)=(g,g(f(x))_{1\cdots i})$. This function is in the spirit of the function that \citet{Rompel90} uses as a first step, but our function uses three-wise independent hash function instead of $n$-wise independent one, and enjoys  a
simpler  structure.\footnote{\cite{Rompel90} started with the function $f'(z,g_1,g_2) \eqdef (g_2(f_0(g_1(z))),g_1,g_2)$, where $g_1$ and $g_2$ are $n$-wise independent hash-functions, and $f_0$ is defined as $f_0(x,y,i) = (f(x) , y^{n-i} , 0^i)$.} 
Our analysis of this function is  simpler than
Rompel's 
and can be viewed as providing a clean abstraction of what it achieves (namely, inaccessible entropy) that makes the subsequent transformation to a UOWHF 
easier.

We obtain improved UOWHF parameters over our first construction for two reasons.  First, we obtain a larger amount of inaccessible
entropy: $(\log n)/n$ bits instead of roughly $1/n^4$ bits. Second, we obtain a bound on a stronger form of accessible entropy, which enables us to get full target collision resistance when we hash the inputs, avoiding the second
amplification step.

The resulting overall construction yields better parameters than
Rompel's original construction.  A one-way function
of input length $n$ yields a UOWHF with output length $\tO(n^7)$,
slightly 
improving Rompel's bound of $\tO(n^{8})$.
Additionally, we are able to reduce the key length needed: Rompel's
original construction uses a key of length $\tO(n^{12})$, whereas our
construction only needs a key of length $\tO(n^7)$.
If we allow the construction to utilize some nonuniform information (namely an estimate of the real entropy
of $\Fc^{-1}$), then we obtain output length $\tO(n^5)$, improving Rompel's bound of $\tO(n^6)$.  For the key length, the improvement in this case is from $\tO(n^7)$ to $\tO(n^5)$.
Of course, these bounds are still far from practical, but they illustrate the utility of inaccessible
entropy in reasoning about UOWHFs, which may prove useful in future constructions (whether based on one-way functions
or other building blocks).

\subsection{Connection to average-case complexity}\label{subsec:averagecase}

We use the notion of inaccessible entropy 
to reprove  
the following theorem by \citet{ImpLev90}, 
given in the realm of
\emph{average-case complexity}.  

\begin{theorem}[\cite{ImpLev90}, informal]\label{thm:intro:impagliazzolevin}
	Assume there exists an  \NP language  that is hard on some (efficiently) samplable distribution for  {\sf heuristics}:  every efficient algorithm fails to decide the language correctly on a noticeable part of the distribution. Then  there exists a language in \NP that is   hard against heuristics on the uniform distribution.
\end{theorem}

Our proof follows to a large extent the footstep of  \cite{ImpLev90}, where the  main novelty is formulating  the proof  in the language of inaccessible entropy, and rephrasing  it to make it resembles our  proof of UOWHFs from  one-way functions. This draws an interesting connection between two seemingly separate lines of research: 
average-case  
complexity and universal one-way hash-functions.

As in \cite{ImpLev90}, we prove \cref{thm:intro:impagliazzolevin} by proving its search variant: hardness to find a witness for a samplable distribution implies hardness to find a witness on the uniform distribution. Let  $(\RelR,\cD)$ be an $\NP$-search problem (\ie $\RelR$ is an \NP relation and $\cD$ is a samplable distribution) that is hard to solve heuristically, and let $D$ be the algorithm sampling instances according to $\cD$. 
For a  family of pair-wise independent hash functions  $\g$, consider the following \NP relation:
$$\RelR' = \set{(x' = (g,i),w' =(x,w)) \colon g\in \g, (D(x),w)\in \RelR \land g(D(x))_{1,\ldots,i} = 0^i}\footnote{To keep the discussion simple, we ignore input length consideration. See \cref{sec:averagecase} for the formal definitions.} $$ 
Namely, $\RelR'_\cL$ (\ie the language of $\RelR'$) consists of those random strings $x$ for $D$ such that $D(x)$ is in $\RelR_\cL$  and $D(x)$ is mapped to $0^i$ by the first $i$ bits of $g$.  While $\RelR'_\cL$ might not be hard on the uniform distribution (interpreted a the uniform distribution over  random pairs $(g,i)$), it is not hard to prove that the distribution is ``somewhat hard''. In particular, it happens noticeably often that $i$ is the ``right one for $(\RelR,\cD)$''; meaning that for a fixed element in $y$ which might be in the language, exactly one $x$ with $D(x) = y$ satisfies $g(x) = 0^i$. Conditioned on this event, letting $\Ac(\cdot)_x$ being the $x$ part in the witness output by $\Ac$, it is not hard to show that $\Hall(\Ac(G,I)_x)$ is noticeably less its information theoretic maximum: for any efficient algorithm $\Ac$, it holds that
$$\Hall(\Ac(G,I)_x) < \Hall(\cD|_{G(\cD)_{1,\ldots,I} = 0^I}),$$
where $(G,I)$ is the parsing of a random string into a pair $(g\in \g,i)$, and  assuming for simplicity that $\Ac$ never fails to provide a correct witness. Namely, the (accessible) entropy of $\Ac$ is smaller than the (real) entropy of $\cD$. Using  similar means   to the ones used to amplify the initial UOWHFs constructions described in  \cref{sec:intr:constructions} , the above gap can be  amplified to induce  hardness over the uniform distribution.

\subsection{Perspective}
The idea of inaccessible entropy was introduced in \cite{HaitnerReVaWe09} for the purpose of constructing
statistically hiding commitment schemes from one-way functions and from zero-knowledge proofs. There, the nature of statistically hiding commitments necessitated more involved notions of inaccessible entropy than we present here --- inaccessible entropy was defined in \cite{HaitnerReVaWe09} for interactive protocols, where one considers adversaries that try to generate next-messages or next-blocks of high entropy. (See \cite{HaitnerReVaWe20} for a simpler notion of inaccessible entropy that suffices for the one-way functions based commitment part.)


Here, we are able to work with a much simpler form of inaccessible entropy (significantly simpler also from the  notion considered in \cite{HaitnerReVaWe20}). The simplicity comes from the non-interactive nature of UOWHFs (and of solving \NP problems) so we only need to measure the entropy of a single string output by the adversary. Thus, the definitions here can serve as a gentler introduction to the concept of inaccessible entropy.

On the other hand, the many-round notions from \cite{HaitnerReVaWe09,HaitnerReVaWe20} allow for a useful ``entropy equalization'' transformation that avoids the need to try all possible guesses for the entropy. We do not know an analogous transformation for constructing UOWHFs. We also note that our simple construction of a function with inaccessible entropy by randomly truncating a one-way function (and its analysis) is inspired by the construction of an ``inaccessible entropy generator'' from a one-way function in \cite{HaitnerReVaWe09}.

Finally, with our constructions, the proof that one-way functions imply UOWHFs now parallels those of
pseudorandom generators \cite{HastadImLeLu99,HaitnerReVa09} and statistically hiding
commitments \cite{HaitnerNgOnReVa09,HaitnerReVaWe09}, with UOWHFs and statistically hiding commitments using
dual notions of entropy (high real entropy, low accessible entropy) to pseudorandom generators (low real entropy,
high pseudoentropy).

\subsection{Related work}
\paragraph{UOWHFs.}
\citet{KatzKo05}  gave a complete write-up of the  \citet{Rompel90,Rompel90-thesis}, with some corrections. Prior to
our paper,   
Rompel's result
represented 
the state of the art for UOWHFs from arbitrary one-way functions. Since the initial publication of this work in 2010 \cite{HaitnerHolReVaWe10}, there have
been several improvements in the setting of regular one-way functions. \citet{AmesRV2012} presented an even  more efficient construction of UOWHFs from (unknown) regular one-way functions. \citet{BarhumM2012} gave even a more efficient construction assuming the regularity of the one-way function is \emph{known}, where   \citet{YuLW2015} improved the result of  \cite{BarhumM2012}  presenting an almost optimal construction (\wrt the known black-box impossibility results)  of UOWHFs from a known regular one-way functions.

\paragraph{Average-case complexity.}
The notion of average-case  
complexity was first introduced by 
\citet{Levin86}. We focus on the result by \citet{ImpLev90} who show that if it is possible use a
polynomial time sampler to pick average-case  
problems which are hard,
then there is a different problem which is hard on average for the
uniform distribution. We give a different perspective on that proof,
and in particular highlight the connections to inaccessible entropy. 
A good overview of average-case  
complexity was given by \citet{BogTre08}.
\medskip

Recently,
\citet{HubacekNY2017} made a different, and very elegant,
connection between constructing   UOWHFs from OWFs and 
average-case  
hardness on the uniform distribution,  showing that a solution to the first challenge implies a solution
to     
the second one. Their approach is surprisingly simple: if OWFs exist, then  UOWHFs also exist,
which  can be seen as a  problem that is hard on the uniform distribution (given a UOWHF key and an input, find a colliding input). On the other hand, 
assuming  OWFs do  not exist,
\wlg the sampler of
a hard-on-the-average   
problem can be assumed  to output its random coins (indeed, its coins can be sampled from its original output under the assumption that OWFs do not exist). So in both cases, 
a hard-on-the-average 
problem implies 
a hard-on-the-average 
problem \wrt the uniform distribution. 

\subsection*{Organization of the paper}  
Formal definitions are given in \cref{sec:prelim}, where  the notion of inaccessible entropy used through the paper is defined in \cref{sec:tiegdefs}. In \cref{sec:OwfToIA} we show how to use any one-way function to get a function with  inaccessible entropy, where in \cref{sec:OWHFsfromIE} we use any function with inaccessible entropy to construct UOWHF. Finally, our result for average-case complexity is described in \cref{sec:averagecase}.

\section{Preliminaries}\label{sec:prelim}
Most of the material in this section is taken almost verbatim from \cite{HaitnerReVaWe09}, and missing proofs can be found in that paper.

\subsection{Notation}\label{subsec:notation}
All logarithms considered here are in base two.  
For $t\in \N$, we let $[t] = \set{1,\dots,t}$. 
A function $\mu \colon \N \rightarrow [0,1]$ is \emph{negligible}, 
denoted $\mu(n) = \negl(n)$, if $\mu(n)=n^{-\omega(1)}$. 
We let $\poly$ denote the set of all polynomials, and 
let \ppt stand for probabilistic polynomial time. 
Given a distribution $D$, we write $d\getsr D$ to indicate 
that $d$ is selected according to $D$. Similarly, given a 
finite set $\cS$, we write $s\getsr \cS$ to indicate that $s$ 
is selected according to the uniform distribution on $\cS$.


\subsection{Random variables}
Let $X$ and $Y$ be random variables taking values in a discrete universe $\Uni$.
We adopt the convention that when the same random variable appears multiple times in an expression, all occurrences refer to the same instantiation. For example, $\Pr[X=X]$ is 1. For an event $E$, we write $X|_E$ to denote the random variable $X$ conditioned on $E$. The \emph{support} of a random variable $X$ is $\Supp(X) \eqdef \set{ x \colon \Pr[X=x]> 0}$. $X$ is \emph{flat} if it is uniform on its support. For an event $E$, we write $I(E)$ for the corresponding 
indicator    
random variable, \ie $I(E)$ is $1$ when $E$ occurs and is $0$ otherwise.

We write $\vecc{X-Y}$ to denote the \emph{statistical difference} 
(also known as  
variation distance) between $X$ and $Y$, \ie
$$\vecc{X-Y} = \max_{T\subseteq \Uni} \size{\Pr[X\in T]-\Pr[Y\in T]}$$
We say that $X$ and $Y$ are \emph{$\eps$-close} if
$\vecc{X-Y}\leq \eps$ and \emph{$\eps$-far} otherwise.    

\subsection{Entropy measures}\label{sec:measures}
In this article we shall refer to several measures of entropy.  
The relation and motivation of these measures is best understood by considering a notion that we will refer to as the \emph{sample-entropy}: For a random variable $X$ and $x\in \Supp(X)$, we define the sample-entropy of $x$ with respect to $X$ to be the quantity $$\Hsam_X(x) \eqdef \log(1/\Pr[X=x]).$$ The sample-entropy measures the amount of ``randomness'' or ``surprise'' in the specific sample $x$, assuming that $x$ has been generated according to $X$. Using this notion, we can define the \emph{Shannon entropy} $\HSh(X)$ and \emph{min-entropy} $\Hmin(X)$ as follows:
\begin{eqnarray*}
	\HSh(X) &\eqdef & \Exp_{x\getsr X}[\Hsam_X(x)]\\
	\Hmin(X) &\eqdef & \min_{x\in \Supp(X)} \Hsam_X(x)
\end{eqnarray*}
We will also discuss the \emph{max-entropy} $\Hmax(X) \eqdef  \log(\size{\Supp(X)})$. The term ``max-entropy'' and its relation to the sample-entropy will be made apparent below.

It can be shown that $\Hmin(X) \leq \HSh(X) \leq \Hmax(X)$ with equality if and only if $X$ is flat. Thus, saying $\Hmin(X)\geq k$ is a strong way of saying that $X$ has ``high entropy'' and $\Hmax(X)\leq k$ a strong way of saying that $X$ as ``low entropy''.

\paragraph{Smoothed entropies.}
Shannon entropy is robust in that it is insensitive to small statistical differences. Specifically, if $X$ and $Y$ are $\eps$-close then $\size{\HSh(X)-\HSh(Y)}\leq \eps\cdot \log \size{\Uni}$. For example, if $\Uni=\zn$ and $\eps=\eps(n)$ is a negligible function of $n$ (\ie $\eps=n^{-\omega(1)}$), then the difference in Shannon entropies is vanishingly small (indeed, negligible). In contrast, min-entropy and max-entropy are brittle and can change dramatically with a small statistical difference. Thus, it is common to work with ``smoothed'' versions of these measures, whereby 
we consider a random variable $X$ to have \emph{high entropy}
if $X$ is $\eps$-close to some $X'$ with $\Hmin(X)\geq k$
and to have \emph{low entropy}
if $X$ is $\eps$-close to some $X'$ with $\Hmax(X)\leq k$,  
for some parameter $k$ and a negligible $\eps$.\protect\footnote{The term ``smoothed entropy'' was coined by  \protect\cite{RenWol04a}, but the notion of smoothed min-entropy has commonly been used (without a name) in the literature on randomness extractors \protect\cite{NisanZu96}.}

These smoothed versions of min-entropy and max-entropy can be captured quite
closely (and more concretely) by requiring that the sample-entropy
be   
large or small, resp.,  
with high probability:

\begin{lemma}
	\begin{enumerate}
		\item Suppose that with probability at least $1-\eps$ over $x\getsr X$, we have
		$\Hsam_X(x) \geq k$. Then $X$ is $\eps$-close to a random variable $X'$ such that
		$\Hmin(X')\geq k$.
		\item Suppose that $X$ is $\eps$-close to a random variable $X'$ such that
		$\Hmin(X')\geq k$. Then with probability at least $1-2\eps$ over $x\getsr X$, we have $\Hsam_X(x)\geq k-\log(1/\eps)$.
	\end{enumerate}
\end{lemma}

\begin{lemma} \label{lem:MaxEntvsSets}
	\begin{enumerate}
		\item Suppose that with probability at least $1-\eps$ over $x\getsr X$, we have
		$\Hsam_X(x) \leq k$. Then $X$ is $\eps$-close to a random variable $X'$ such that
		$\Hmax(X')\leq k$.
		\item Suppose that $X$ is $\eps$-close to a random variable $X'$ such that
		$\Hmax(X')\leq k$. Then with probability at least $1-2\eps$ over $x\getsr X$, we have $\Hsam_X(x)\leq k+\log(1/\eps)$.
	\end{enumerate}
\end{lemma}

Think of $\eps$ as inverse polynomial or a 
slightly negligible 
function in $n=\log(\size{\Uni})$. The above lemmas show that up to negligible statistical difference and a slightly super-logarithmic number of entropy bits,
the min-entropy and the max-entropy are captured by
a lower and an upper bound on sample-entropy, respectively.   

\paragraph{Conditional entropies.}
We will also be interested in conditional versions of entropy. For jointly distributed random variables $(X,Y)$ and $(x,y)\in \Supp(X,Y)$, we define
the \emph{conditional sample-entropy} to be $\Hsam_{X\mid Y}(x\mid y) = \log(1/\Pr[X=x\mid Y=y])$.
Then the standard \emph{conditional Shannon entropy} can be written as:
$$\HSh(X\mid Y) = \Exp_{(x,y)\getsr (X,Y)} \left[\Hsam_{X\mid Y}(x\mid y)\right]
= \Exp_{y\getsr Y} \left[\HSh(X|_{Y=y})\right] = \HSh(X,Y)-\HSh(Y).$$
There is no standard definition of conditional min-entropy and max-entropy, or even their smoothed versions. For us, it will be most convenient to generalize the sample-entropy characterizations of smoothed min-entropy and max-entropy given above. Specifically we will think of $X$ as
having ``high min-entropy'' and ``low max-entropy'' given $Y$ if
with probability at least $1-\eps$ over $(x,y)\getsr (X,Y)$, we have
$\Hsam_{X\mid Y}(x\mid y) \geq k$ and $\Hsam_{X\mid Y}(x\mid y)\leq k$,
resp. 

\paragraph{Flattening Shannon entropy.}  
The \emph{asymptotic equipartition property} in information theory states
that for 
a random variable $X^t = (X_1,\ldots,X_t)$, whose marginals $X_i$
are independent,  
with high probability,  
the sample-entropy $\Hsam_{X^t}(X_1,\ldots,X_t)$ is close to its expectation.
In \cite{HastadImLeLu99} a quantitative bound on this was shown
by reducing it to the Hoeffding
bound.  (One
cannot directly apply the Hoeffding bound, because $\Hsam_{X}(X)$ does not
have an upper bound, but one can define a related random variable
which
does.)    
We use a different bound here, which was proven in \cite{HolRen11}.
The bound has the advantage that it is somewhat easier to state,
even though the proof is longer.  We remark that the bound
from \cite{HastadImLeLu99} would be sufficient for our 
purposes.   


\begin{lemma} \label{lem:flattening}~
	\begin{enumerate}
		\item Let $X$ be a random variable taking values in a universe $\Uni$, let $t\in \N$, and let $\eps>2^{-t}$. Then with probability at least $1-\eps$ over $x\getsr X^t$,
		$$\size{\Hsam_{X^t}(x) - t\cdot \HSh(X)} \leq O(\sqrt{t\cdot \log(1/\eps)} \cdot \log (|\Uni|))$$
		
		\item Let $X,Y$ be jointly distributed random variables where $X$ takes values in a universe $\Uni$, let $t\in \N$, and let $\eps>2^{-t}$. Then with probability at least $1-\eps$ over $(x,y)\getsr (X^t,Y^t)\eqdef (X,Y)^t$,
		$$\size{\Hsam_{X^t\mid Y^t}(x\mid y) - t\cdot \HSh(X\mid Y)} \leq O(\sqrt{t\cdot \log(1/\eps)} \cdot \log (\size{\Uni}))$$
	\end{enumerate}
\end{lemma}
The statement follows directly from \cite[Thm 2]{HolRen11}.

\subsection{Hashing} \label{sec:hashing}
A family of functions $\Fc = \set{ f \colon \zn \mapsto \zo^m}$ is \emph{2-universal} if
for every $x\neq x'\in \zn$, when we choose $f\getsr \Fc$, we have $\Pr[f(x)=f(x')] \leq 1/\size{\zo^m}$.  $\Fc$ is \emph{$t$-wise independent} if for all distinct $x_1,\ldots,x_t\in \zn$, when we choose $f\getsr \Fc$, the random variables $f(x_1),\ldots,f(x_t)$ are independent and each
of them is  
uniformly distributed over $\zo^m$.

$\Fc$ is \emph{explicit} if given the description of a function $f\in \Fc$ and $x\in \zn$,
the value  
$f(x)$ can be computed in time $\poly(n,m)$.  $\Fc$ is \emph{constructible} if it is explicit and there is a probabilistic polynomial-time algorithm that given $x\in \zn$, and $y\in \zo^m$, outputs a random $f\getsr \Fc$ such that $f(x)=y$.

It is well-known that there are constructible families of 
$t$-wise independent functions in which choosing a function $f\getsr \Fc$ uses only $t\cdot \max \set{n,m}$ random bits.

Most of the material in this section is taken almost verbatim from \cite{HaitnerReVaWe09}, and missing proofs can be found in that paper.


A family of functions $\Fc = \set{ f \colon \zn \mapsto \zo^m}$ is \emph{2-universal} if
for every $x\neq x'\in \zn$, when we choose $f\getsr \Fc$, we have $\Pr[f(x)=f(x')] \leq 1/\size{\zo^m}$.  $\Fc$ is \emph{$t$-wise independent} if for all distinct $x_1,\ldots,x_t\in \zn$, when we choose $f\getsr \Fc$, the random variables $f(x_1),\ldots,f(x_t)$ are independent and each
of them is  
uniformly distributed over $\zo^m$.

$\Fc$ is \emph{explicit} if given the description of a function $f\in \Fc$ and $x\in \zn$,
the value  
$f(x)$ can be computed in time $\poly(n,m)$.  $\Fc$ is \emph{constructible} if it is explicit and there is a probabilistic polynomial-time algorithm that given $x\in \zn$, and $y\in \zo^m$, outputs a random $f\getsr \Fc$ such that $f(x)=y$.

It is well-known that there are constructible families of 
$t$-wise independent functions in which choosing a function $f\getsr \Fc$ uses only $t\cdot \max \set{n,m}$ random bits.

\section{Inaccessible entropy for inversion problems}\label{sec:tiegdefs}
In this section we define, following the infomercial description given in the introduction,  the real and accessible entropy of the inverse of a function. The inaccessible entropy of the inverse is define as the gap between the two.

\subsection{Real entropy}
For a function $\Fc$, we define the \emph{real entropy} of $\Fc^{-1}$
to be the amount of entropy left in the input after revealing the output. We measure the above entropy using  Shanon entropy (average case),  min-entropy  and  max-entropy.  

\begin{definition}[real entropy]\label{def:RealEntropy}
	Let $n$ be a security parameter, and $\Fc \colon \zn\mapsto \zo^m$ a function.
	We say that $\Fc^{-1}$ has {\sf real Shannon entropy} $k$ if
	$$\HSh(X\mid \Fc(X))= k,$$
	where $X$ is uniformly distributed on $\zn$.
	We say that $\Fc^{-1}$ has {\sf real min-entropy} at least $k$ if there is a negligible function $\eps=\eps(n)$ such that
	$$\Pr_{x\getsr X}\nolimits \left[\Hsam_{X\mid \Fc(X)}(x\mid \Fc(x))\geq k\right] \geq 1-\eps(n).$$
	We say that $\Fc^{-1}$ has 
	{\sf real max-entropy} 
	at most $k$ if there is a negligible function $\eps=\eps(n)$ such that
	$$\Pr_{x\getsr X}\nolimits \left[\Hsam_{X\mid \Fc(X)}(x\mid \Fc(x))\leq k\right] \geq 1-\eps(n).$$
\end{definition}
It is easy to verify that, ignoring negligible terms, the min-entropy of $F^{-1}$ is at most its Shannon-entropy,  which in turn is at most its max-entropy, where equality holds only if  $F$ is regular. We also note that more concrete formulas for the entropies above are:

\begin{align*}
	\Hsam_{X\mid \Fc(X)}(x\mid \Fc(x)) &= \log \size{\Fc^{-1}(\Fc(x))}\\
	\HSh(X\mid \Fc(X)) &= \Exp\left[\log \size{\Fc^{-1}(\Fc(X))}\right].
\end{align*}

As our goal is to construct UOWHFs that are shrinking, achieving 
high real entropy is a natural intermediate step. Indeed, the amount
by which $\Fc$ shrinks is a lower bound on the real entropy of $\Fc^{-1}$:

\begin{proposition}
	If $\Fc \colon \zn\mapsto\zo^m$, then the real Shannon entropy of $\Fc^{-1}$ is at least $n-m$, and the real min-entropy of $\Fc^{-1}$ is at least $n-m-s$ for any $s=\omega(\log n)$.
\end{proposition}

\begin{proof}
	For Shannon entropy, we have
	$$\HSh(X\mid \Fc(X)) \geq \HSh(X)-\HSh(\Fc(X)) \geq n-m.$$
	
	For min-entropy, let $S = \set{ y\in \zo^m \colon \Pr[f(X)=y] < 2^{-m-s}}$. Then $\Pr[f(X)\in S] \leq 2^m \cdot 2^{-m-s} = \negl(n)$, and for every $x$ such that $f(x)\notin S$, we
	have
	\begin{align*}
		\Hsam_{X\mid \Fc(X)}(x\mid \Fc(x)) &= \log\frac{1}{\Pr[X=x\mid F(X)=f(x)]}
		= \log\frac{\Pr[f(X)=f(x)]}{\Pr[X=x]}
		\geq \log\frac{2^{-m-s}}{2^{-n}}
		= n-m-s.
	\end{align*}
\end{proof}

\subsection{Accessible entropy}
We  define accessible entropy of $F^{-1}$ using the notion of ``collision-finding'' algorithm,  an algorithm that aims to find a second-pre-image of $F(X)$ with ``maximal entropy''.  The accessible entropy of $F$ will be defined as the entropy of the best \emph{efficient} collision-finding algorithm.


\begin{definition}[collision finding algorithm]\label{def:CollisionFinder}
	For a function $\Fc \colon \zn\mapsto \zo^m$, an {\sf $\Fc$-collision-finder} is a randomized algorithm $\Ac$ such that for every $x\in \zn$
	and coin tosses $r$ for $\Ac$, we have
	$\Ac(x;r)\in \Fc^{-1}(\Fc(x))$.
\end{definition}

Note that $\Ac$ is required to \emph{always} produce an input $x'\in \zn$ such that
$\Fc(x)=\Fc(x')$. This is a reasonable constraint because $\Ac$ has the option of outputting $x'=x$ if it does not find a true collision. We consider $\Ac$'s goal to be maximizing the entropy of its output $x'=A(x)$, given a random input $x$.

It is easy to see that If we let $\Ac$
be \emph{computationally unbounded}, then the optimum turns out to equal exactly the real entropy:

\begin{proposition}\label{prop:optCF}
	Let $\Fc \colon \zn\mapsto \zo^m$. Then the real Shannon entropy of $\Fc^{-1}$
	equals the maximum of $\HSh(\Ac(X;R)\mid X)$ over all (computationally unbounded) $\Fc$-collision finders $\Ac$, where the random variable $X$ is uniformly distributed in $\zn$ and $R$ is uniformly random coin tosses for $\Ac$. That is,
	$$\HSh(X\mid \Fc(X)) = \max_{\Ac} \HSh(\Ac(X;R)\mid X),$$
	where the maximum is taken over {\sf all} $\Fc$-collision finders $\Ac$.
\end{proposition}

\begin{proof}
	The ``optimal'' $\Fc$-collision finder $\Ac$ that maximizes $\HSh(\Ac(X)\mid X)$ is the algorithm $\At$ that, on input $x$, outputs a uniformly random element of $f^{-1}(f(x))$. Then $$\HSh(\At(X;R)\mid X) = \Exp[\log \vecc{f^{-1}(f(X))}] = \HSh(X\mid \Fc(X)).$$
\end{proof}

The notion of \emph{accessible entropy} simply restricts the above to \ppt  algorithms. We consider both Shanon and max-entropy variants (since we aim to upper bound the accessible entropy, we care not about the min-entropy variant).

\begin{definition}[accessible entropy]\label{def:AccH}
	Let $n$ be a security parameter and $\Fc \colon \zn\mapsto \zo^m$ a function.
	We say that $\Fc^{-1}$ has {\sf accessible Shannon entropy} at most $k$ if for every \ppt $\Fc$-collision-finder $\Ac$,
	we have
	\begin{align*}\HSh(\Ac(X;R)\mid X)\leq k
	\end{align*}
	for all sufficiently large $n$,
	where the random variable $X$ is uniformly distributed on $\zn$ and $R$ is uniformly random coin tosses for $\Ac$.

	We say that $\Fc^{-1}$ has {\sf $p$-accessible max-entropy} at most $k$ if
	for every \ppt $\Fc$-collision-finder $\Ac$, there exists a family of
	sets $\set{\cL(x)}_{x \in \Supp(X)}$ each of size at most $2^k$, such that
	$x \in \cL(x)$ for all $x \in \Supp(X)$, and
	$$\Pr\left[\Ac(X;R) \in \cL(X)\right] \geq 1-p$$
	for all sufficiently large $n$,
	where the random variable $X$ is uniformly distributed on $\zn$ and $R$ is uniformly random coin tosses for $\Ac$.
	In addition, if $p = \eps(n)$ for some negligible function $\eps(\cdot)$, then we simply say that
	$\Fc^{-1}$ has {\sf accessible max-entropy} at most $k$.
\end{definition}
It is easy to verify that, ignoring negligible terms,  the accessible Shannon entropy of $F^{-1}$ is at most its accessible max-entropy, \ie if the  accessible max-entropy of $F^{-1}$ is at most $k$, then its accessible Shannon entropy is at most $k$.  (We will later, \cref{sec:AccAvgMaxEnt}, introduce an in-between variant of accessible entropy; larger than Shanon smaller than max)


\remove{
	\begin{definition}\label{def:AccMaxEnt}
		Let $n$ be a security parameter and $\Fc \colon \zn\mapsto \zo^m$ a function. For $p = p(n) \in [0,1]$,
		we say that $\Fc^{-1}$ has \emph{$p$-accessible max-entropy} at most $k$ if
		for every \ppt $\Fc$-collision-finder $\Ac$, there exists a family of
		sets $\set{\cL(x)}_{x \in \Supp(X)}$ each of size at most $2^k$ such that
		$x \in \cL(x)$ for all $x \in \Supp(X)$ and
		$$\Pr\left[\Ac(X;R) \in \cL(X)\right] \geq 1-p$$
		for all sufficiently large $n$,
		where the random variable $X$ is uniformly distributed on $\zn$ and $R$ is uniformly random coin tosses for $\Ac$.
		In addition, if $p = \eps(n)$ for some negligible function $\eps(\cdot)$, then we simply say that
		$\Fc^{-1}$ has \emph{accessible max-entropy} at most $k$.
	\end{definition}
}

The reason that having an upper bound on accessible entropy is useful as an intermediate step towards constructing UOWHFs, is that
accessible max-entropy 0 is equivalent to target collision resistance (on random inputs):

\begin{definition}[$q$-collision-resistant]\label{def:collision-resistant}
	Let $\Fc \colon \zn\mapsto \zo^m$ be a function and $q = q(n) \in [0,1]$.
	We say that $\Fc$ is {\sf
		$q$-collision-resistant on random inputs} if for every
	\ppt $\Fc$-collision-finder $\Ac$,
	$$\Pr[\Ac(X;R) = X] \geq q,$$
	for all sufficiently large $n$,
	where the random variable $X$ is uniformly distributed on $\zn$ and $R$ is uniformly random coin tosses for $\Ac$.
	In addition, if $q = 1-\eps(n)$ for
	some negligible function $\eps(\cdot)$, we say that $\Fc$ is
	collision-resistant on random inputs.
\end{definition}

\begin{lemma}\label{lem:0accmax2CR}
	Let $n$ be a security parameter and $\Fc \colon \zn \mapsto \zo^m$
	be a function. Then, for any $p = p(n) \in (0,1)$, the following
	statements are equivalent:
	\begin{enumerate}
		\item[\rm{(1)}] $\Fc^{-1}$ has $p$-accessible max-entropy $0$.
		\item[\rm{(2)}] $\Fc$ is $(1-p)$-collision-resistant on random inputs.
	\end{enumerate}
	In particular, $\Fc^{-1}$ has accessible max-entropy $0$ iff
	$\Fc$ is collision-resistant on random inputs.
\end{lemma}

\begin{proof}
	Note that (1) implies (2) follows readily from the definition. To see
	that (2) implies (1), simply take $\cL(x) = \set{x}$.
\end{proof}

While bounding $p$-accessible max-entropy with negligible $p$ is our ultimate goal,
one of our constructions will work by first giving a bound on accessible Shannon entropy, and then deducing a bound on $p$-accessible max-entropy for a value of $p<1$ using the following lemma:
\begin{lemma}\label{lem:AccHtoAccMax}
	Let $n$ be a security parameter and $\Fc \colon \zn \mapsto \zo^m$
	be a function. If $\Fc^{-1}$ has accessible Shannon entropy at most $k$,
	then $\Fc^{-1}$ has $p$-accessible max-entropy at most $k/p+O(2^{-k/p})$ for any
	$p = p(n) \in (0,1)$.
\end{lemma}
\begin{proof}
	Fix any \ppt $\Fc$-collision-finder $\Ac$. From the bound on accessible Shannon entropy,
	we have that $\HSh(\Ac(X;R)\mid X) \leq k$. Applying Markov's inequality, we have
	\begin{align*}
		\Pr_{x\getsr X,r\getsr R} \left[ \Hsam_{\Ac(X;R)\mid X}(\Ac(x;r)\mid x)\leq k/p \right] \geq 1-p
	\end{align*}
	Take $\cL(x)$ to be the set:
	$$\cL(x) = \set{ x } \cup \set{x' \colon \Hsam_{\Ac(X;R)\mid X}(x'\mid x)\leq k/p
	}$$ We may rewrite $\cL(x)$ as $\set{ x } \cup \set{ x'
		\colon\Pr_{r}[\Ac(x;r) = x'] \geq 2^{-k/p}}$. It is easy to see that $\size{\cL(x)}
	\leq 2^{k/p}+1$ and thus $\Fc^{-1}$ has $p$-accessible max-entropy at
	most $k/p+O(2^{-k/p})$.
\end{proof}
Once we have a bound on $p$-accessible max-entropy for some $p<1$, we need to apply several transformations to obtain a function with a good bound on $\negl(n)$-accessible max-entropy.

\subsubsection{Accessible average max-entropy}\label{sec:AccAvgMaxEnt}
Our second construction (which achieves better parameters), starts with a bound
on a different average-case form of accessible entropy, which is stronger
than bounding the accessible Shannon entropy. The benefit of this notion it that it can be converted
more efficiently to $\negl(n)$-accessible max-entropy, by simply taking repetitions.

To motivate the definition, recall that
a bound on accessible Shannon entropy means that the sample entropy $\Hsam_{\Ac(X;R)\mid X}(x'\mid x)$ is small on average over $x\getsr X$ and
$x'\getsr \Ac(x;R)$. This sample entropy may depend on both the input $x$ and the $x'$
output by the adversary (which in turn may depend on its coin tosses). A stronger
requirement is to say that we have upper bounds $k(x)$ on the sample entropy that depend \emph{only on $x$}. The following definition captures this idea, thinking of
$k(x)=\log\size{\cL(x)}$. (We work with sets rather than sample entropy to avoid paying the $\log(1/\eps)$ loss in \cref{lem:MaxEntvsSets}.)

\begin{definition}[accessible average max-entropy]\label{def:AccAvgMaxEnt}
	Let $n$ be a security parameter and $\Fc \colon \zn\mapsto \zo^m$ a function.
	We say that $\Fc^{-1}$ has {\sf accessible average max-entropy} at
	most $k$ if for every \ppt $\Fc$-collision-finder $\Ac$, there exists
	a family of sets $\set{\cL(x)}_{x \in \Supp(X)}$ and a negligible function
	$\eps=\eps(n)$ such that $x \in \cL(x)$ for all $x \in \Supp(X)$,
	$\Exp[\log \size{\cL(X)}] \leq k$ and
	$$\Pr\left[\Ac(X;R) \in \cL(X)\right] \geq 1-\eps(n),$$
	for all sufficiently large $n$,
	where the random variable $X$ is uniformly distributed on $\zn$ and $R$ is uniformly random coin tosses for $\Ac$.
\end{definition}
It is easy to verify that, ignoring negligible terms, the accessible average max-entropy of $F^{-1}$ is at least its accessible Shannon entropy and at most its  accessible max-entropy.

\section{Inaccessible entropy from one-way functions}\label{sec:OwfToIA}
We present two constructions of inaccessible entropy functions from one-way functions. The one in  \cref{sec:OwfToIASahannon} is extremely simple and merely trims the one-way function output. The one in \cref{sec:OwfToIAERomp} is somewhat more complicated (in the spirit of the first step of \citet{Rompel90}, thought still significantly simpler) that yields a more efficient overall construction.

\subsection{A direct construction}\label{sec:OwfToIASahannon}

The goal of this section is to prove the following theorem:
\begin{theorem}\label{thm:OwfToIASahannon}
	Let $f: \set{0,1}^n \mapsto \set{0,1}^n$ be a one-way function and
	define $\Fc$ over $\set{0,1}^n \times [n]$ as $\Fc(x,i) =
	f(x)_{1,\ldots,i-1}$.  Then, $\Fc^{-1}$ has accessible Shannon entropy
	at most $\HSh(Z\mid \Fc(Z)) - \frac{1}{64 n^2}$, where $Z=(X,I)$ is uniformly
	distributed over $\set{0,1}^n \times [n]$.
\end{theorem}
We do not know whether the function $\Fc^{-1}$ has even less accessible
Shannon entropy, (say, with a gap of $\Omega(\frac{1}{n})$).
However, it seems that a significantly stronger bound would require much more
effort, and even improving the bound to $\Omega(\frac{1}{n})$
does not seem to yield an overall construction which is as efficient
as the one resulting from \cref{sec:OwfToIAERomp}.  Therefore
we aim to present a proof which is as simple as possible.

We begin with a high-level overview of our approach. Recall from
\cref{prop:optCF} the ``optimal'' $\Fc$-collision-finder $\At$
that computes $\Fc^{-1}(\Fc(\cdot))$. The proof basically
proceeds in three steps:
\begin{enumerate}
	\item First, we show that it is easy to invert $f$ using $\At$ (\cref{lem:nOfRepetitions}).
	\item Next, we show that if a $\Fc$-collision-finder $\Ac$ has high
	accessible Shannon entropy, then it must behave very similarly
	to $\At$ (\cref{lem:expectationOfEpsilon}).
	\item Finally, we show that if $\Ac$ behaves very similarly to $\At$,
	then it is also easy to invert $f$ using
	$\Ac$ (\cref{lem:boundingTheEvents}).
\end{enumerate}
We may then deduce that if $f$ is one-way, any $\Fc$-collision-finder $\Ac$
must have accessible Shannon entropy bounded away from $\HSh(Z\mid \Fc(Z))$.

\paragraph{Step 1.} Suppose we have an optimal collision finder $\At(x,i;r)$ that outputs a uniform random element from
$\Fc^{-1}(\Fc(x,i))$.
In order to invert an element $y$, we repeat the following process:
start with an arbitrary element $x^{(0)}$ and use $\At$ to
find an element $x^{(1)}$ such that $f(x^{(1)})$ has the same
first bit as $y$. In the $i$'th step find $x^{(i)}$ such that
the first $i$ bits of $f(x^{(i)})$ equal $y_{1,\ldots,i}$ (until $i=n$).

This is done more formally in the following algorithm for
an arbitrary oracle $\CF$ which we set to $\At$ in the
first lemma we prove.
The algorithm $\ExtendOnex$ does a single step.
Besides the new symbol $x'$ which we are interested in,
$\ExtendOnex$ also returns the number of calls which it did to the oracle.
This is completely uninteresting to the overall algorithm,
but we use it later in the analysis when we bound the number of
oracle queries made by $\ExtendOnex$.

\bigskip
\noindent\framebox{
	\begin{minipage}{16cm}
		\noindent \textbf{Algorithm $\ExtendOnex$}
		
		\medskip\hrule\medskip
		\textbf{Oracle:} An $\Fc$-collision finder $\CF$.\\
		\textbf{Input:} $x \in \set{0,1}^n$, $b \in \set{0,1}$, $i \in [n]$.
		
		\medskip\hrule\medskip
		
		$j \eqdef 0$\\
		\textbf{repeat}\\
		\mbox\qquad $x' \eqdef \CF(x,i)$\\
		\mbox\qquad $j := j+1$\\
		\textbf{until} $f(x')_i = b$\\
		\textbf{return} $(x',j)$
	\end{minipage}
}
\bigskip

\bigskip
\noindent\framebox{
	\begin{minipage}{16cm}
		\noindent \textbf{Inverter $\Inv$}
		\medskip
		
		\hrule
		
		\medskip
		
		\textbf{Oracle:} An $\Fc$-collision finder $\CF$.\\
		\textbf{Input:} $y \in \set{0,1}^n$
		
		\medskip\hrule\medskip
		
		$x^{(0)} \getsr U_n$\\
		\textbf{for} $i=1$ \textbf{to} $n$ \textbf{do}:\\
		\mbox\qquad $(x^{(i)},j) \eqdef \ExtendOnex^\CF(x^{(i-1)},y_i,i)$\\
		\textbf{done}\\
		\textbf{return} $x^{(n)}$
	\end{minipage}
}
\bigskip

We first show that with our optimal collision finder $\At$, the inverter
inverts with only $2n$ calls in expectation (even though it can
happen that it runs forever).
Towards proving that, we define $p(b\mid y_{1,\ldots,i-1})$ as the probability that the $i$'th bit of $f(x)$
equals $b$, conditioned on the event that $f(x)_{1,\ldots,i-1} =
y_{1,\ldots,i-1}$ (or $0$ if $f(x)_{1,\ldots,i-1} = y_{1,\ldots,i-1}$ is
impossible).

\begin{lemma}\label{lem:nOfRepetitions}
	The expected number of calls to $\At$ in a random execution of $\Inv^{\At}(y=f(x))$ with $x\getsr \zn$, is at most $2n$.
\end{lemma}
\begin{proof}
	Fix some string $y_{1,\ldots,i-1}$ in the image of $\Fc$.  We want
	to study the expected number of calls to $\At(x^{(i-1)},i)$ in case
	$F(x^{(i-1)},i) = y_{1,\ldots,i-1}$.
	
	If we would know $y_{i}$, then this expected number of calls would
	be $\frac{1}{p(y_i\mid y_{1,\ldots,i-1})}$.  Since $y_i = 0$ with
	probability $p(0\mid y_{1,\ldots,i-1})$ we get that the expected number
	of calls is $1$ if either of the probabilities is $0$ and
	$p(0\mid y_{1,\ldots,i-1}) \cdot \frac{1}{p(0\mid y_{1,\ldots,i-1})} +
	p(1\mid y_{1,\ldots,i-1}) \cdot \frac{1}{1\mid p(y_{1,\ldots,i-1})} = 2$
	otherwise.  Using linearity of expectation we get the result.
\end{proof}

\paragraph{Step 2.} Given an
$\Fc$-collision finder $\Ac$,  we define $\epsilon(x,i)$ to be the
statistical distance of the distribution of $\Ac(x,i;r)$ and the
the output distribution of $\At(x,i;r)$ (which equals the
uniform distribution over $\Fc^{-1}(\Fc(x,i))$).

We want to show that if $\Ac$ has high accessible Shannon entropy,
then $\Ac$ behaves very similarly to $\At$.
The next lemma formalizes this by stating that $\eps(x,i)$ is small
on average (over the uniform random choice of $x \in \zo^n$
and $i \in [n]$).
The lemma follows
by applying Jensen's inequality on the well known relationship between
entropy gap and statistical distance.
\begin{lemma}\label{lem:expectationOfEpsilon}
	Assume $\HSh(\Ac(Z)) \geq \HSh(Z\mid \Fc(Z)) - \frac{1}{64 n^2}$, then
	$\Exp_{i\getsr [n],x\getsr \zn}[\eps(x,i)] \leq \frac{1}{8n}$.
\end{lemma}
\begin{proof}
	\begin{align*}
		\vecc{(Z,\At(Z)) - (Z,\Ac(Z))}
		&=
		\Exp_{z\getsr Z}[\size{(z,\At(z)) - (z,\Ac(z))}]\\
		&\leq
		\Exp_{z\getsr Z}[\sqrt{\HSh(\At(z)) - \HSh(\Ac(z))}]\\
		&\leq
		\sqrt{\Exp_{z\getsr Z}[\HSh(\At(z)) - \HSh(\Ac(z))]}\\
		&\leq \frac{1}{8n}.
	\end{align*}
	The first inequality uses the fact that if $W$ is a random
	variable whose support is contained in a set $S$ and $U$ is the
	uniform distribution on $S$, then $\vecc{U-W}\leq
	\sqrt{\HSh(U)-\HSh(W)}$ (see \cite[Lemma 11.6.1]{CoverTh:InfoTheory}).
	The second inequality follows by Jensen's inequality.
	The final inequality uses $\HSh(\At(Z)) = \HSh(Z \mid \Fc(Z))$
	(\cref{prop:optCF}).
\end{proof}

\paragraph{Step 3.} We have seen now that $\Inv^{\At}$ inverts $f$ with $2n$
calls in expectation and that $\Ac$ behaves similarly to $\At$.  We
now want to show that $\Inv^\Ac$ also inverts $f$ efficiently. The
main technical difficulty is that even though
$\Inv^{\At}$ makes $2n$ calls to $\At$ in expectation and $\Ac$
and $\At$ are close in statistical distance, we cannot immediately deduce
an upper bound on the number of calls $\Inv^{\Ac}$ makes to $\Ac$.
Indeed, our analysis below exploits the fact that $\Inv$ and $\ExtendOnex$
have a fairly specific structure.

We will assume \wlg that
$$\Pr_{R}[\Ac(x,i;R) = \At(x,i;R)] = 1-\epsilon(x,i),$$ where $\At$
is an optimal collision finder as above.  This follows from
a standard coupling argument
since we do not require $\At$ to be polynomial time, and also
because we can extend the number of random bits~$\Ac$ uses (we
assume it just ignores unused ones).  To do this, $\At$ first computes
the statistics of $\Ac$ on input $(x,i)$, and also the result of
$\Ac(x,i;r)$.  He checks whether $\Ac(x,i;r)$ is one of the elements which
occur too often, and outputs a different, carefully chosen one, with
appropriate probability if this is the case.

We now show that in most executions of $\ExtendOnex$ it does not
matter whether we use~$\Ac$ or~$\At$ (that is, $\ExtendOnex$
makes the same number of oracle queries to $\Ac$ and $\At$,
and outputs the same value).
\begin{lemma}\label{lem:boundingTheEvents}
	For any $(x,i) \in \zn \times [n]$,
	and any $y_i \in \zo$, we have
	\begin{align}
		\Pr_{R}[\ExtendOnex^\Ac(x,y_i,i;R) = \ExtendOnex^{\At}(x,y_i,i;R)] \geq
		1-\frac{2\eps(x,i)}{p(y_i\mid y_{1,\ldots,i-1})}
	\end{align}
\end{lemma}
Note that the oracle algorithm $\ExtendOnex$ is deterministic,
and in the above expressions, $R$ refers to the coin tosses used
by the oracles that $\ExtendOnex$ queries, namely $\Ac$ and $\At$.
We stress that the lemma says that both the value $x'$ \emph{and} the
number $j$ returned are equal with high probability.

\begin{proof}
	Let $J=J(R)$ be the second coordinate of the output of
	$\ExtendOnex^{\Ac}(x,y_i,i;R)$ (\ie the counter) and $\Jt=\Jt(R)$ the analogous output of $\ExtendOnex^{\At}(x,y_i,i;R)$.
	We write
	\begin{align}
		\lefteqn{\Pr_{R}[\ExtendOnex^\Ac(x,y_i,i;R) = \ExtendOnex^{\At}(x,y_i,i;R)]=}\nonumber
		\\&
		\sum_{j \geq 1}
		\Pr_{R}[\min(J,\Jt) = j]\cdot
		\Pr_{R}[\ExtendOnex^\Ac(x,y_i,i;R) = \ExtendOnex^{\At}(x,y_i,i;R) |
		\min(J,\Jt)=j]\nonumber
		\\&
		=\Ex_{j\getsr P_{J}}\bigl[
		\Pr_{R}[\ExtendOnex^\Ac(x,y_i,i;R) = \ExtendOnex^{\At}(x,y_i,i;R) |
		\min(J,\Jt)=j]\bigr]\label{eq:extendOne}
	\end{align}
	where $P_J$ is some distribution over the integers which, as it turns out,
	we do not need to know.
	
	Let now $R'$ be the randomness used by $\Ac$ or
	$\At$ in round $j$.
	Then,
	\begin{align*}
		\lefteqn{\Pr_{R}[\ExtendOnex^\Ac(x,y_i,i;R) = \ExtendOnex^{\At}(x,y_i,i;R) |
			\min(J,\Jt)=j]}\\
		&=
		\Pr_{R'} \bigl [\Ac(x,i;R') = \At(x,i;R') \mid f(\Ac(x,i;R'))_i = y_i \lor f(\At(x,i;R'))_i
		= y_i \bigr],
	\end{align*}
	because each iteration of $\ExtendOnex$ uses fresh independent randomness.
	
	Let $P$ be the distribution over $\zo^n$ produced by $\Ac(x,i;R')$,
	and $P^*$ be the (uniform) distribution produced by $\At(x,i;R')$.
	For $p =
	p(y_i \mid y_{1,\ldots,i-1})$ and $\eps = \eps(x,i)$, it holds that 
	\begin{align*}
		\MoveEqLeft{\ppr{R'} {\Ac(x,i;R') = \At(x,i;R') \mid f(\Ac(x,i;R'))_i = y_i \lor f(\At(x,i;R'))_i
				= y_i }}\\
		&= \frac{\pr {\Ac(x,i;R') = \At(x,i;R') \land  (f(\Ac(x,i;R'))_i = y_i \lor f(\At(x,i;R'))_i
				= y_i )}}{\pr {  f(\Ac(x,i;R'))_i = y_i \lor f(\At(x,i;R'))_i
				= y_i }}\\
		&= \frac{\sum_{x' \in
				\Fc^{-1}(y_{1,\ldots,i})}  \pr{\Ac(x,i;R') = \At(x,i;R') = x' \land (\Ac(x,i;R') = x' \lor \At(x,i;R') = x')}}{\sum_{x' \in
				\Fc^{-1}(y_{1,\ldots,i})}  \pr{ \Ac(x,i;R') = x' \lor \At(x,i;R') = x'}}\\
		&= \frac{\sum_{x' \in
				\Fc^{-1}(y_{1,\ldots,i})} \min(P(x'),P^*(x'))}{\sum_{x \in
				\Fc^{-1}(y_{1,\ldots,i})} \max(P(x'),P^*(x'))}
		=
		\frac{p - \eps}{p + \eps}
		= \frac{1-\frac{\eps}{p}}
		{1+\frac{\eps}{p}}
		\geq 1-\frac{2\eps}{p}\;.
	\end{align*}
	For the  penultimate equality, note that  
	\begin{align*}
		\eps &= \SD(\Ac(x,i), \At(x,i)) = \sum_{x' \in
			\Fc^{-1}(y_{1,\ldots,i})} P^*(x') - \min(P(x'),P^*(x')).
	\end{align*}
	Hence,
	\begin{align*}
		\sum_{x' \in
			\Fc^{-1}(y_{1,\ldots,i})}\min(P(x'),P^*(x')) &= \sum_{x' \in
			\Fc^{-1}(y_{1,\ldots,i})}\min(P(x'),P^*(x'))  + P^\ast(x') - P^\ast(x')\\
		& = \sum_{x' \in
			\Fc^{-1}(y_{1,\ldots,i})} P^*(x')- \left(\sum_{x' \in
			\Fc^{-1}(y_{1,\ldots,i})} P^*(x') - \min(P(x'),P^*(x'))\right)\\
		&= p -\eps. 
	\end{align*}
	And similarly, $\sum_{x' \in
		\Fc^{-1}(y_{1,\ldots,i})}\max(P(x'),P^*(x'))  = p+\eps$.
	
	Collecting the equations and inserting into (\ref{eq:extendOne})
	proves the lemma.
\end{proof}

\paragraph{Putting everything together.}
We can now finish the proof of \cref{thm:OwfToIASahannon}.  Consider the following random variables: let $X$ be uniformly drawn from $\zn$ and let $Y = f(X)$.  Run $\Inv^\Ac(Y)$ and
$\Inv^{\At}(Y)$ in parallel, using the same randomness in both
executions.  Let $\XT^{(0)},\ldots,\XT^{(n)}$ be the random variables
which have the values assigned to $x^{(0)},\ldots,x^{(n)}$ in the run
of $\Inv^{\At}$. Finally, let the indicator variables $Q_i$ be $1$, iff the $i$'th call to
$\ExtendOnex$ in the above parallel run  is the \emph{first} call such that $\ExtendOnex^\Ac(\XT^{(i)},Y_i,i;\cdot) \neq \ExtendOnex^{\At}(\XT^{(i)},Y_i,i;\cdot)$.

We proceed to obtain an upper bound on $\Pr[Q_i = 1]$. Observe that for all $x \in \zo^n$:
\begin{align*}
	\lefteqn{\Pr[Q_i = 1 \mid \XT^{(i-1)} = x]}\\
	&= \Pr[\ExtendOnex^\Ac(x,Y_i,i;R) \neq \ExtendOnex^{\At}(x,Y_i,i;R) \mid \XT^{(i-1)} = x]\\
	&= \sum_{y_i \in \zo}
	p(y_i\mid f(x)_{1,\ldots,i-1}) \cdot
	\Pr[\ExtendOnex^\Ac(x,y_i,i;R) \neq \ExtendOnex^{\At}(x,y_i,i;R)] \\
	&\leq \sum_{y_i \in \zo}
	p(y_i\mid f(x)_{1,\ldots,i-1}) \cdot
	\frac{2 \epsilon(x,i)}{p(y_i\mid f(x)_{1,\ldots,i-1})} \\
	&= 4 \epsilon(x,i)
\end{align*}
where the inequality above follows by \cref{lem:boundingTheEvents}.
Averaging over $x$, we have that for all $i=1,\ldots,n$:
\begin{align}\label{eq:OwfToIASahannon:1}
	\Pr[Q_i = 1] \leq 4 \Exp_{x \getsr \zo^n}[\epsilon(x,i)]
\end{align}
Here, we use the fact that by induction on $i$, the random variable
$\XT^i$, for $i\in \set{0,\dots,n}$, is uniformly distributed in $\zn$
(it is uniform preimage of a uniformly chosen output). Using
\cref{eq:OwfToIASahannon:1}, we have
\begin{align*}
	\pr{\sum_{i=1}^n Q_i \geq 1} &= \sum_{i=1}^n \Pr[Q_i = 1]\\
	&=  n\cdot \Exp_{i\getsr [n],x\getsr \zn}[Q_i] \nonumber\\
	&\leq 4n \cdot \Exp_{i\getsr [n],x\getsr \zn}[\eps(x,i)]\\
	&\leq \frac{1}{2}
\end{align*}
where the last inequality follows from \cref{lem:expectationOfEpsilon}.
Hence, with probability $\frac{1}{2}$, a run of $\Inv^{\Ac}$
and $\Inv^{\At}$ produce the same output and use the same number
of queries to the oracles $\Ac$. Moreover, the probability that
$\Inv^{\At}$ uses more than $8n$ oracle queries is at most $\frac{1}{4}$
(by applying Markov's inequality on \cref{lem:nOfRepetitions}). Hence,
with probability $\frac{1}{4}$, $\Inv^{\Ac}$ inverts $f$ using $8n$ oracle queries
in total, which contradicts the one-wayness of $f$.
In order to make sure that $\Inv^{\Ac}$ runs in polynomial time, we just halt it
after $8n$ calls.

\newcommand{\tL}{\widetilde{\cL}}

\subsection{A more efficient construction}\label{sec:OwfToIAERomp}
The following theorem shows that a simplified variant of the first step of \cite{Rompel90} (which is also the first step of \cite{KatzKo05}) yields inaccessible entropy with much stronger guarantees than those obtained in \cref{sec:OwfToIASahannon}. The function we construct is $\Fc(x,g,i) = (g(f(x))_{1,\dots,i},g)$, where $g : \zn\mapsto \zn$ is a three-wise independent function. Since the composition of $g$ and $f$ is still a one-way function, \cref{thm:OwfToIASahannon} already implies that $\Fc^{-1}$ has inaccessible entropy. The benefits of the additional hashing step
are that \begin{enumerate}
	\item we get more inaccessible entropy ($\tilde\Theta(1/n)$ bits rather
	than $\tilde\Theta(1/n^2)$ bits), and
	\item we get a bound
	on accessible average max-entropy rather than accessible Shannon entropy. \end{enumerate}
These allow for a more efficient and simpler transformation of $\Fc$
into a UOWHF.

\begin{theorem}[Inaccessible average max-entropy from one-way functions]\label{thm:OwfToIAERomp}
	Let $f\colon \zn \mapsto \zn$ be a one-way function and let $\g = \set{g \colon \zn\mapsto \zn}$ be a family of constructible, three-wise independent hash functions. Define $\Fc$  over $\Dom(F)\eqdef \zn\times \g \times [n]$ by
	$$\Fc(x,g,i) = (g(f(x))_{1,\dots,i},g,i).$$
	Then, for every constant $d>0$, $\Fc^{-1}$ has accessible average max-entropy at most $\HSh(Z\mid \Fc(Z)) - (d\log n)/n$, where $Z$ is uniformly distributed over $\Dom(\Fc)$.
\end{theorem}

\begin{proof}
	Let $c$ be a sufficiently large constant (whose value to be determined later as a function of the constant $d$ in the theorem statement). The sets $\set{\cL(x,g,i)}_{x\in \zn, i\in [n],g\in \g}$ realizing the inaccessible entropy of $\Fc^{-1}$ are defined by
	\begin{align}
		\cL(x,g,i) = \{ (x',g,i) \colon f(x') \in \tL(f(x),i) \wedge
		g(f(x'))_{1,\ldots,i} = g(f(x))_{1,\ldots,i} \}
	\end{align}
	
	where for $y\in \zn$ and $i\in [n]$, we let
	\begin{align}\label{eq:rompelconstructTildeL}
		\tL(y,i) &= \set{y} \cup \set{y'\in \zn\colon \Hsam_{f(X)}(y') \geq (i + c\cdot \log n)}\\
		&= \set{y} \cup \set{y'\in \zn\colon |f^{-1}(y')| \leq 2^{n-i}/n^c}.\nonumber
	\end{align}
	Namely, $\tL(y,i)$ consists, in addition to $y$ itself, of ``$i$-light'' images \wrt $f$.\footnote{Recall that the sample entropy is defined as
		$\Hsam_{f(X)}(y) = \log(1/\Pr[f(X)=y]) = n-\log\size{f^{-1}(y)}$, so the ``heavy'' images, where $f^{-1}(y)$ is large, have low sample entropy.} As
	a warm-up, it is helpful to write down $\tL(y,i)$ and $\cL(x,g,i)$ for
	the case where $f$ is a one-way permutation.\footnote{If $f$ is
		a permutation, then $\tL(y,i)$ is given by:
		\begin{eqnarray*}
			\tL(y,i) = \begin{cases}
				\zn & \mbox{if $i \leq n-c \log n$} \\
				\{y\} & \mbox{otherwise.}
			\end{cases}
		\end{eqnarray*}
		Then, for all $x \in \zo^n$, we have
		$\Exp[|\cL(x,G,i)|] = 2^{n-i}$ for all $i \leq
		n-c \log n$ and $|\cL(x,g,i)| = 1$ for all $g \in \g$ and all $i
		> n- c \log n$. This means that the entropy gap between
		$F^{-1}(F(Z))$ and $\cL(X,G,I)$ is roughly
		$\frac{1}{n}\sum_{i > n-c \log n} n-i = \Omega(c^2 \log^2 n / n)$.
	}

	The proof of the theorem immediately follows by the following two claims.
	\begin{claima}\label{claim:hardSibling}
		For every \ppt $\Fc$-collision-finder $\Ac$ and every constant $c>0$, it holds that
		\[\Pr[\Ac(Z;R) \notin \cL(Z)] \leq \negl(n),\]
		where $Z$ is uniformly distributed over $\Dom(F)$ and $R$ is uniformly distributed over the random coins of $\Ac$.
	\end{claima}
	
	\begin{claima}\label{claim:ManyHardSibling}
		For any constant $c$  it holds that
		$$\Ex\left[\log \size{\cL(Z)}\right] \leq
		\Ex\left[\log \size{\Fc^{-1}(\Fc(Z))}\right] - \Omega\left(\frac{c\log n}{n}\right),$$
		where $Z$ is uniformly distributed in $\Dom(\Fc)$.
	\end{claima}
\end{proof}

\remove{
	Since  $\size{\cL(x,g,i)} = \size{f^{-1}(\tL(f(x),i))}$ for every $(x,g,i)$, the proof of the theorem follows.
	In the rest of the proof, we work primarily with $\tL$ and not $\cL$, as it significantly simplifies notations.
	
	Further, note that the sets $\tL$ are independent of $\g$.

	We show below in \cref{claim:hardSibling} and
	\cref{claim:ManyHardSibling} that the sets $\cL(x,g,i)$
	satisfy the following properties:
	\begin{itemize}
		\item First, the only accessible inputs of $\Fc$ come from preimages of
		$\tL(y,i)$, that is, for every \ppt $\Fc$-collision-finder $\Ac$,
		\[\Pr[\Ac(X,G,I;R) \notin \cL(X,G,I)] \leq \negl(n)\]
		\item Next, the sets $\tL$ and $\cL$ are not too large (relative to the
		real entropy of $\Fc$), that is,
		\[\Ex\left[\log \size{f^{-1}(\tL(f(X),I))}\right] \leq
		\Ex\left[\log \size{\Fc^{-1}(\Fc(X,G,I))}\right] - \Omega\left(\frac{c\log n}{n}\right)\]
		Note that for all $x,g,i$, $\size{\cL(x,g,i)} = \size{f^{-1}(\tL(f(x),i))}$.
	\end{itemize}
	The bound on accessible average max-entropy follows readily
	once we establish these claims.}

\subsubsection{Accessible inputs of $\Fc$ --- Proving \cref{claim:hardSibling}}


\begin{proof}[Proof of \cref{claim:hardSibling}]
	Recall that $Z= (X,G,I)$ is uniformly distributed over $\Dom(F)$, and that $R$ is uniformly distributed over the random coins of $\Ac$.
	Let $\Ac_1$ denote the first component of $\Ac$'s output. It suffices to show that
	\begin{align}
		\Pr[\Ac_1(X,G,I;R) \notin f^{-1}(\tL(f(X),I))] \leq \negl(n)
	\end{align}
	since the other two output components of $\Ac$ are required to equal
	$(G,I)$, due to the fact that $\Fc(X,G,I)$ determines $(G,I)$.

	We construct an inverter $\Inv$ such that for all
	$\Fc$-collision-finders $\Ac$ and for $c$ as in \cref{eq:rompelconstructTildeL} we have
	\begin{align}
		\Pr[\Inv^{\Ac}(Y) \in f^{-1}(Y)]
		\geq \frac{1}{n^c} \cdot \Pr[\Ac_1(X,G,I;R) \notin f^{-1}(\tL(f(X),I))]
	\end{align}
	where $Y = f(X)$, and the proof of \cref{claim:hardSibling} follows readily from the one-wayness of $f$.
	
	\bigskip
	\noindent\framebox{
		\begin{minipage}{16cm}
			\noindent \textbf{Inverter $\Inv^{\Ac}$}
			\medskip
			
			\hrule
			
			\medskip
			
			\textbf{Oracle:} An $\Fc$-collision finder $\Ac$.\\
			\textbf{Input:} $y \in \set{0,1}^n$
			
			\medskip\hrule\medskip
			
			$x \getsr \zn$\\
			$i \getsr [n]$\\
			$g' \getsr \g_{y,x,i} \eqdef \set{g \in \g \colon  g(y)_{1\dots i} = g(f(x))_{1\dots i}}$\\
			\textbf{return} $A_1( x, g', i; r)$
		\end{minipage}
	}
	\bigskip
	
	\noindent
	Observe that $\Inv$ can be implemented efficiently by sampling
	$g'$ as follows:
	pick first $z,z^\ast \in \zn$
	such that $z_{1\dots i} = z^\ast_{1\dots i}$ and use the constructibility
	of $\g$ to pick $g$ with $g(f(x)) = z$ and $g(y) = z^\ast$.
	
	We analyze the success probability of $\Inv^\Ac$. Using the short hand notation $\Pr_{g'}[\cdots]$ for $\Pr_{g'\getsr \g_{y,x,i}}[\cdots]$ we observe that 
	\begin{align}
		\Pr[\Inv^{\Ac}(Y) \in f^{-1}(Y)]&=\!\!\!\! \Exp_{x\getsr \zn,i \getsr [n]} \Bigl[ \sum_{y \in \zn}\nolimits
		\Pr[f(X) = y] \cdot \Pr_{g',r} [ \Ac_1(x,g',i;r) \in f^{-1}(y)] \Bigr]\\
		&\geq \Exp_{x,i} \Bigl[ \sum_{y \notin \tL(f(x),i)}\nolimits
		\frac{2^{-i}}{n^c} \cdot \Pr_{g',r}[ \Ac_1(x,g',i;r) \in f^{-1}(y)] \Bigr]\nonumber
	\end{align}
	where the inequality holds since $\Pr[f(X) = y] \geq 2^{-i}/n^c$ for any $y \notin \tL(f(x),i)$.
	
	Next, observe that for any  tuple $(y,x,i)$ such that $y \neq f(x)$, it holds that (where we distinguish $\Pr_{g'}[\cdots]$ as above from $\Pr_{g}[\cdots]= \Pr_{g \getsr \g}[\cdots]$)
	\begin{align}
		\Pr_{g',r}\bigl[ \Ac_1(x,g',i;r) \in f^{-1}(y)\bigr] &= \Pr_{g\getsr\g,r}\bigl[ \Ac_1(x,g,i;r) \in f^{-1}(y)
		\mid g(f(x))_{1\cdots i} = g(y)_{1\cdots i} \bigr] \\
		&= \frac{\Pr_{g,r}\bigl[ \Ac_1(x,g,i;r) \in f^{-1}(y) \wedge g(f(x))_{1\cdots i} = g(y)_{1\cdots i}\bigr]}
		{\Pr_{g,r}\bigl[g(f(x))_{1\cdots i} = g(y)_{1\cdots i} \bigr]}\nonumber \\
		&= \frac{\Pr_{g,r}\bigl[ \Ac_1(x,g,i;r) \in f^{-1}(y) \bigr]}
		{\Pr_{g,r}\bigl[g(f(x))_{1\cdots i} = g(y)_{1\cdots i} \bigr]}\nonumber \\
		&= 2^i \cdot \Pr_{g,r}\bigl[ \Ac_1(x,g,i;r) \in f^{-1}(y) \bigr].\nonumber
	\end{align}
	The second equality follows by Bayes' rule and the third uses the fact
	that $\Ac$ is a $\Fc$-collision finder. The last equality follows since $\g$ is two-wise independent (recall we assumed that $\g$ is three-wise independent) and $f(x) \neq y$.
	
	Combining the two preceding observations, and the fact that
	$f(x) \in \tL(f(x),i)$, we have that
	\begin{align*}
		\Pr[\Inv^{\Ac}(Y) \in f^{-1}(Y)]&\geq \Exp_{x\getsr \zn,i \getsr [n]} \Bigl[ \sum_{y \notin \tL(f(x),i)}\nolimits
		\frac{2^{-i}}{n^c} \cdot 2^i \cdot \Pr_{g\getsr \g,r}\bigl[ \Ac_1(x,g,i;r) \in f^{-1}(y)\bigr] \Bigr]\\
		&\geq \frac{1}{n^c} \cdot \Exp_{x,i} \Bigl[ \sum_{y \notin \tL(f(x),i)}\nolimits
		\cdot \Pr_{g,r}\bigl[ \Ac_1(x,g,i;r) \in f^{-1}(y)\bigr] \Bigr]\\
		&= \frac{1}{n^c}
		\cdot \Pr_{x,g,i,r}\bigl[ \Ac_1(x,g,i;r) \notin f^{-1}(\tL(f(x),i))\bigr],
	\end{align*}
	and the proof of the claim follows.
\end{proof}


\subsubsection{Upper bounding the size of $\cL$ --- Proving \cref{claim:ManyHardSibling}}
Recall that $Z=(X,G,I)$ is uniformly distributed over $\Dom(\Fc)$. In the following we relate the size of $\cL(Z)$ to that of $\Fc^{-1}(\Fc(Z))$.

We make use of the following property of
\emph{three}-wise independent  hash-functions.

\begin{claima}\label{claim:almost-uniform}
	Let $i$, $x$ and $x^\ast$, be such that $f(x) \neq f(x^\ast)$
	and $i \leq \HSh_{f(X)}(f(x^\ast)) \leq \HSh_{f(X)}(f(x))$. Then,
	\begin{align*}
		\ppr{{g \getsr \g}\atop{z'  \getsr \Fc^{-1}(\Fc(x,g,i))}}{z' = (x^\ast,g,i)} \geq \frac{2^{-n}}{8}.
	\end{align*}
	
\end{claima}
Note that in the above experiment it is always the case that $(g',i') = (g,i)$, where $z' = (x',g',i')$.

\begin{proof}
	Note that with probability $2^{-i}$ over $g \getsr \g$, it holds that
	$g(f(x^\ast))_{1\cdots i} = g(f(x))_{1\cdots i}$. Henceforth, we
	condition on this event that we denote by $E$, and let $w =
	g(f(x^\ast))_{1\cdots i} = g(f(x))_{1\cdots i}$. Observe that for
	a fixed $g$ satisfying $E$, it holds that 
	\begin{align}
		\ppr{z' \getsr \Fc^{-1}(\Fc(x,g,i))}{z' = (x^\ast,g,i) \mid E}
		\geq \frac{1}{\size{F^{-1}(F(x,g,i))}}
	\end{align}
	In order to obtain a lower bound on $\size{F^{-1}(F(x,g,i))}$, we first
	consider $x'$ such that $f(x') \notin \{ f(x), f(x^\ast) \}$. By the
	three-wise independence of $\g$,
	\begin{align}
		\Pr_{g\getsr \g}\bigl[g(f(x')) = w \mid E\bigr] = 2^{-i}
	\end{align}
	This implies that the expected number of $x'$ such that $g(f(x')) = w$
	and $f(x') \notin \{ f(x), f(x^\ast) \}$ is at most $2^{n-i}$. By
	Markov's inequality, we have that with probability at least $1/2$ over
	$g \getsr \g$ (conditioned on $E$),
	\begin{align}
		|F^{-1}(F(x,g,i))| \leq 2 \cdot 2^{n-i} + |f^{-1}(f(x))| +
		|f^{-1}(f(x^\ast))| \leq 4 \cdot 2^{n-i},
	\end{align}where the second inequality uses the fact that
	$i \leq \HSh_{f(X)}(f(x^\ast)) \leq \HSh_{f(X)}(f(x))$.
	Putting everything together, we have that the probability
	we obtain $x^\ast$ is at least $2^{-i} \cdot 1/2 \cdot (4 \cdot 2^{n-i})^{-1}
	= 2^{-n}/8$.
\end{proof}
We now use \cref{claim:almost-uniform} for proving \cref{claim:ManyHardSibling}.

\begin{proof}[Proof of \cref{claim:ManyHardSibling}]
	Let  $Z' = (X',G,I) \getsr F^{-1}(F(Z = (X,G,I)))$ (note that indeed the second and third coordinates of $Z$ and $Z'$ are guaranteed to  match).
	We claim that for proving \cref{claim:ManyHardSibling} it suffices to show that 
	\begin{align}\label{eq:OwfToIAERomp:3}
		\Pr[X' \notin f^{-1}(\tL(f(X),I))] \in \Omega\left(\frac{c \log(n)}{n}\right)
	\end{align}
	
	Indeed, let $\Lbar(z) \eqdef F^{-1}(F(z)) \setminus \cL(z)$,
	and  compute
	\begin{align}
		\Ex\left[\log\size{\Fc^{-1}(\Fc(Z))}\right] -
		\Ex\left[\log \size{\cL(Z)}\right] &=
		\Ex\left[\log\Bigl(1+ \frac{\size{\Lbar(Z)}}{\size{\cL(Z)}}\Bigr)\right]\\
		&\geq
		\Ex\left[\log\Bigl(1+ \frac{\size{\Lbar(Z)}}{\size{F^{-1}(F(Z))}}\Bigr)\right]\nonumber\\
		&\geq
		\frac12
		\Ex\left[\frac{\size{\Lbar(Z)}}{\size{F^{-1}(F(Z))}}\right] \nonumber\\
		&= \frac12 \Pr[X' \notin f^{-1}(\tL(f(X),I))]\nonumber\\
		& \in \Omega\left(\frac{c\log(n)}{n}\right).\nonumber
	\end{align}
	The first equality holds since by definition $\cL(Z) \subseteq \Fc^{-1}(\Fc(z))$, the second inequality holds since $\log(1+\alpha) \geq \frac{\alpha}{2}$ for $\alpha \in [0,1]$ and the containment by  \cref{eq:OwfToIAERomp:3}.

	We prove \cref{eq:OwfToIAERomp:3} in two steps. First, observe that for all $x$:
	\begin{align}
		\lefteqn{\Pr_{(g,i) \getsr \g \times [n], (x',g,i) \getsr \Fc^{-1}(\Fc(x,g,i))} \bigl [ f(x') \notin \tL(f(x),i) \big]}\\
		&\geq \Pr_{g,i,(x',g,i) \getsr \Fc^{-1}(\Fc(x,g,i))} \bigl [ f(x') \neq f(x)
		\land (i \leq \HSh_{f(X)}(f(x')) < i+c \log n)
		\big]\nonumber\\
		&= \frac{1}{n}\cdot
		\sum_{x^\ast \colon f(x) \neq f(x^\ast)}
		\Bigl(\sum_{i \leq \HSh_{f(X)}(f(x^\ast)) < i+c \log n}
		\Pr_{g} \bigl [ (x^\ast,g,i) \getsr \Fc^{-1}(\Fc(x,g,i)) \big]\Bigr)\nonumber\\
		&\geq \frac{1}{n} \cdot 
		\sum_{x^\ast \colon f(x) \neq f(x^\ast)
			\atop{\land (\HSh_{f(X)}(f(x^\ast)) \leq \HSh_{f(X)}(f(x)))}}
		\Bigl(\sum_{i \leq \HSh_{f(X)}(f(x^\ast)) < i+c \log n}
		\frac{2^{-n}}{8} \Bigr) \qquad \mbox{(by \cref{claim:almost-uniform})}\nonumber\\
		&\geq \frac{c \log n}{n} \cdot 
		\sum_{x^\ast \colon f(x) \neq f(x^\ast)
			\atop{\land (\HSh_{f(X)}(f(x^\ast)) \leq \HSh_{f(X)}(f(x)))}}
		\frac{2^{-n}}{8}\nonumber\\
		&= \frac{c \log n}{8n} \cdot \Pr_{x^\ast}\bigl[f(x) \neq f(x^\ast)
		\land (\HSh_{f(X)}(f(x^\ast)) \leq \HSh_{f(X)}(f(x)))\bigr].\nonumber
	\end{align}
	It follows that 
	\begin{align*}
		\Pr[X' \notin f^{-1} (\tL(f(X),I))] &= \Pr_{(x,g,i)  \getsr \zn\times \g \times [n]} \bigl [ x' \notin \tL(f(x),i) \colon
		(x',g,i) \getsr \Fc^{-1}(\Fc(x,g,i)) \big]\\
		&\geq \frac{c \log n}{8n} \cdot\Pr_{(x,x^\ast) \getsr \zn \times \zn}\bigl[f(x) \neq f(x^\ast)
		\land (\HSh_{f(X)}(f(x^\ast)) \leq \HSh_{f(X)}(f(x)))\bigr]\nonumber\\
		&\geq \frac{c \log n}{8n} \cdot \frac{1}{2} \cdot
		\Bigl(1-\Pr_{x,x^\ast}[f(x) \neq f(x^\ast)]\Bigr)\nonumber\\
		&\geq \frac{c \log n}{8n} \cdot \frac{1}{2} \cdot
		\frac12.\nonumber
	\end{align*}
	The last inequality holds since  the one-wayness of $f$ yields that $\Pr_{x,x^\ast}[f(x)=f(x^\ast)]$ is negligible (otherwise inverting $f$ is
	trivial). This concludes the the proof of \cref{eq:OwfToIAERomp:3}, and hence of the claim.
\end{proof}

\section{UOWHFs from inaccessible entropy}\label{sec:OWHFsfromIE}
In this section we show how to construct a UOWHF from any efficiently
computable function with a noticeable gap between real Shannon entropy
and either accessible average max-entropy or accessible Shannon
entropy. Recall that the more efficient construction from
\cref{sec:OwfToIAERomp} satisfies the former, and the more direct
construction from \cref{sec:OwfToIASahannon} satisfies the
latter. Combined with these constructions, we obtain two new
constructions of UOWHFs from any one-way function.

In both cases, we first transform the entropy gap into a noticeable
gap between real Shannon entropy and accessible \emph{max}-entropy. We begin
with the construction that starts from a gap between real Shannon entropy
and accessible average max-entropy because the transformation involves fewer steps
(and is also more efficient).

\newcommand{\Xt}{X^{(t)}}
\newcommand{\xt}{x^{(t)}}
\newcommand{\xtTag}{{x'}^{(t)}}
\renewcommand{\labelenumi}{\roman{enumi}.}
\newcommand{\lin}{\ell^\textsc{in}}
\newcommand{\lout}{\ell^\textsc{out}}
\newcommand{\fin}{\widetilde{n}}
\newcommand{\fout}{m}

\newcommand{\zeroIn}{n_0}
\newcommand{\zeroOut}{m_{0}}
\newcommand{\oneIn}{n_{1}}
\newcommand{\oneOut}{m_{1}}
\newcommand{\twoIn}{n_{2}}
\newcommand{\twoOut}{m_{2}}
\newcommand{\threeIn}{n_{3}}
\newcommand{\threeOut}{m_{3}}
\newcommand{\fourIn}{n_{3}}
\newcommand{\fourOut}{m_{3}}

\subsection{The more efficient UOWHF}\label{sec:WufRomp}
\begin{theorem} \label{thm:TiegToUOHF1}
	Suppose there exists a polynomial-time computable function $\Fc :
	\zo^{\fin} \mapsto \zo^{\fout}$ such that $\Fc^{-1}$ has a noticeable gap
	$\Delta$ between real Shannon entropy and accessible average
	max-entropy. Then, there exists a family of universal one-way hash
	functions with output length $O(\fin^4s/\Delta^3)$ and key
	length $O(\fin^4s/\Delta^3 \cdot \log n)$ for any $s = \omega(\log n)$,
	where $n$ is the security parameter.\footnote{Note that $\Delta$ is not required to be efficiently computable.}
\end{theorem}

We first show how to combine this with \cref{thm:OwfToIAERomp}
to get a universal one-way hash function.

\begin{theorem}\label{thm:OWFToUOWHFStrong}
	Suppose there exists a one-way function $f :
	\zo^{n} \mapsto \zo^{n}$.
	Then, there exists a family of universal one-way hash
	functions with key and output length $O(n^7)$.
\end{theorem}
\begin{proof}
	Fixing $s := \log^2(n)$, we use \cref{thm:OwfToIAERomp} to get
	a function $\Fc : \zo^{\fin} \mapsto \zo^{\fout}$ with $\fin = O(n)$
	and gap $\Delta = \log n/n$ between real Shannon entropy and accessible
	average max-entropy.
	By \cref{thm:TiegToUOHF1} we get a family of universal one-way hash
	functions with output length $O(\fin^4 s/\Delta^3) = 
	O(n^7 \log^2(n)/ \log^3(n))$ and key length 
	$O(\fin^4 s /\Delta^3 \cdot \log(n)) = O(n^7)$. 
\end{proof}
\paragraph{Overview.}
The construction proceeds via a series of transformations
as outlined in \cref{sec:intr:constructions}:
gap amplification (via repetition), entropy
reduction (by hashing inputs) and reducing output
length (by hashing outputs).
In each of these transformations, we use $\zeroIn$ to denote the input
length of the function $F$ we start with, and 
$n$ to denote the security parameter.

\subsubsection{Gap amplification}
Here, we show that a direct product
construction increases the gap
between real entropy and accessible entropy. Another useful effect of
direct product (for certain settings of parameters) is turning real
Shannon entropy into real min-entropy, and turning accessible average
max-entropy into accessible max-entropy.

\begin{lemma}[Gap amplification]\label{lem:gapamp}
	Let $n$ be a security parameter and $\Fc \colon  \zo^{\fin} \mapsto \zo^{\fout}$
	be a function. For $t \in \poly(n)$, let $\Fc^t$ be the $t$-fold
	direct product of $\Fc$. Then, $\Fc^t$ satisfies the following
	properties:
	\begin{enumerate} 
		\item\label{amp:realH2min} If $\Fc^{-1}$ has real Shannon entropy at
		least $k$, then $(\Fc^t)^{-1}$ has real min-entropy at least $t
		\cdot k - \fin\cdot \sqrt{st}$ for any $s = \omega(\log n)$
		and $t > s$.

		\item\label{amp:accavgmax}
		If $\Fc^{-1}$ has accessible average max-entropy at most $k$, then
		$(\Fc^t)^{-1}$ has accessible max-entropy at most $t \cdot k + \fin
		\cdot \sqrt{st}$ for any $s = \omega(\log n)$.
	\end{enumerate}
\end{lemma}

\begin{proof}
	In the following $X$ and $\Xt = (X_1,\dots,X_t)$ are
	uniformly distributed over $\zo^{\fin}$ and $(\zo^{\fin})^t$\,,
	respectively.  
	
	\begin{enumerate} 
		\item Follows readily from \cref{lem:flattening} (with $\epsilon = 2^{-s}$).
		
		\item Given any \ppt $\Fc^t$-collision-finder $\Ac'$, we construct a \ppt
		$\Fc$-collision-finder $\Ac$ that:
		\begin{quote}
			On input $x$, picks a random $i$ in $[t]$ along with random
			$x_1,\ldots,x_{i-1},x_{i+1},\ldots,x_t$, computes
			$\Ac'(x_1,\ldots,x_t) \mapsto (x'_1,\ldots,x'_t)$, and outputs
			$x'_i$.
		\end{quote}
		By the bound on the accessible average max-entropy of $\Fc^{-1}$, we
		know that there exists a family of sets $\set{ \cL(x)}$ such that $\Exp\bigl[\log \size{\cL(X)}\bigr] \leq k$, $x\in \cL(x)$, and
		$\Pr[\Ac(X) \notin \cL(X)] \leq \negl(n)$. Consider the family of
		sets $\set{ \cL'(\xt) \colon  \xt \in (\zo^{\fin})^t}$ given by:
		$$\cL'(\xt) = \cL(\xt_1) \times \cL(\xt_2) \times \cdots \times \cL(\xt_t).$$
		By linearity of expectations, we have $\Exp\bigl[\log
		\size{\cL'(X_1,\ldots,X_t)}\bigr] \leq t \cdot k$. Moreover, by the
		Chernoff-Hoeffding bound and using the fact that $\log\size{\cL(X)}$ assumes
		values in $[0,\fin]$, we have
		\begin{align}
			\lefteqn{\Pr\bigl[ \log \size{\cL'(\Xt)} \geq t\cdot k + \fin\sqrt{st}\bigr]}\\
			&= \Pr\bigl[ \log \size{\cL(\Xt_1)} + \cdots + \log \size{\cL(\Xt_t)} \geq t\cdot k + \fin \sqrt{st} \bigr] \leq \eee^{-2s}.\nonumber
		\end{align}
		We claim that this implies that $\Ac'$ has accessible max-entropy at most
		$t \cdot k + \fin \sqrt{st}$. Suppose otherwise, then there
		exists a non-negligible function $\epsilon$ such that
		\begin{align*}
			\Pr[\Ac'(\Fc^t(\Xt)) \notin \cL'(\Xt) ] \geq \epsilon-\eee^{-2s}
			\geq \epsilon/2
		\end{align*}
		Therefore,
		\begin{align*}
			\Pr[\Ac(\Fc(X)) \notin \cL(X)]
			= \Pr[\Ac'(\Fc^t(\Xt)) \notin \cL'(\Xt) ]/t
			\geq \epsilon/2t
		\end{align*}
		which contradicts our assumption on $\Ac$.
	\end{enumerate}
\end{proof}

\subsubsection{Entropy reduction}
Next we describe a construction that given $\Fc$ and any parameter
$\ell$, reduces the accessible max-entropy of $\Fc^{-1}$ by roughly
$\ell$ bits, while approximately preserving the gap between real
min-entropy and accessible max-entropy.

\begin{lemma}[Reducing entropy]\label{lem:redent}
	Let $n$ be a security parameter and $\Fc \colon  \zo^{\fin} \mapsto \zo^{\fout}$
	be a function. Fix a family of pairwise independent hash functions $\g
	= \set{ g \colon \zo^{\fin} \mapsto \zo^\ell}$. Then, $\Fc' \colon  \zo^{\fin} \times \g
	\mapsto \zo^{\fout} \times \g \times \zo^\ell$ as given by
	$\Fc'(x,g)=(\Fc(x),g,g(x))$ satisfies the following properties:
	
	\begin{enumerate} 
		\item\label{red:realmin}
		Assuming $\Fc^{-1}$ has real min-entropy at least $k$, then $(\Fc')^{-1}$
		has real min-entropy at least $k-\ell-s$ for any $s=\omega(\log n)$.
		
		\item\label{red:accmax} Assuming $\Fc^{-1}$ has accessible max-entropy at most $k$, then $(\Fc')^{-1}$
		has accessible max-entropy at most $\max\set{k-\ell+s,0}$ for any
		$s = \omega(\log n)$.
	\end{enumerate}
\end{lemma}
\begin{proof}
	In the following $X$ and $\G$ are uniformly distributed over $\zo^{\fin}$ and $\g$, respectively.
	\begin{enumerate}
		\item Fix $g\in \g$ and let $S_g = \set{z \in \zo^\ell
			\colon \Pr[g(X) = z] \leq 2^{-\ell-s} }$. Observe
		the following.  
		\begin{enumerate}
			\item $\Pr[g(X) \in S_g] \leq 2^{-s}$ (by a union bound over $z \in S_g$);
			\item Fix any $z \notin S_g$ and any $x \in \zo^n$ such that
			$\Hsam_{X \mid F(X)}(x \mid F(x)) \geq k$. Then,
			\begin{align*}
				\Pr[X = x \mid \Fc'(X,g) = (\Fc(x),g,z)]
				&=\Pr[X = x \mid \Fc(X) = \Fc(x) \wedge g(X) = z]\\
				&\leq \frac{\Pr[X = x \mid \Fc(X) = F(x)]}{\Pr[g(X) = z]}\\
				&\leq \frac{2^{-k}}{2^{-\ell-s}} = 2^{-(k-\ell-s)}.
			\end{align*}
			where the second inequality follows from our assumptions on $z$ and $x$.
		\end{enumerate}
		Combining the above two observations and the bound on the real
		min-entropy of $F$, it follows that for all $g \in \g$,
		with probability $1-2^{-s}-\negl(n)$ over $x \getsr X$, we have
		\begin{align*}
			\Pr[X = x \mid \Fc'(X,g) = \Fc'(x,g)] \leq 2^{-(k-\ell-s)}.
		\end{align*}
		The bound on the real min-entropy of $\Fc'$ follows readily.
		\item 
		Given a \ppt $\Fc'$-collision-finder $\Ac'$, we construct a
		\ppt $\Fc$-collision-finder $\Ac$ as follows:
		\begin{quote}
			On input $x$, picks a pair $(g,r)$ uniformly at random and output $\Ac'(x,g;r)$.
		\end{quote}
		By the bound on the accessible max-entropy of $\Fc^{-1}$, we know that 
		there exists a family of sets $\set{ \cL(x) \subseteq \zo^{\fin} \colon
			x\in \zo^{\fin}}$ such that $\size{\cL(x)} \leq 2^k$, $x\in \cL(x)$, and
		\begin{align}\label{eq:re1}
			\Pr\bigl[\Ac(X,\G;R) \in \cL(X)\bigr] \geq 1-\negl(n),
		\end{align}
		where $R$ is uniformly distributed over the random coins of $\Ac$.
		
		
		Let $\cL'(x,g) \eqdef \set{(x',g) \colon x' \in \cL(x)\land g(x') = g(x)}$.
		\cref{eq:re1} yields that
		\begin{align}
			\Pr\bigl[\Ac'(X,\G;R) \in \cL'(X,\G)\bigr] \geq 1-\negl(n)  
		\end{align}
		We next bound the size of the set $\cL'(x,g)$. Fix any $x \in \zo^n$.
		For any $x' \neq x$, pairwise independence of $\g$ tells us that
		$\Pr[G(x') = G(x)] = 2^{-\ell}$. It follows from linearity of expectation that
		\begin{align*}
			\ex{\size{\cL'(x,\G) \setminus \{x\}}} \leq  |\cL(x)| \cdot 2^{-\ell} \leq 2^{k-\ell}
		\end{align*}
		Then, by Markov's inequality, we have
		\begin{align}
			\Pr\bigl[ \size{\cL'(x,\G)} \leq 2^{k-\ell+s-1}+1\} \bigr] \geq 1-2^{-(s-1)},
		\end{align}
		Combining the last two inequalities, we obtain
		\begin{align}
			\Pr\bigl[\Ac'(X,\G;R) \in \cL'(X,\G)\land
			\size{\cL'(X,\G)} \leq \max\set{2^{k-\ell+s},1}\bigr] \geq 1-\negl(n)-2^{-(s-1)}
		\end{align}
		
		The above yields an upper bound of $\max\set{k-\ell+s,0}$ on the
		accessible max-entropy of $(\Fc')^{-1}$.
	\end{enumerate}
\end{proof}
\subsubsection{Reducing output length}
The next transformation gives us a way
to derive a function that is both length-decreasing and collision-resistant
on random inputs.
\begin{lemma}[Reducing output length]\label{lem:redout}
	Let $n$ be a security parameter and $\Fc \colon \zo^{\fin} \mapsto \zo^m$
	be a function. Fix a family of pairwise independent hash functions $\g
	= \set{g \colon \zo^m \mapsto \zo^{\fin-\log n}}$ and let $\Fc' \colon \g \times \zo^{\fin}
	\mapsto \zo^{\fin-\log n} \times \g$ be defined by
	$\Fc'(x,g)=(g,g(\Fc(x)))$. The following holds: if
	$\Fc^{-1}$ has real min-entropy at least $\omega(\log n)$
	and $\Fc$ is collision-resistant on random inputs, then
	$\Fc'$ is collision-resistant on random inputs.
\end{lemma}

\begin{proof}
	The bound on real min-entropy implies that there exists a
	subset $S \subseteq \zo^{\fin}$ of density at most $\negl(n)$, such that for all
	$x \notin S$ it holds that $\size{\Fc^{-1}(\Fc(x))} = n^{\omega(1)}$. Hence,
	\begin{align}
		\size{\Image(\Fc)} &\leq \size{ \Fc(S) } + \size{ \Fc(\bar{S}) }\leq \size{S} + \size{\bar{S}}/n^{\omega(1)} \leq \negl(n) \cdot 2^n
	\end{align}
	By the two-wise independent of $\g$,
	\begin{align}\label{eqn:output-coll}
		\Pr\bigl[ \exists y' \in \Image(\Fc) \colon y' \neq \Fc(X) \land
		\G(y') = \G(\Fc(X)) \bigr] \leq \frac{\size{\Image(\Fc)}}{2^{\fin-\log n}} \leq \negl(n)
	\end{align}
	Namely, $g(\Fc(x))$ uniquely determines $\Fc(x)$ with high
	probability. In particular, a collision for $g \circ \Fc$ is also a collision
	for $\Fc$. Given any \ppt $\Fc'$-collision-finder $\Ac'$, we
	construct a \ppt $\Fc$-collision-finder $\Ac$ as follows:
	\begin{quote}
		On input $x$, pick $g$ and $r$ at random and compute $x' =
		\Ac'(x,g;r)$. If $\Fc(x') = \Fc(x)$, output $x'$, else output $x$.
	\end{quote}
	\cref{eqn:output-coll} implies that $\Pr[\Ac'(X,\G;R) \neq
	(\Ac(X;\G,R),\G)] \leq \negl(n)$. Therefore, $\Pr[\Ac'(X,\G;R) =
	(X,\G)] \geq 1-\negl(n)$. Namely, $\Fc'$ is also collision-resistant on
	random inputs.
\end{proof}

\subsubsection{Additional transformations}
We present two more standard transformations
that are needed to complete the construction.

\begin{lemma}[From random inputs to targets, folklore]\label{lem:rnd2target}
	Let $n$ be a security parameter and $\Fc \colon \zo^{\fin} \mapsto \zo^m$ be
	a length-decreasing function. Suppose $\Fc$ is collision-resistant on
	random inputs. Then, $\set{ \Fc'_y \colon \zo^{\fin} \mapsto \zo^m }_{y \in
		\zo^{\fin}}$ as defined by $\Fc'_y(x) = \Fc(y+x)$ is
	a family of target collision-resistant hash functions.
\end{lemma}
\begin{proof}
	Given a \ppt adversary $\Ac'$ that breaks target collision-resistance of $\Fc'_y$,
	we can construct a \ppt adversary $\Ac$ that breaks $\Fc$ as follows:
	\begin{quote}
		On input $x$, run $\Ac'(1^n)$ to compute $(x_0,\state)$, and then run $\Ac'(\state, x \oplus x_0)$
		to compute $x_1$. Output $x \oplus x_0 \oplus x_1$.
	\end{quote}
	Note that $(x_0, x_1)$ is a collision for $\Fc'_{x \oplus x_0}$ iff
	$(x, x \oplus x_0 \oplus x_1)$ is a collision for $\Fc$. It then follows quite readily that
	$\Ac$ breaks $\Fc$ with the same probability that $\Ac'$ breaks $\Fc'_y$.
\end{proof}

The following result of \cite{Shoup00:wuf} (improving on
\cite{NaorYu89,BeRo97}) shows that we can construct target
collision-resistant hash functions for arbitrarily long inputs
starting from one for a fixed input length.

\begin{lemma}[Increasing the input length \cite{Shoup00:wuf}]\label{lem:increaseinput}
	Let $n$ be a security parameter, $t = \poly(n)$ be a parameter and let
	$\set{ \Fc_y \colon \zo^{\fin+\log n} \mapsto \zo^{\fin}}$ be a family of target
	collision-resistant hash functions. Then, there exists a family of
	target collision-resistant hash functions $\set{ \Fc'_{y'} \colon \zo^{\fin + t\log n}
		\mapsto \zo^{\fin} }$ where $\size{y'} = O(\size{y} \log t)$.
\end{lemma}

\subsubsection{Putting everything together}

Using these transformations, we can now prove \cref{thm:TiegToUOHF1}.
\begin{proof}[Proof of \cref{thm:TiegToUOHF1}]
	Recall that we have given $\Fc: \zo^{\zeroIn}\mapsto\zo^{\zeroOut}$,
	$s \in \omega(\log n)$, the gap $\Delta$ and that $n$ is the security
	parameter.
	
	\begin{description}
		
		\item[{\sc step 1} (gap amplification):] For a parameter $t$, we define $\Fc_{1}$ as $\Fc_1(x_1,\ldots,x_t)=(\Fc(x_1),\ldots,\Fc(x_t))$, the $t$-fold direct product of $\Fc$.
		We choose the parameter $t \in O(\zeroIn^2s/\Delta^2)$ such that
		$$t\cdot \kreal - \zeroIn \cdot \sqrt{st} \geq t\cdot (\kreal-\Delta/2)+ \zeroIn
		\cdot \sqrt{st} + 3s.$$
		
		\cref{lem:gapamp} yields that this repetition increases both the
		real and accessible entropies of $\Fc_1$ by a factor of $t$ (comparing
		to $\Fc$). In addition, this repetition converts real Shannon
		entropy to real min-entropy and accessible average max-entropy to
		accessible max-entropy (up to additive terms that are sub-linear in
		$t$). More precisely, we have the following properties:
		\begin{itemize}
			\item $\Fc_1: \zo^{\oneIn}\mapsto \zo^{\oneOut}$, where $\oneIn(n) = t \cdot \zeroIn$ and $\oneOut(n) = t\cdot \zeroOut$.
			\item $\Fc_1^{-1}$ has real min-entropy at least $t\cdot\kreal - \zeroIn \cdot \sqrt{st} \geq t\cdot (\kreal-\Delta/2)+ \zeroIn \cdot \sqrt{st} + 3s$.
			\item $\Fc_1^{-1}$ has accessible max-entropy at most $t\cdot (\kreal-\Delta)+ \zeroIn \cdot \sqrt{st}$.
		\end{itemize}
		In steps 2 to 4, the construction uses non-uniform advice $k$,
		which corresponds to an approximation to $\kreal$. In step 5, we will
		remove this non-uniform advice via ``exhaustive search''.
		Concretely, for steps 2 to 4, we are given $k$ satisfying
		\begin{align}\label{eq:WufRomp1}
			k \in [\kreal, \kreal+\Delta/2]
		\end{align}
		This means that
		\begin{itemize}
			\item $\Fc_1^{-1}$ has real min-entropy at least
			$t\cdot (k-\Delta)+ \zeroIn \cdot \sqrt{st} + 3s$.
			\item $\Fc_1^{-1}$ has accessible max-entropy at most
			$t\cdot (k-\Delta)+ \zeroIn \cdot \sqrt{st}$.
		\end{itemize}
		This yields a gap of $3s$ between real min-entropy
		and accessible max-entropy. 
		
		\item[{\sc step 2} (entropy reduction):] We next apply entropy reduction to
		$\Fc_1$ to obtain $\Fc_2^{(k)}$.
		That is, $\Fc_2^{(k)}(x,g)=(\Fc_1(x),g,g(x))$,
		where $g\colon\zo^{\oneIn} \mapsto \zo^\ell$ is selected from a
		family of pairwise independent hash functions with $\ell 
		= t\cdot( k-\Delta)+\zeroIn \cdot \sqrt{st}+s=O(t\zeroIn)$.
		\cref{lem:redent} yields that this additional hashing reduces the
		real min-entropy and accessible max-entropy by $\ell$ (up to an
		additive term of $s$). More exactly, we have the following properties:
		\begin{itemize}
			\item $\Fc_2^{(k)} \colon \zo^{\twoIn} \mapsto \zo^{\twoOut}$ where $\twoIn(n,k) = O(t\zeroIn)$
			and $\twoOut(n,k) = O(t\zeroIn)$. Note that in particular $\twoIn$ and $\twoOut$ also depend on $k$ (unlike $\oneIn$ and $\oneOut$).
			\item If (\ref{eq:WufRomp1}) holds, then $(\Fc_2^{(k)})^{-1}$ has real min-entropy at least $s$.
			\item If (\ref{eq:WufRomp1}) holds, then $(\Fc_2^{(k)})^{-1}$ has accessible max-entropy at most $0$. Hence, $\Fc_2^{(k)}$ is collision-resistant on random inputs (by \cref{lem:0accmax2CR}).
		\end{itemize}
		
		\item[{\sc step 3} (reducing the output length):] We next reduce the output
		length of $\Fc_2^{(k)}$ by hashing the output to $\twoIn-\log n$ bits. That is,
		$\Fc_3^{(k)}(x,g) = (g,g(\Fc_2^{(k)}(x)))$ where $g \colon \zo^{\twoOut} \mapsto
		\zo^{\twoIn - \log n}$ is selected from a family of pairwise-independent
		hash functions.
		\begin{itemize}
			\item $\Fc_3^{(k)} \colon \zo^{\threeIn} \mapsto \zo^{\threeOut}$ where
			$\threeIn(n,k) = O(t\zeroIn)$ and $\threeOut(n,k) = \threeIn-\log n$.
			\item By \cref{lem:redout}, $\Fc_3^{(k)}$ is collision-resistant
			on random inputs, assuming that (\ref{eq:WufRomp1}) holds.
		\end{itemize}
		\item[{\sc step 4} (adding random shifts)]
		We then transform $\Fc_3^{(k)}$ into a family $\set{\Gc^{(k)}_y}$ of target
		collision-resistant hash functions via a random shift, following
		\cref{lem:rnd2target}. That is, $\Gc^{(k)}_y(x) = \Fc_3^{(k)}(y+x)$.
		We then have that 
		\begin{itemize}
			\item $\Gc^{(k)}_y(x): \zo^{\threeIn} \mapsto \zo^{\threeOut}$ and $\Gc^{(k)}_y$ uses a key $y$ of length $\threeIn(n,k)$.
			\item If (\ref{eq:WufRomp1}) holds, then $\set{ \Gc^{(k)}_y }$ is target collision-resistant.
		\end{itemize}
		
		\item[{\sc step 5} (removing non-uniformity):] To remove the
		non-uniform advice $k$, we ``try all possibilities'' from $0$ to
		$\zeroIn$ in steps of size $\Delta/2$, similar to the approach used
		in \cite{Rompel90} (see also \cite[Section~3.6]{KatzKo05})
		\begin{enumerate}
			\item First, we construct $\kappa = \zeroIn \cdot 2/\Delta$ families of
			functions $\set{\Gc^{(k)}_y}$, where we instantiate $\set{\Gc^{(k)}_y}$ for
			all $k \in \set{\frac{\Delta}{2},2\cdot\frac{\Delta}{2},3\cdot\frac{\Delta}{2},\ldots,\zeroIn}$.  These $\kappa$ families of functions satisfy the
			following properties:
			\begin{itemize}
				\item Each of $\Gc^{(k)}_y$ is length-decreasing; in
				particular, $\Gc^{(k)}_y$ has input length $\threeIn(n,k)$
				and output length $\threeIn(n,k)-\log n$.
				Note that $\Gc^{(\zeroIn)}_y$ has the longest input length, \ie
				$\threeIn(n,i\Delta/2) \leq \threeIn(n,\zeroIn)$ for all $i$
				because $\ell(n,k)$ increases as a function of $k$. We may
				then assume that all $\kappa$ functions $\Gc^1_y,\ldots,\Gc^\kappa_y$
				have the same input length $\threeIn(n,\zeroIn)$ and the
				same output length $\threeIn(n,\zeroIn)-\log n$ by padding
				``extra part'' of the input to the output.
				
				\item At least one of the $\set{\Gc^{(k)}_y}$ is target
				collision-resistant; this is because $\kreal \in [0,\zeroIn]$, and so (\ref{eq:WufRomp1}) holds for some $k$ which we picked.
			\end{itemize}
			
			\item Next, for each $k$, we construct a family of
			functions $\set{ \tilde{\Gc}^{(k)}_{\tilde{y}}}$ from $\set{ \Gc^{(k)}_y}$ with
			input length $\kappa \cdot \threeIn(n,\zeroIn)$, key length
			$O(\threeIn(n,\zeroIn) \cdot \log n)$ and output length
			$\threeIn(n,\zeroIn)-\log n$, by following the construction
			given by \cref{lem:increaseinput}.
			Again, at least one of the
			$\set{\tilde{\Gc}^{(k)}_{\tilde{y}}}$ for $k$ as above is
			target collision-resistant.
			
			\item Finally, we define a family of functions $\set{\Gc_{\tilde{y}_1,\ldots,\tilde{y}_\kappa}}$ to be the concatenation of all $\tilde{\Gc}^{(k)}_{\tilde{y}}$ on the same
			input. That is, $\Gc_{\tilde{y}_{1},\ldots,\tilde{y}_{\kappa}}(x) =
			\tilde{\Gc}^{(\Delta/2)}_{\tilde{y}_{1}}(x) \circ \cdots \circ
			\tilde{\Gc}^{(n_0)}_{\tilde{y}_\kappa}(x)$.
			\begin{itemize}
				\item Note that $\Gc$ has input length $\kappa \cdot
				\threeIn(n,\zeroIn)$ and output length
				$\kappa \cdot
				(\threeIn(n,\zeroIn) -\log n)$, so $\Gc$ is
				length-decreasing.
				\item Moreover, since at least one of
				$\set{\tilde{\Gc}^{(\Delta/2)}_{\tilde{y}_{1}}(x)},\ldots, \set{\tilde{\Gc}^{(n_0)}_{\tilde{y}_\kappa}}$ is
				target collision-resistant, $\set{\Gc_{\tilde{y}_1,\ldots,\tilde{y}_\kappa}}$ must
				also be target collision-resistant. This is because a
				collision for $\Gc_{\tilde{y}_1,\ldots,\tilde{y}_\kappa}$ is a collision for each
				of $\tilde{\Gc}^{(\Delta/2)}_{\tilde{y}_1},\ldots, \tilde{\Gc}^{(n_0)}_{\tilde{y}_\kappa}$.
			\end{itemize}
		\end{enumerate}
	\end{description}
	The family $\set{\Gc_{\tilde{y}_1,\ldots,\tilde{y}_\kappa}}$
	is the universal one-way hash function we wanted to construct, and so this finishes the proof of \cref{thm:TiegToUOHF1}.
\end{proof}

\subsection{UOWHF via a direct construction}
\begin{theorem} \label{thm:TiegToUOHF2}
	Suppose there exists a polynomial-time computable function $\Fc :
	\zo^{\fin} \mapsto \zo^{\fout}$ such that $\Fc^{-1}$ has a noticeable gap
	$\Delta$ between real Shannon entropy and accessible Shannon
	entropy. Then, there exists a family of universal one-way hash
	functions with output length $O(\fin^8s^2/\Delta^7)$ and key length
	$O(\fin^8s^2/\Delta^7 \cdot \log n)$ for any $s = \omega(\log n)$.
\end{theorem}

As before, we can use \cref{thm:TiegToUOHF2} together with results
from the previous sections to get a universal one-way hash function.

\begin{theorem}\label{thm:OWFToUOWHFWeak}
	Suppose there exists a one-way function $f :
	\zo^{n} \mapsto \zo^{n}$.
	Then, there exists a family of universal one-way hash
	functions with key and output length $\tO(n^{22})$.
\end{theorem}
\begin{proof}
	We set $s := \log^2(n)$, and use \cref{thm:OwfToIASahannon} to get a $F$ with $\fin = O(n)$, and $\Delta := O(\frac{1}{n^2})$.
	Using $F$ in \cref{thm:TiegToUOHF2} this gives key and output length 
	$\tO(n^{22})$.
\end{proof}

In order to prove \cref{thm:TiegToUOHF2}, we show how to transform a
noticeable gap between real Shannon entropy and accessible Shannon
entropy to one between real Shannon entropy and accessible
max-entropy, and then follow the construction from the previous
section. This step is fairly involved as we are unable to show that
parallel repetition directly transforms an upper bound on accessible
Shannon entropy into one for accessible max-entropy.
We proceed by first 
establishing some additional properties
achieved by gap amplification and entropy reduction.

\begin{lemma}[Gap amplification, continued]\label{lem:gapampmore}
	Let $n$ be a security parameter and $\Fc : \zo^{\fin} \mapsto \zo^m$
	be a function. For $t \in \poly(n)$, let $\Fc^t$ be the $t$-fold
	direct product of $\Fc$. Then, the following holds:
	\begin{enumerate} 
		\item\label{amp:realH2max} If $\Fc^{-1}$ has real Shannon entropy at
		most $k$, then $(\Fc^t)^{-1}$ has real max-entropy at most $t \cdot
		k + \fin \cdot \sqrt{st}$ for any $s= \omega(\log n)$ and $t > s$.
		\item\label{amp:realmin} 
		If $\Fc^{-1}$ has real min-entropy at least $k$, then $(\Fc^t)^{-1}$
		has real min-entropy at least $t \cdot k$.
		
		\item\label{amp:realmax} 
		If $\Fc^{-1}$ has real max-entropy at most $k$, then $(\Fc^t)^{-1}$
		has real max-entropy at most $t \cdot k$.
		
		\item\label{amp:accH}
		If $\Fc^{-1}$ has accessible Shannon entropy at most $k$, then
		$(\Fc^t)^{-1}$ has accessible Shannon entropy at most $t \cdot k$.
		
		\item\label{amp:accmax}
		
		If $\Fc^{-1}$ has accessible max-entropy at most $k$, then
		$(\Fc^t)^{-1}$ has accessible max-entropy at most $t \cdot k$.
		
		\item\label{amp:wuf} If $\Fc$ is $q$-collision-resistant on random
		inputs and $\Fc^{-1}$ has real max-entropy at most
		$k$, then $(\Fc^t)^{-1}$ has accessible max-entropy at most
		$(1-q/8) \cdot tk + t$, provided that $t = \omega((1/q) \cdot \log n)$.
	\end{enumerate}
\end{lemma}

\begin{proof}
	Again, $X$ and $\Xt = (X_1,\dots,X_t)$ are uniformly distributed
	over $\zo^{\fin}$ and $(\zo^{\fin})^t$, respectively.
	\begin{enumerate} 
		\item Follows readily from \cref{lem:flattening}.
		
		\item This follows from a union bound and that fact that for all $x_1,\ldots,x_t$:
		\begin{align*}
			\Hsam_{\Xt}(x_1,\ldots,x_t \mid \Fc^t(x_1,\ldots,x_t)) &=
			\sum_{i=1}^t \Hsam_{X\mid \Fc(X)}(x_i\mid \Fc(x_i))
		\end{align*}
		
		\item Same as previous part.
		
		\item Given any \ppt $\Fc^t$-collision-finder $\Ac'$, we construct the following \ppt
		$\Fc$-collision-finder $\Ac$:
		\begin{quote}
			On input $x$, pick a random $i$ in $[t]$ along with random
			$x_1,\ldots,x_{i-1},x_{i+1},\ldots,x_t$, compute
			$\Ac'(x_1,\ldots,x_t) \mapsto (x'_1,\ldots,x'_t)$, and output
			$x'_i$.
		\end{quote}
		Define the random variables $(X'_1,\ldots,X'_t) = \Ac'(X_1,\ldots,X_t)$. Then,
		\[\begin{array}{rlll}
			\lefteqn{\HSh(X'_1,\ldots,X'_t\mid X_1,\ldots,X_t)}\\
			&\leq& \HSh(X'_1\mid X_1) + \cdots + \HSh(X'_t\mid X_t)&\mbox{subadditivity of conditional Shannon entropy}\\
			&=& t\cdot \HSh(X'_I \mid X_I)&\mbox{where $I$ has the uniform distribution over $[t]$}\\
			&=& t\cdot \HSh(\Ac(X)\mid X)&\mbox{by definition of $\Ac$}\\
			&\leq& t\cdot k&\mbox{by the bound on accessible Shannon entropy of $\Fc^{-1}$}
		\end{array}\]
		
		\item Analogous to \cref{lem:gapamp} part ii, but simpler, since we do not have to
		use the Chernoff-Hoeffding bound.
		
		\item Suppose on the contrary that there exists a \ppt $\Fc^t$-collision-finder $\Ac'$ that violates the guarantee on accessible max-entropy.
		For $\xt \in (\zo^{\fin})^t$, 
		
		let $B(\xt) \eqdef
		\set{\xtTag\in (\zo^{\fin})^t \colon \Fc^t(\xt) = \Fc^t(\xtTag) \land \size{\set{i\in [t]\colon \xtTag_i= \xt_i}} \geq qt/8}$. By the bound on real
		max-entropy, we have that $\Pr[\exists i\in[t] \colon \size{\Fc^{-1}(\Fc(\Xt_i))} > 2^k] \leq t\cdot\negl(n) = \negl(n)$. Hence,
		\begin{align}
			\Pr\Bigl[ \size{B(\Xt)} > {t \choose qt/8} 2^{(1-q/8)tk} \Bigr]
			\leq \negl(n)
		\end{align}
		Since $\Ac'$ achieves accessible max-entropy greater than
		$(1-q/8)tk+t$, there must exists a non-negligible
		function $\epsilon$ such that
		$\Pr[ \Ac'(\Xt;R') \notin
		B(\Xt)] \geq \epsilon - t \cdot \negl(n) \geq \epsilon/2$,
		where $R'$ is uniformly distributed over the random coins of $\Ac'$.
		Namely, $\Ac'$ finds collisions on at least a $1-q/8$ fraction
		of the coordinates with non-negligible probability.
		
		Since $\Fc$ is $q$-collision resistant, this violates a standard Chernoff-type direct product theorem. We provide a proof sketch, following a similar analysis done for standard collision resistance in
		\cite{CRSTVW07}.
		Consider the following \ppt
		$\Fc$-collision-finder $\Ac$:
		\begin{quote}
			On input $x\in \zo^{\fin}$, pick a random $i\in[t]$ along with random
			$x_1,\ldots,x_{i-1},x_{i+1},\ldots,x_t$, compute
			$\Ac'(x_1,\ldots,x_t) \mapsto (x'_1,\ldots,x'_t)$, and output
			$x'_i$.
		\end{quote}
		To analyze the success probability of $\Ac'$, fix any subset $S$ of $\zo^{\fin}$ of density
		$q/2$. If $t =\omega(\log n/q)$, then a Chernoff bound yields that
		\begin{align*}
			\Pr[\Ac'(\Xt) \notin B(\Xt) \land \size{\set{i\in [t] \colon \Xt_i\in S}} \geq q/4] \geq \epsilon/4.
		\end{align*}
		This means that
		\begin{align*}
			\Pr_{i\getsr [t]}[\Ac'(\Xt) \mapsto (X'_1,\ldots,X'_t) \land X_i \in S \land X'_i \neq X_i] \geq \epsilon/4 \cdot q/8.
		\end{align*}
		We may then deduce (following the same calculations in \cite[Prop 2]{CRSTVW07}) that
		\begin{align*}
			\Pr_{x \getsr X} \Bigl[ \Pr[\Ac(x; R) \neq x] \geq \epsilon/4 \cdot q/8 \cdot 2/q \Bigl] \geq 1-q/2.
		\end{align*}
		where $R$ is uniformly distributed over the random coins of $\Ac$.
		By repeating $\Ac$ a sufficient number of times, we may find
		collisions on random inputs of $\Fc$ with probability $1-q$,
		contradicting our assumption that $\Fc$ is $q$-collision-resistant on
		random inputs.
	\end{enumerate}
\end{proof}


\begin{lemma}[Reducing entropy, continued]\label{lem:redentmore}
	Let $n$ be a security parameter and $\Fc : \zo^{\fin} \mapsto \zo^m$
	be a function. Fix a family of 2-universal hash functions $\g
	= \set{ g \colon \zo^{\fin} \mapsto \zo^\ell}$. Then, $\Fc' \colon \zo^{\fin} \times \g
	\mapsto \zo^m \times \g \times \zo^\ell$ as given by
	$\Fc'(x,g)=(\Fc(x),g,g(x))$ satisfies the following properties:
	
	\begin{enumerate} 
		
		\item\label{red:realmax}
		If $\Fc^{-1}$ has real max-entropy at most $k$, then $(\Fc')^{-1}$
		has real max-entropy at most $\max\set{k-\ell+s,0}$ for any $s=\omega(\log n)$.
		
		\item\label{red:paccmax} If $\Fc^{-1}$ has $p$-accessible max-entropy at most $k$, then $(\Fc')^{-1}$
		has $p+2^{-\Omega(s)}$-accessible max-entropy at most $\max\set{k-\ell+s,0}$ for any $s$.
	\end{enumerate}
\end{lemma}

\begin{proof}
	In the following $X$ and $\G$ are uniformly distributed over $\zo^{\fin}$ 
	and $\g$, respectively.  
	\begin{enumerate}
		\item Fix an $x$ such that $\size{\Fc^{-1}(\Fc(x))} \leq 2^k$. By $2$-universal hashing,
		\begin{align*}
			\Exp \bigl[ \size{ \G^{-1}(\G(x)) \cap (\Fc^{-1}(\Fc(x)) \setminus \set{x})} \bigr] \leq
			(2^k-1) \cdot 2^{-\ell} \leq 2^{k-\ell}.
		\end{align*}
		The bound on the real max-entropy
		of $\Fc^{-1}$ and 
		Markov's inequality  
		yield that
		\begin{align*}
			\Pr\bigl[ \size{ \G^{-1}(\G(X)) \cap (\Fc^{-1}(\Fc(X)) \setminus \set{x}) }
			\geq 2^{(s-1)}\cdot 2^{k-\ell} \bigr] \leq
			2^{-(s-1)}+\negl(n).
		\end{align*}
		The bound on the real max-entropy of $(\Fc')^{-1}$ follows.
		
		\item Readily follows from the proof of \cref{lem:redent} part ii.
\end{enumerate}\end{proof}

\subsubsection{Putting everything together}
\begin{proof}[Proof of \cref{thm:TiegToUOHF2}]
	Recall that we start out with a function $\Fc: \zo^{n_0} \mapsto \zo^{m_0}$ 
	with a gap $\Delta$ between real Shannon entropy and accessible Shannon
	entropy. Let $\kreal$ denote the real Shannon entropy of
	$\Fc^{-1}$.

	\begin{description}
		
		\item[{\sc step 1} (gap amplification):] Let $\Fc_1$ be the $t$-fold
		direct product of $\Fc$ for a sufficiently large $t$
		to be determined later.
		That is, $\Fc_1(x_1,\ldots,x_t)=(\Fc(x_1),\ldots,\Fc(x_t))$.
		
		\cref{lem:gapamp} yields that this repetition
		increases both the real and accessible entropies of $\Fc_1$ by a
		factor of $t$.
		In addition, the repetition
		converts real Shannon entropy to real min-entropy and real
		max-entropy (up to an additive $o(t)$ term). More precisely:
		\begin{itemize}
			\item $\Fc_1 \colon \zo^{\oneIn} \mapsto \zo^{\oneOut}$ where $\oneIn(n) = t \cdot\fin$
			and $\oneOut(n) = t \cdot \fout$.
			\item $\Fc_1^{-1}$ has real min-entropy at least $t \cdot \kreal - \zeroIn\sqrt{st}$ and real max-entropy
			at most $t \cdot \kreal + \zeroIn\sqrt{st}$.
			\item $\Fc_1^{-1}$ has accessible Shannon entropy at most $t \cdot \kreal -
			t\Delta$.
		\end{itemize}
		From the next step on, the construction again uses an additional
		parameter $k$. 
		We will be especially interested in the case
		\begin{align}\label{eq:WufRomp:2} 
			k \in [\kreal, \kreal+\Delta^2/128\zeroIn].
		\end{align}
		In case this holds, 
		\begin{itemize}
			\item $\Fc_1^{-1}$ has accessible Shannon entropy at most
			$tk-t\Delta$. \cref{lem:AccHtoAccMax} yields that $\Fc_1^{-1}$
			has $(1-\Delta/4k)$-accessible max-entropy at most $tk-t\Delta/2$.
		\end{itemize}
		
		\item[{\sc step 2} (entropy reduction):]
		Apply entropy reduction to $\Fc_1$ with $\ell = tk-t\Delta/2+s$ to
		obtain $\Fc_2^{(k)}$. That is, $\Fc_2^{(k)}(x,g)=(\Fc_1(x),g,g(x))$, where
		$g\colon\zo^{\oneIn} \mapsto \zo^\ell$ is selected from a family of
		$2$-universal hash functions.
		
		By \cref{lem:redent} and \cref{lem:redentmore}, this reduces the
		accessible max-entropy to $0$, which allows us to deduce that
		$\Fc_2^{(k)}$ is weakly collision-resistant on random inputs.
		Assuming \cref{eq:WufRomp:2} we have
		\begin{itemize}
			\item $\Fc_2^{(k)} \colon \zo^{\twoIn} \mapsto \zo^{\twoOut}$ where $\twoIn(n,k) = O(t\zeroIn + \ell(n,k))=O(t\zeroIn)$
			and $\twoOut(n,k) = O(t\zeroOut + \ell(n,k))=O(t\zeroIn)$.
			\item $(\Fc_2^{(k)})^{-1}$ has real min-entropy at least $t\cdot (\kreal-k+\Delta/2) - \zeroIn\sqrt{st}-2s$,
			which is at least
			$$t \cdot (\Delta/2-\Delta^2/128\zeroIn) - \zeroIn\sqrt{st}-2s$$
			and real max-entropy at most $t\cdot (\kreal-k+\Delta/2) + \zeroIn\sqrt{st} \leq t\cdot \Delta/2 + \zeroIn\sqrt{st}$.
			
			\item $(\Fc_2^{(k)})^{-1}$ has
			$(1-\Delta/4k+2^{-\Omega(s)})$-accessible max-entropy at most
			$0$. Thus, $\Fc_2^{(k)}$ is $q$-collision-resistant on random
			inputs (by \cref{lem:0accmax2CR}), for $q = \Delta/4k-2^{-\Omega(s)}$.
		\end{itemize}
		
		\item[{\sc step 3} (gap amplification):] $\Fc_3^{(k)}$ is $t'$-fold direct
		product of $\Fc_2^{(k)}$, where $t' = s/q = O(ks/\Delta)$.
		That is, $\Fc_3^{(k)}(x_1,\ldots,x_{t'})=(\Fc_2^{(k)}(x_1),\ldots,\Fc_2^{(k)}(x_{t'}))$.
		
		By \cref{lem:gapampmore}, this allows us to amplify the weak
		collision-resistance property of $\Fc_2^{(k)}$ to obtain a gap between
		real min-entropy and accessible max-entropy in $\Fc_3^{(k)}$, again assuming
		\cref{eq:WufRomp:2}.
		\begin{itemize}
			\item $(\Fc_3^{(k)})^{-1}$ has real min-entropy at least
			$$t' \cdot \bigl(t \cdot ({\Delta}/{2}-{\Delta^2}/{128\zeroIn}) -\zeroIn\sqrt{st}-2s\bigr).$$
			\item $(\Fc_3^{(k)})^{-1}$ has accessible max-entropy at most
			$t' \cdot \bigl((1-q/8) \cdot (t\Delta/2 + \zeroIn\sqrt{st}) +1\bigr)$,
			which is at most:
			$$t' \cdot \bigl(t \cdot (\Delta/2-\Delta q/16) + \zeroIn\sqrt{st}) +1\bigr).$$
			Now, $k \leq \zeroIn$, so $q = \Delta/4k-2^{-\Omega(s)}
			\geq \Delta/4\zeroIn-2^{-\Omega(s)}$. This means $(\Fc_3^{(k)})^{-1}$ has
			accessible max-entropy at most:
			$$t' \cdot \bigl(t \cdot (\Delta/2-\Delta^2/64\zeroIn+2^{-\Omega(s)}) + \zeroIn\sqrt{st}) +1\bigr).$$
		\end{itemize}
		Note that the gap is at least
		$t' \cdot \bigl(t \cdot \Delta^2/128\zeroIn - 2^{-\Omega(s)} - (2\zeroIn \sqrt{st} + 2s + 1)\bigr)$,
		which is at least $3s$ as long as:
		$$t \cdot \Delta^2/128\zeroIn \geq 2^{-\Omega(s)} + 2\zeroIn \sqrt{st} + 2s + 1 + 3s/t'$$
		Since $3s/t' = 3q \leq 3\Delta$, we can set $t =
		O(\zeroIn/\Delta + \zeroIn s/\Delta^2 + \zeroIn^4s / \Delta^4) = O(\zeroIn^4 s/\Delta^4)$ so that
		$(\Fc_3^{(k)})^{-1}$ has a gap of $3s$ between real min-entropy and accessible
		max-entropy, and moreover, we know where this gap is (given $k$).
		\item[{\sc step 4:}]
		We follow steps 2, 3, 4, and 5 in the previous construction, with the following
		modifications in the parameters:
		\begin{itemize}
			\item We apply entropy reduction first, with
			$$\ell = t' \cdot \bigl(t \cdot (\Delta/2-\Delta q/16) + \zeroIn\sqrt{st}) +1\bigr)+s.$$
			\item To remove the non-uniform advice $k$, we ``try all
			possibilities'' from $0$ to $\zeroIn$ in steps of size $\Delta^2/128\zeroIn$.
		\end{itemize}
		We then obtain a non-uniform construction of UOWHFs with output and
		key length $O(\zeroIn \cdot t \cdot t') = O(\zeroIn^6s^2/\Delta^5)$, since $t =
		O(\zeroIn^4s/\Delta^4)$ and $t' = O(\zeroIn s/\Delta)$. We also obtain a uniform
		construction with output length $O( \zeroIn/(\Delta^2/\zeroIn) \cdot \zeroIn \cdot t
		\cdot t' \cdot \log n) = O(\zeroIn^8s^2/\Delta^7)$ and key length
		$O(\zeroIn^8s^2/\Delta^7 \cdot \log n)$.
		
	\end{description}
	This finishes the proof of \cref{thm:TiegToUOHF2}.
\end{proof}


\section{Connection to average-case complexity}\label{sec:averagecase}  
In this section we use the notion of inaccessible entropy 
to reprove   
a result by \citet{ImpLev90}, given
in the realm of \emph{average-case complexity}. Our proof follows to a large extent the footstep of  \cite{ImpLev90}, where the  main novelty is formulating  the proof  in the language of inaccessible entropy, and rephrasing  it to make it resembles our  proof of UOWHFs from  one-way functions.

\cref{sec:avgPrelim} introduces the basic notion and definitions used through the section, and in particular what ``success on the average'' means. It also formally describes the result of \citet{ImpLev90}, a result that we reprove in \cref{sec:avg}.

\subsection{Preliminaries and the Impagliazzo and Levin result} \label{sec:avgPrelim}
We start by introducing some basic notions from average-case complexity.\footnote{We limit the following discussion only to notions that we actually use, so this section should not be regarded as an comprehensive introduction to the field of average-case complexity (see \citet{BogTre08} for such an introduction). While the definitions given here are equivalent to the ones given in \cite{BogTre08}, some of them are formulated somewhat differently (the interested reader is welcome to check their equivalence).}

\subsubsection{Algorithms that err} \label{sec:AlgErr}
Let $\Lang$ be some language, and suppose $\Ac(y;\random)$ is a randomized
algorithm with input $y \in \zs$, randomness $\random$, and output domain $\set{0,1,\bot}$.\footnote{We make the less common choice of using $y$ as the input variable (and not $x$), since in our applications the element $y$ is sampled as the output
	of a one-way function.} It is useful to think that $\Ac(y,\cdot)$ is trying to decide $\Lang$, where $0$ and $1$ are guesses
whether $y \in \Lang$, and $\bot$ signals that $\Ac$ refuses to guess.

\begin{definition}\label{def:CorrectAndErr}
	A randomized algorithm $\Ac$ is \emph{$\alpha$-correct} on input $y\in \zs$ \wrt a language $\Lang \subseteq \zs$, if $\Pr[\Ac(y;\Random) = \Lang(y)] \geq \alpha$, where we identify languages with their characteristic functions, and $\Random$ is uniformly distributed over the possible random coins of $\Ac$.
\end{definition}

Recall that when defining a worst case complexity class (\eg $\BPP$), one requires $\Pr[\Ac(y;\Random) = \Lang(y)] \geq \frac23$ for \emph{any} $y \in \zs$ (the choice of the constant $\frac23$ is somewhat arbitrary). In other words, we require that an algorithm is $\frac23-$correct for every input.

In contrast, in average-case complexity an algorithm is allowed to be wrong on some inputs. Specifically, the success probability of a decider $\Ac$ is measured not \wrt a \emph{single} input, but \wrt a \emph{distribution} over the elements of $\zs$.

To make this formal, we first a problem as such a pair of language and distribution family.
\begin{definition}[Problem]\label{def:problem}
	A {\sf problem} is a pair $(\Lang,\cD)$ of language and distribution family, where $\Lang \subseteq \zs$ and $\cD = \set{D_{i}}_{i \in \N}$ is a family of distributions over $\zs$.
\end{definition}
The problem class $\HeurBPP$ (see, \eg \cite[Definition 15]{BogTre08}) contains those problems for which there exists an efficient algorithm that is correct on all but a ``small'' fraction of the inputs.

\begin{definition}[$\HeurBPP$]\label{def:HeurAvgBPP}
	A problem $(\Lang,\cD)$ is in $\HeurBPP$, if there exists a four-input algorithm $\Ac$ such that the following holds for every $(n,\delta) \in \N \times (0,1]$: $\Ac(\cdot,1^n,\delta;\cdot)$ runs in time
	$p(n, 1/\delta)$ for some $p\in \poly$, and
	\begin{align*}
		\Pr_{Y \getsr D_n}[\text{$\Ac(Y,1^n,\delta;\cdot)$ is $\tfrac23$-correct on $Y$}]
		\geq 1-\delta.
	\end{align*}
\end{definition}
In one second we will restrict ourselves to the case that the sampler runs in polynomial time. Then, this means that $\Ac$ is allowed to run in time polynomial in the ``sampling complexity'' of the instance, and inverse polynomial in the probability with which it is allowed to err.

\subsubsection{Samplable distributions}
We next study the families of distributions $\cD$ to consider.
While it is most natural to focus on efficiently samplable distributions,
the definition of $\HeurBPP$ does not pose such limits on the distributions considered; a pair $(\Lang,\cD)$ can be decided efficiently on average even if sampling the distribution $\cD$ is a computationally hard problem. For the reduction, however, we restrict ourselves to polynomial-time samplable distributions. This limitation is crucial, since there exist (not efficiently samplable) distributions $\cD$ with the property that $(L,\cD) \in \HeurBPP$ if and only if $L \in \BPP$ (see \cite{LiVit92} or \cite[Section 2.5]{BogTre08}).
\begin{definition}[Polynomial-time samplable distributions, $(\NP,\polySamp)$]\label{def:polynomialSamplable}
	A distribution family $\cD = \set{D_n}_{n\in \N}$ is {\sf polynomial-time samplable}, denoted $\cD\in \polySamp$, if there exists a polynomial-time computable function $\Dc$ and a polynomial $p(n)$, such that $\Dc(1^n,U_{p(n)})$ is distributed according to $D_n$ for every $n\in \N$, where $U_m$ is the uniform distribution over $m$-bit strings.
	
	The product set $(\NP,\polySamp)$ denotes the set of all pairs $(\Lang,\cD)$ with $\Lang \in \NP$
	and $\cD \in \polySamp$.
\end{definition}

Note that the input $U_{p(n)}$ to $\Dc$ above is the \emph{only} source of randomness
used to sample the elements of $\cD$. In the following we make use the distribution family $\cU = \set{U_n}_{n\in \N}$ (clearly, $\cU \in \polySamp$).

\subsubsection{Impagliazzo and Levin result}
In the above terminology the result of \citet{ImpLev90} can be stated as follows (\cf \cite[Thm.\ 29]{BogTre08}):
\begin{theorem}[\cite{ImpLev90}]\label{thm:impagliazzolevin}
	If $(\NP,\polySamp) \not\subseteq \HeurBPP$, 
	then $\exists \Lang \in \NP$ with
	$(\Lang,\mathcal{U}) \notin \HeurBPP$.
\end{theorem}
In other words, suppose there exists an average-case hard problem whose distribution is polynomial-time samplable,
then there exists an average-case hard problem whose distribution is uniform.

\subsubsection{Search problems}
The proof of \cref{thm:impagliazzolevin} uses the notion of ``\NP-search problems'':

\begin{definition}[Search problems]\label{def:SearchProblem}
	A {\sf search problem} is a pair $(\RelR,\cD)$,
	where $\RelR \subseteq \zs \times \zs$ is a binary relation, and
	$\cD = \set{D_{i}}_{i \in \N}$ is a family of distributions over $\zs$.
	In case $\RelR$ is an $\NP$-relation, then $(\RelR,\cD)$ is an {\sf \NP-search problem}.
\end{definition}

For a relation $\RelR$, let $\RelR_\Lang$ be the corresponding language, \ie $\RelR_\Lang = \{y \colon \exists w\colon (y,w) \in \RelR\}$.

The notion of heuristics is naturally generalized to $\NP$-search problems.
The only change is that in case the algorithm claims $y \in \RelR_\Lang$,
it additionally has to provide a witness to prove that. A search algorithm $\Ac$ always outputs a pair, and we let $\Ac_1$
be the first component of this pair, and $\Ac_2$ be the second component.

\begin{definition}[$\SearchHeurBPP$]\label{def:SearchHeurAvgBPP}
	An \NP-search problem $(\RelR,\cD)$ is in $\SearchHeurBPP$, if there exists an algorithm $\Ac$ outputting pairs in $\set{0,1,\bot} \times \zs$ such that the following holds: (1) $\Ac_1$ is a heuristic for $(\RelR_\Lang,\cD)$ and
	(2) $\Ac_1(y,1^n,\delta;\random)=1
	\implies (y,\Ac_2(y,1^n,\delta;\random)) \in \RelR$.
	
	Algorithm $\Ac$ is called a {\sf (randomized) heuristic search algorithm} for $(\RelR,\cD)$.
\end{definition}

\subsubsection{Search problems vs.\ decision problems}
Suppose $(\Lang,\cD) \in (\NP,\polySamp)$ is a ``difficult decision problem''. Then any $\NP$-relation associated with $\Lang$ gives a ``difficult
\NP-search problem'', because finding a witness also solves the decision problem.

The converse direction is less obvious (recall that even
in worst-case complexity, one invokes self-reducibility).
Nevertheless, \cite{BCGL92}
prove the following (see also \cite[Thm.\ 4.5]{BogTre08}):
\begin{theorem}[\cite{BCGL92}]\label{thm:ben-david}
	Suppose there is an $\NP$-relation $\RelR$ with $(\RelR,\cU)\notin \SearchHeurBPP$, then $\exists \Lang \in \NP$ with $(\Lang, \cU) \notin \HeurBPP$.
\end{theorem}

Using \cref{thm:ben-david} the proof of \cref{thm:impagliazzolevin} proceeds as follows (see \cref{fig:implev}): suppose
there is a pair $(\Lang,\cD) \in (\NP,\polySamp) \setminus \HeurBPP$
and let $\RelR$ be an $\NP$-relation for $\Lang$. Then $(\RelR,\cD) \notin\SearchHeurBPP$ (if $(\RelR,\cD)$ would have a search heuristic algorithm,
the first component of this algorithm, that outputs its left hand side output, would place $(\Lang,\cD) \in \HeurBPP$.)
The following lemma states that in this case there is a pair $(\RelV,\cU) \notin \SearchHeurBPP$, and therefore \cref{thm:ben-david} yields the conclusion.

\begin{figure}
	
	\centering
	\begin{tikzpicture}
		
		\node (LD) at (0,2) [draw] {$(\Lang,\cD) \notin \HeurBPP$};
		\node (LU) at (8,2) [draw] {$(\Lang',\cU) \notin \HeurBPP$};
		
		\node (RD) at (0,0) [draw] {$(\RelR,\cD) \notin \SearchHeurBPP$} ;
		\node (RU) at (8,0) [draw] {$(\RelV,\cU) \notin \SearchHeurBPP$} ;
		
		\draw[->] (RU) -- node [left] {\cref{thm:ben-david}} (LU);
		\draw[->] (LD) -- node [right] {Trivial} (RD);
		\draw[->] (RD) -- node [above] {\cref{lem:impagliazzo-levin-main}} (RU);
		
	\end{tikzpicture}
	\caption{Schematic proof of \cref{thm:impagliazzolevin}.\label{fig:implev}}
\end{figure}

\begin{lemma}[\citet{ImpLev90} main lemma, reproved here]\label{lem:impagliazzo-levin-main}
	Assume that there exists an $\NP$-search problem $(\RelR,\cD)$ with $\cD \in \polySamp$ such that
	$(\RelR,\cD) \notin \SearchHeurBPP$, then there is an $\NP$-relation $\RelV$ such that $(\RelV,\cU) \notin\SearchHeurBPP$.
\end{lemma}
Intuitively \cref{lem:impagliazzo-levin-main} states the following:
suppose some sampling algorithm $\Dc(1^n,\cdot)$ generates hard search problems, then there exist an \NP-search problem that is hard over the uniform distribution.

Consider the following application of \cref{lem:impagliazzo-levin-main}; suppose
that one-way functions exist and let $f\colon \zn \mapsto \zn$
be a length-preserving one-way function. Let $\Dc(1^n,\random)$ be the algorithm that applies
$f\colon \zn \mapsto \zn$ on the input randomness $x = \random$ and outputs
$f(x)$, and set $\cD = \set{D_n}_{n\in\N}$ to the corresponding distribution family.
Furthermore, consider the \NP-search problem given by the relation $\RelR =
\{(f(x),x) \colon x \in \zs\}$.
It is easy to verify that the $\NP$-search problem
$(\RelR,\cD)$ is not in $\SearchHeurBPP$.

\cref{lem:impagliazzo-levin-main} implies that hard problems exist for some uniform distribution, under the assumption that
one-way functions exist.
However, we knew this before: if one-way functions exist, then
UOWHFs also exist. Let $\calF_k$ be such a family as in \cref{def:UOWHF},
and consider the relation $\RelV = \{((z,x),x') \colon F_z(x) = F_z(x'),
x \neq x'\}$,
which asks us to find a non-trivial collision in $F_z$ with a given $x$.
By the security property of UOWHF, if we pick $(z,x)$ uniformly at
random, then this is a hard problem, and it is possible to show that
$(\RelV,\cU) \notin \SearchHeurBPP$.
Thus, Rompel's result gives the conclusion of
\cref{lem:impagliazzo-levin-main} in case we have
the stronger assumption that one-way functions exist.

Given the above, it seems natural to ask whether the strategy used for constructing UOWHF from one-way functions, can be used for proving the general case stated in \cref{lem:impagliazzo-levin-main}. In the following section we answer the above question in 
the affirmative.   
Specifically, we present an alternative proof for \cref{lem:impagliazzo-levin-main} following a similar approach to that taken in the first part of this paper for constructing UOWHF from one-way functions.

\subsubsection{The Valiant-Vazirani lemma}
We make use of the Valiant-Vazirani Lemma, originated in \cite{valiant1986np}.
\begin{lemma}[\cite{valiant1986np}]\label{lem:valiantvazirani}
	Let $\cS \subseteq \zs$ be a set such that $2^k \leq \size{\cS} \leq 2^{k+1}$, and let $\g = \set{g\colon \zs \mapsto \zo^{k+2}}$ be a family of pairwise independent hash-functions. Then for any $x \in \cS$, it holds that $\Pr_{g\getsr \g}[g^{-1}(0^{k+2})=\set{x}] \ge 1/{2^{k+3}}$.
\end{lemma}
A more usual formulation of the Valiant-Vazirani lemma
states that with probability
at least~$\frac18$ there is exactly one element $x \in \cS$
with $g(x) = 0^{k+2}$.
This follows immediately from the above form.
\begin{proof}
	The probability that $g(x) = 0^{k+2}$ for a fixed $x \in \cS$ is
	$\frac{1}{2^{k+2}}$. Conditioned on this event, due to the
	pairwise independence of $\g$, the probability that any other element of
	$\cS$ is mapped to $0^{k+2}$ is at most $\size{\cS} \frac{1}{2^{k+2}} \leq
	\frac{1}{2}$.
\end{proof}

\subsection{Proving \cref{lem:impagliazzo-levin-main} via inaccessible entropy}\label{sec:avg}
Recall the basic idea underlying the two constructions of UOWHF from one-way functions presented in the first part of this paper. In a first step, we use the one-way function $f$ to construct a function $F$ that induces a gap between its real and accessible entropy (\ie $F$ has ``inaccessible entropy''). Roughly, the distribution induced by the output of any efficient ``collision-finder'' algorithm getting a random $x$ and returning a random $x' \in F^{-1}(F(x))$, has a \emph{smaller} entropy than that induced by random preimage of $F(x)$.
Afterwards, we use $F$ to build the UOWHF.

We want to redo this first step in the current setting. Now, however, it is not anymore important to talk about collisions.\footnote{One advantage of using an ``inversion problem'' instead of a ``collision problem'' is that it becomes possible to use two-wise independent hash-functions (instead of three-wise independent).} Thus, we can instead define $F$ such that $F^{-1}(y)$ has some inaccessible entropy for a uniform random $y$. This is in fact compatible with the construction given in \cref{sec:OwfToIAERomp}: it is possible to show that the image of $F$ is close to uniform in case $i \approx \Hall_{f(X)}(f(x))$ (recall that $i$ is the number of bits hashed out from $f(x)$ in the definition of $F$).

Let now $(\RelR,\cD)$ be an $\NP$ search problem with $\cD \in \polySamp$ which is not in $\SearchHeurBPP$.
We would like to use a similar approach as above to define a relation with limited accessible max-entropy.
One might suggest that the following search problem has inaccessible entropy:
given a four tuple $(n,i,g,z)$, where $g$ is a pairwise independent hash-function, and $z$ has $i$ bits, find as solution an input $x$ such that $g(\Dc(1^n,x))_{1,\ldots,i} = z$. However, it turns out that one does not in fact need the randomness inherent in the choice of $z$ (note that a typical pairwise independent hash-function XORs the output with a random string anyhow). Instead, it makes no difference to fix $z = 0^i$, and so we adopt this to simplify the notation, so that the suggested search problem becomes to find $x$ with $g(\Dc(1^n,x))_{1,\ldots,i} = 0^i$ for a given triple $(n,i,g)$.

\paragraph{Problems with the above intuition and postprocessing the witness.}
A moment of thought reveals that there can be cases where this suggested search problem is easy. For example if the sampler $\Dc(1^n,x)$
simply outputs $y=x$ itself, which is possible if finding $w$ with $(y,w) \in \RelR$ is difficult for a uniform random $y$.
The solution is easy: ask the solving algorithm to output also a matching witness $w$ with $(\Dc(1^n,x),w) \in \RelR$ (ignore invalid outputs).

Thus, the suggested search problem becomes: ``given $(n,i,g)$, find $(x,w)$ such that $g(\Dc(1^n,x))_{1\ldots i} = 0^i$ and $(\Dc(1^n,x),w) \in \RelR$''. The hope is then that this search problem has limited accessible entropy in the coordinate corresponding to $x$ (we do not want to talk about the entropy in $w$ because it arise from the number of witnesses which $\RelR$ has, and at this point we have no control over this number).

There is a last little problem to take care of: it is not obvious how to encode $n$ into the search problem, as $(n,i,g)$ does not look like a uniform bitstring of a certain length, even if $i$ and $g$ look random. However, it is possible to ensure that the length of $(i,g)$ uniquely define $n$, and we assume that this is done in such a way that $n$ can be easily computed from the length of $(i,g)$.

\subsubsection{A Relation with bounded accessible average max-entropy}
Using the above discussion, we now finally have enough intuition to define the relation $\RelQ$. For $\cD \in \polySamp$, we let $\Canon(\cD)$ be an arbitrary polynomial-time sampler for $\cD$.

\begin{constructiona}\label{def:npRelationQ}
	Let $\RelR$ be an $\NP$ relation, let $\cD \in \polySamp$, let $\Dc = \Canon(\cD)$ and let $d\in \N$ be such that $\Dc$'s running time on input $(1^n,\cdot)$ is bounded by $n^d$.
	Let $\g$ be an explicit and constructible family of pairwise independent hash functions, where the family $\g_m$ maps all strings of length at most $m$ to strings of length $m$.
	
	For $n\in \N$, define
	\begin{align*}
		\RelQ^{(n)} \eqdef \set{\bigl((i,g),(x,w)\bigr) \colon x \in \zo^{n^d}, i\in [n^d], g\in \g_{n^d},
			g(\Dc(1^n,x))_{1\ldots i} = 0^i, (\Dc(1^n,x),w)\in \RelR},
	\end{align*}
	and let $\RelQ \eqdef \bigcup_{n \in \N} \RelQ^{(n)}$.
\end{constructiona}
Note that the elements of $\g_m$ in \cref{def:npRelationQ} have domain $\bigcup_{i=0}^{m} \zo^{i}$. This somewhat unusual requirement is needed since the sampler might output strings of arbitrary lengths (up to $n^d$).

From now on, we will only consider the case where we have some fixed
sampler $\Dc$ in mind. In this case, whenever $n$ is given, we will
assume that $(i,g)$ are elements satisfying the conditions
in \cref{def:npRelationQ}. Furthermore, we assume \wlg that (the encoding of) a uniform random bitstring, of the right length, induces the uniform distribution on $\g_{n^d} \times [n^d]$. 


\subsubsection{Accessible average max-entropy}
We next define what it means for an $\NP$-search problem to have limited accessible max-entropy, \wrt a part of its witness. This notion is modeled by introducing a function $f$ that outputs the ``interesting part'' of the witness.

\begin{definition}\label{def:AvgAE}
	Let $\RelQ$ be an $\NP$ relation, and $f\colon \zs \mapsto \zs$ a function.
	For $y \in \zs$ let
	\begin{align*}
		\cS_{\RelQ,f}(y) \eqdef \set{ f(w) \colon (y,w) \in \RelQ}.
	\end{align*}
	The {\sf real average max-entropy of $(\RelQ,\cD)$ \wrt $f$}, is the function
	\begin{align*}
		\Hreal_{\RelQ,\cD,f}(m) = \Exp_{Y \getsr D_m}[\log(\size{\cS_{\RelQ,f}(Y)})]
	\end{align*}
	letting $\log(0) \eqdef -1$.\footnote{The convention $\log(0) = -1$ helps us to simplify notation in a few places.
		(We need that $\log(0) < \log(1)$ since an algorithm which produces
		no valid witness should produce less entropy than an algorithm which produces some valid witness).}
\end{definition}
In case the relation $\RelR$ and $f$ are clear from the context, we sometimes write $\cS(y)$ instead of $\cS_{\RelQ,f}(y)$.

We next define a useful notion of limited accessible max-entropy in
this setting.
Here, one should think of algorithm $\Ac$ as an algorithm which, on input $y$ produces a witness $w$ with $(y,w) \in \RelQ$.
It furthermore ``aims'' to produce witnesses $w$ for which $f(w)$ has as much entropy as possible.

\begin{definition}\label{def:avgCaseAccessibleEntropy}
	Let $f\colon \zs \mapsto \zs$, $p$, $k\colon \N \times (0,1] \mapsto \R$,
	$\RelQ$ an $\NP$-relation, and $\Dc = \set{D_n}_{n\in\N}$ a family of distributions.
	The pair $(\RelQ,\cD)$ has {\sf \ioS (infinitely often) $\rho$-accessible average
		max-entropy at most $k$ \wrt $f$},
	if for every four-input algorithm $\Ac(y,1^{\secParQ},\qDelta;r)$ running in time $\ell =\ell(\secParQ,1/\qDelta)$ for some $\ell\in \poly$, there exists infinitely many $\secParQ$'s in $\N$, a function $\qDelta(\secParQ) = \qDelta\in (0,1]$ and an ensemble of set families $\set{\set{\cL_{\secParQ}(y)\subseteq \zs}_{y\in \Supp(D_\secParQ)}}_{\secParQ\in \N}$, such that
	\begin{align}\label{eq:Avg:2}
		\Pr_{Y\getsr D_{\secParQ}, \Random \getsr \zl}
		\bigl[\Gamma(Y, \Ac(Y,1^{\secParQ},\qDelta;\Random)) \in \bigl(\cL_{\secParQ}(Y) \cup \set{\perp}\bigr)\bigr]
		&\geq 1-\rho(\secParQ,\qDelta)
	\end{align}
	where $\Gamma = \Gamma_{\RelQ,f}(y,w)$ equals $f(w)$ in case $(y,w) \in \RelQ$ and equals $\bot$ otherwise, and
	\begin{align}\label{eq:Avg:3}
		\Exp_{Y\getsr D_\secParQ}\bigl[\log(\size{\cL_{\secParQ}(Y)})\bigr] &\leq k(\secParQ,\qDelta)
	\end{align}
\end{definition}

The following lemma, proven in \cref{sec:avgPartI}, states that the relation $\RelQ$ defined in \cref{def:npRelationQ}
has limited accessible max-entropy \wrt the function $(x,w) \mapsto x$.
\begin{lemma}\label{lem:impagliazzolevin-part1N}
	Let $(\RelR,\cD)$ be an $\NP$ relation with $\cD \in \polySamp$ and $(\RelR,\cD) \notin \SearchHeurBPP$. Define $\RelQ$ from $(\RelR,\cD)$ as in \cref{def:npRelationQ}, and let $f\colon \zs \times \zs \mapsto \zs$ be given by $f(x,w) = x$.
	Then, for any $c \in \N$,
	$(\RelQ,\cU)$ has \ioS $(\frac{\qDelta}{\secParQ})^c$-accessible average
	max-entropy at most $k(\secParQ,\qDelta) = \Hreal_{\RelQ,\cU,f}(\secParQ) - \sqrt[c]{\qDelta}\cdot \secParQ^c$ \wrt $f$.
\end{lemma}
\cref{lem:impagliazzolevin-part1N} is proven below, and in \cref{sec:avgPartII} we use \cref{lem:impagliazzolevin-part1N} for proving \cref{lem:impagliazzo-levin-main}. The latter is done by additionally fixing the value of $h(x)_{1\ldots j}$, where $h$ is an additional random hash function, and $j$ is a random integer (in a certain range). The ratio is that an algorithm producing $(x,w)$ with $h(x)_{1\ldots j} = 0^j$, can be used to access max-entropy roughly $2^j$.

\subsubsection{Proving \cref{lem:impagliazzolevin-part1N}}\label{sec:avgPartI}
The proof of \cref{lem:impagliazzolevin-part1N} follows similar lines to the proof of \cref{thm:OwfToIAERomp}.

\begin{proof}[Proof (of \cref{lem:impagliazzolevin-part1N})]
	Let $\Ac$ be an algorithm that ``aims to produce max-entropy for $\RelQ$''.
	Without loss of generality, we assume that $\Ac$ either outputs a valid witness $(x,w)$ for a given input $(i,g)$
	or $\bot$.
	We show how to find infinitely many $\secParQ \in \N$, $\qDelta = \qDelta(\secParQ) \in (0,1]$, and ensemble of set families $\set{\cL_\secParQ = \set{\cL_\secParQ(i,g)}_{(i,g) \in \RelQ_\Lang}}_{\secParQ\in \N}$ with the properties as required in the lemma (we write $\cL_\secParQ(i,g)$ instead of $\cL_\secParQ(y)$ because the elements of
	$\RelQ_\Lang$ are pairs $(i,g)$). Towards achieving the above, consider the following candidate algorithm $\Bc$ for a search heuristics for $(\RelR,\cD)$.
	
	Let $\beta \in \N$ be a constant to be determined by the analysis, let $d \in \N$ be such that $\secParR^d$ is an upper bound on the runtime of the sampler $\Dc$ (recall that we have fixed $\Dc$ above) on input $(1^\secParR,\cdot)$, and let $\ell = \ell(\secParQ,\qDelta)$ be an upper bound on the running time of $\Ac$ on parameters $\secParQ$ and $\qDelta$. Let $\secParQ (\secParR)$ be the description length of a pair in $[\secParR^d] \times \g_{\secParR^d}$ and let $\qDelta(\secParR,\rDelta) = (\rDelta(\secParR)/\secParR^\beta)^\beta$.

	\bigskip
	\noindent\framebox{
		\begin{minipage}{16cm}
			\noindent \textbf{Search Heuristics $\Bc^{\Ac}$}
			\medskip
			
			\hrule
			
			\medskip
			
			\textbf{Oracle:} $\Ac$\qquad\qquad /\!\!/ Entropy generator for $(\RelQ,\cU)$.\\
			\textbf{Input:} $\ys \in \zs$, $1^{\secParR}$ and $\rDelta\in (0,1]$.
			\medskip\hrule\medskip
			
			$\qDelta = \qDelta(\secParR,\rDelta)$; $\secParQ = \secParQ(\secParR)$;
			$\ell = \ell(\secParQ,\qDelta)$\\
			\textbf{repeat} $\secParR\cdot(\secParR/\rDelta)^{\beta}$ \textbf{times:}\\
			\mbox\qquad $i \getsr [\secParR^d]$\\
			\mbox\qquad $g \getsr \set{g' \in \g_{\secParR^d} \colon g'(\ys)_{1\dots i} = 0^i}$\\
			\mbox\qquad $\random \getsr \zl$\\
			\mbox\qquad $(x,w) = \Ac((i,g),1^\secParQ,\qDelta; \random)$\\
			\mbox\qquad \textbf{if} $(\ys,w) \in \RelR $ \textbf{and} $\Dc(1^\secParR,x)=\ys$\\
			\mbox\qquad\qquad \textbf{return} $(1,w)$\\
			\textbf{return} $0$
		\end{minipage}
	}
	\bigskip
	
	It is clear that, for any $\beta \in \N$, the running time of $\Bc$ is $\poly(\secParR,1/\rDelta)$. Recall that, by assumption, $(\RelR, \cD)$ has no search heuristics. Algorithm $\Bc$ satisfies item (2) of \cref{def:SearchHeurAvgBPP},
	and since $(\RelR,\cD)$ has no search heuristics, it must fail
	to satisfy item (1) (\ie $\Bc_{1}$ cannot be a randomized heuristic). Note now that algorithm $\Bc_1$ is always correct if $\ys \notin \Lang$ (always outputs $0$). Thus, $\Bc$ must fail to be correct on some $\ys \in \Lang$. In fact, there exist infinitely many $\secParR$'s in $\N$ and a function $\rDelta = \rDelta(\secParR) \in (0,1]$. such that $\Bc_1(y,1^\secParR,\rDelta;\cdot)$ is not $\frac{2}{3}$-correct
	for more than a fraction $\rDelta$ of the inputs $y \in \RelR_\Lang$
	produced by $\Dc(1^{\secParR},\cdot)$.

	The following discussion is \wrt any fixed pair $(\secParR,\rDelta = \rDelta(\secParR))$ from the above infinite set.
	
	We present a family of sets $\set{\cL_{\secParQ}(i,g)}_{(i,g)\in\g_{\secParR^d} \times [\secParR^d]}$ for which \cref{eq:Avg:2,eq:Avg:3} holds \wrt algorithm $\Ac$ and $f$, for the parameters $\secParQ = \secParQ(\secParR)$, $\qDelta = \qDelta(\secParR,\rDelta) = \qDelta(\secParQ)$, and $\rho$ and $k$ as stated in the lemma. Since this holds for \emph{any} such pair $(\secParR,\rDelta)$ and since $\secParQ(\secParR) \in \Omega(\secParR)$ (and thus, there are infinitely many different $\secParQ$'s) the proof of the lemma would follow.
	
	Consider the following set
	\begin{align}
		\YB = \set{ \ys \in \Lang \colon \Pr_{{I\getsr [\secParR^d]}\atop{G \getsr \g_{\secParR^d},\Random \getsr \zl}}[\Ac_1((I,G),1^\secParQ,\qDelta;\Random)
			\in \Dc^{-1}(1^\secParR,\ys) \colon G(\ys)_{1 \ldots I} = 0^I] <
			\Bigl(\frac{\qDelta}{\secParR}\Bigr)^\beta},
	\end{align}
	letting $\Dc^{-1}(1^\secParR,\ys)\eqdef \set{x \in \zo^{\secParR^d} \colon \Dc(1^\secParR,x) = \ys}$ and let $\ell = \ell (\secParQ)$. Note that $\YB$ contains all the $y$'s in $\Lang$ for which $\Bc_1$ is not $\frac23$-correct, and the above discussion implies
	\begin{align}\label{eq:calYisLarge}
		\Pr_{Y \getsr D_{\secParR}}[Y \in \YB] > \rDelta
	\end{align}
	Towards defining the sets $\set{\cL_\secParQ(i,g)}$, we partition the preimages of the elements in $\YB$ into buckets; for $i \in \set{0,\ldots,\secParR^{d}-1}$ let 
	\begin{align}
		\cLv(i) \eqdef \set{x \in \zo^{\secParR^d} \colon \Dc(1^\secParR,x) \in \left(\YB
			\cap \set{y \colon \Hall_{D_\secParR(y)}(y) \in [i,i+1)}\right)},
	\end{align}
	where $\Hall_{D_\secParR(y)}(y) = -\log(1/D_\secParR(y))$ is the sample entropy of $y$ \wrt the distribution $D_{\secParR}$.
	In words: $\cLv(i)$ are those $x$ for which $y=\Dc(1^\secParR,x) \in \Lang$
	is an element for which $\Bc(y)$ is unlikely to produce a witness,
	and for which $y$ has roughly $2^{(\secParR^d)-i}$ preimages.
	
	For $(i,g)\in\g_{\secParR^d}$, the set $\cL_\secParQ(i,g)$ is defined as
	\begin{align}
		\cL_\secParQ(i,g) \eqdef \cS(i,g) \setminus \cLv(i),
	\end{align}
	where $\cS(i,g)$ is taken from \cref{def:AvgAE}. In the remaining of the proof we show that, for the right choice of $\beta$, \cref{eq:Avg:2,eq:Avg:3} holds \wrt the family $\set{\cL_\secParQ(i,g)}$ for the functions $\rho(\secParQ,\qDelta) = (\frac{\qDelta}{\secParQ})^c$ and $k(\secParQ,\qDelta) = \Hreal_{\RelQ,\cU,f}(\secParQ) - \sqrt[c]{\qDelta}\cdot \secParQ^c$. The proof easily follow by the next two claims.
	\begin{claima}\label{claim:avgHardPreimages}
		We have
		\begin{align*}
			\Pr_{\MyAtop{I\getsr [\secParR^d], G \getsr \g_{\secParR^d}}{\Random \getsr \zl}}[\Ac_1((I,G),1^\secParQ,\qDelta;\Random) \in \cLv(I)]
			\leq 2\secParR^d \cdot \Bigl(\frac{\qDelta}{\secParR}\Bigr)^\beta
		\end{align*}
	\end{claima}
	
	\begin{claima}\label{claim:avgManyHardPreimages}
		For $i\in [\secParR^d]$ and $g\in \g_{\secParR^d}$, let $\tilde\cS(i,g)= \cS(i,g)$ in case this set in non-empty and $\tilde\cS(i,g) = \set{\perp}$ otherwise, then
		\begin{align*}
			\Pr_{\MyAtop{I\getsr [\secParR^d], G \getsr \g_{\secParR^d}}{X \la\tilde\cS(I,G) }}[X \in \cLv(I)]
			\geq \frac{\rDelta}{10\secParR^d}
		\end{align*}
	\end{claima}
	
	Before proving the above claims, we first use them to conclude the proof of the lemma.

	\cref{claim:avgHardPreimages} implies that for large enough $\beta$
	\begin{align}\label{eq:avgHardPreimages}
		\Pr_{\MyAtop{I\getsr [\secParR^d], G \getsr \g_{\secParR^d}}{\Random \getsr \zl}}[\Gamma\bigl((I,G),\Ac((I,G),1^\secParQ,\qDelta;\Random)\bigr) \in \left(\cL_\secParQ(I,G) \cup \set{\perp}\right)] &\geq 1 - 2\secParR^d \cdot \Bigl(\frac{\qDelta}{\secParR}\Bigr)^\beta\geq 1 - (\frac{\qDelta}{\secParQ})^c
	\end{align}
	yielding that \cref{eq:Avg:2} holds for $\set{\cL_\secParQ(i,g)}$ and $\rho$.

	Applying Markov's inequality on \cref{claim:avgManyHardPreimages}, yields that
	\begin{align}
		\Pr_{X \getsr \cS(i,g)}[\Dc(1^\secParR,X) \in \cLv(i)]
		\geq \frac{\rDelta}{20\secParR^d}
	\end{align}
	for at least $\frac{\rDelta}{20\secParR^d}$ fraction of the pairs $(i,g) \in [\secParR^d]\times \g_{\secParR^d}$. Where for any such pair it holds that
	\begin{align}
		\log(\size{\cL_\secParQ(i,g)})&\leq \log((1-\frac{\rDelta}{20\secParR^d})\cdot \size{\cS(i,g)})\\
		&\leq \log(\size{\cS(i,g)}) - \frac{\rDelta}{20\secParR^d}.\nonumber
	\end{align}
	It follows that for large enough $\beta$
	\begin{align*}
		\Hreal_{\RelQ,\cU,f}(\secParQ) - \Exp[\log(\size{\cL_\secParQ(I,G)})] &=\Exp_{(I,G) \getsr [\secParR^d]\times \g_{\secParR^d}}[\log(\size{\cS(I,G)})-\log(\size{\cL(\secParQ,I,G)})]\\
		&\geq \frac{\rDelta^2}{400\secParR^{2d}} = \frac{(\sqrt[\beta]{\qDelta}\cdot \secParQ^\beta)^2}{400\secParR^{2d}}\\
		& \geq \sqrt[c]{\qDelta}\cdot \secParQ^c
	\end{align*}
	Hence, \cref{eq:Avg:3} holds for $\set{\cL_\secParQ(i,g)}$ and $k$, and the proof of the lemma follows.
\end{proof}

\begin{proof}[Proof of \cref{claim:avgHardPreimages}]
	Compute
	\begin{align}
		\secParR^d \cdot \Bigl(\frac{\qDelta}{\secParR}\Bigr)^\beta &\geq \sum_{y \in \YB} D_\secParR(y)\cdot \secParR^d \cdot \Bigl(\frac{\qDelta}{\secParR}\Bigr)^\beta\\
		&\geq \sum_{i=1}^{\secParR^d}\sum_{y\in \Dc(1^\secParR,\cLv(i))} D_\secParR(y)\cdot \Pr_{G\getsr \g_{\secParR^d},\Random \getsr \zl}[\Ac_1((i,G),1^\secParQ,\qDelta;\Random) \in \Dc^{-1}(1^\secParR,y) \mid G(y)_{1 \ldots i} = 0^i],\nonumber
	\end{align}
	letting $\Dc(1^\secParR,\cLv(i)) \eqdef \set{\Dc(1^\secParR,x)\colon x\in \cLv(i))}$. In addition, for any $(i,\random)$ it holds that
	\begin{align*}
		\sum_{y\in \Dc(1^\secParR,\cLv(i))}& D_\secParR(y)\cdot \Pr_{G\getsr \g_{\secParR^d}}[\Ac_1((i,G),1^\secParQ,\qDelta;\random) \in \Dc^{-1}(1^\secParR,y) \mid G(y)_{1\ldots i} = 0^i]\\
		&=\sum_{y\in \Dc(1^\secParR,\cLv(i))} D_\secParR(y)\cdot 2^i\cdot \Pr_{G\getsr \g_{\secParR^d}}[\Ac_1((i,G),1^\secParQ,\qDelta;\random) \in \Dc^{-1}(1^\secParR,y) \land G(y)_{1\ldots i} = 0^i]\\
		&=\sum_{y\in \Dc(1^\secParR,\cLv(i))} D_\secParR(y) \cdot 2^i \cdot \Pr_{G\getsr \g_{\secParR^d}}[\Ac_1((i,G),1^\secParQ,\qDelta;\random)) \in \Dc^{-1}(1^\secParR,y)]\\
		&\geq\frac12\cdot \sum_{y\in \Dc(1^\secParR,\cLv(i))} \Pr_{G\getsr \g_{\secParR^d}}[\Ac_1((i,G),1^\secParQ,\qDelta;\random)) \in \Dc^{-1}(1^\secParR,y)]\\
		&=\frac12\cdot \Pr_{G\getsr \g_{\secParR^d}}[\Ac_1((i,G),1^\secParQ,\qDelta;\random)) \in \cLv(i)].
	\end{align*}
	Collecting the equations yields the claim.
\end{proof}
For the proof of \cref{claim:avgManyHardPreimages}, we need first a pairwise independence analogue of \cref{claim:almost-uniform}. The proof is exactly the same, except a bit simpler as we fix the output $w$ instead of fixing another preimage. We provide it for completeness.
\begin{claima}\label{claim:two-universal-picking}
	Let $i \in [{\secParR^d}]$, $w \in \zo^i$ and $x^\ast \in \zo^{\secParR^d}$ be such that
	$\HSh_{f(X)}(f(x^\ast)) \geq i$. Then,
	\begin{align*}
		\ppr{{g\getsr \g_{n^d}}\atop{x \getsr (g \circ f)^{-1}(w)}}{x = x^\ast} \geq \frac{2^{-{\secParR^d}}}{10},
	\end{align*}
	letting  $(g \circ f)^{-1}(w)$  equals the set
	$\set{x \colon g(f(x))_{1\ldots i} = w}$ in case this
	set is not empty, and $\set{\bot}$ otherwise.
\end{claima}
\begin{proof}
	Let $G$ be uniformly distributed over $\g_{n^d}$, and let $E$ be the event that $G(f(x^\ast)) = w$. Note that $\Pr[E] = 2^{-i}$ and that
	\begin{align}
		\ppr{x \getsr (G \circ f)^{-1}(w)}
		{x = x^\ast \mid E} = \frac{1}{\size{(G \circ f)^{-1}(w)}}
	\end{align}
	The pairwise independence of $\g_{n^d}$ yields that $\Pr[G(f(x)) = w \mid E]= 2^{-i}$  for any $x \in \zo^{\leq\secParR^d} \setminus  f^{-1}(f(x^\ast))$. Hence, $\ex{\size{(G \circ f)^{-1}(w) \setminus f^{-1}(f(x^\ast))} \mid E} \leq 2\cdot 2^{-i+{\secParR^d}}$
	and by Markov's inequality
	\begin{align}
		\pr{\size{(G \circ f)^{-1}(w) \setminus f^{-1}(f(x^\ast))} \leq 4\cdot 2^{-i+{\secParR^d}}  \mid E}\geq \frac12
	\end{align}
	Combining the above inequality and the assumption  $\HSh_{f(X)}(f(x^\ast)) \geq i$, we get
	\begin{align}
		\pr{\size{(G \circ f)^{-1}(w)} + |f^{-1}(f(x^\ast)| \leq 5\cdot 2^{-i+{\secParR^d}}  \mid E}\geq \frac12,
	\end{align}
	and  conclude that
	\begin{align*}
		\ppr{x \getsr (G \circ f)^{-1}(w)}{x = x^\ast} &= \Pr[E] \cdot \ppr{x \getsr (G \circ f)^{-1}(w)}{x = x^\ast \mid E}\\
		& \geq 2^{-i} \cdot \frac12 \cdot \frac{1}{5\cdot 2^{-i+{\secParR^d}}}\\
		& = \frac{2^{-{\secParR^d}}}{10}.
	\end{align*}
\end{proof}

\begin{proof}[Proof of \cref{claim:avgManyHardPreimages}]
	For $i\in \set{0\dots,\secParR^d}$ and $x \in \cLv(i)$, \cref{claim:two-universal-picking} yields that
	\begin{align}
		\Pr_{G\getsr \g_{\secParR^d}, X\getsr \tilde\cS(i,G)}[X = x]\geq \frac{2^{-\secParR^{d}}}{10}
	\end{align}
	By \cref{eq:calYisLarge} it holds that $\Pr_{Y \getsr D_\secParR}[Y \in \YB] > \rDelta$, and therefore $\Pr_{X\getsr \zo^{\secParR^d}}[X \in \bigcup_{i=1}^{\secParR^d} \cLv(i)] = \Pr_{Y \getsr D_\secParR}[Y \in \YB] > \rDelta$. We conclude that
	\begin{align*}
		\Pr_{\MyAtop{I\getsr [\secParR^d], G \getsr \g_{\secParR^d}}{X \la\tilde\cS(I,G) }}[X \in \cLv(I)] &\geq \Ex\left[\size{ \cLv(I)}\right] \cdot \frac{2^{-\secParR^{d}}}{10}\\
		&\geq \frac{\rDelta \cdot 2^{\secParR^d}}{\secParR^d} \cdot \frac{2^{-\secParR^{d}}}{10} = \frac{\rDelta}{10\secParR^d}.
	\end{align*}
\end{proof}

\subsubsection{A difficult problem for the uniform distribution}\label{sec:avgPartII}
In this section we show how to transform a uniform search problem with a gap between its real and accessible entropy, into a uniform search problem for which no heuristic search algorithm exists (\ie the problem is not in $\SearchHeurBPP$). Combining it with \cref{lem:impagliazzolevin-part1N} concludes the proof of \cref{lem:impagliazzo-levin-main}.

The transformation is achieved by adding additional restriction on the witness of the given search problem. Specifically, requiring its ``hash value'' \wrt a randomly chosen pairwise independent hash function to be the all zero string.

We use the following construction:
\begin{constructiona}\label{def:rPrime}
	Let $\RelQ$ be an $\NP$-relation, let $f\colon \zs \mapsto \zs$ be a function, let $\g = \set{\g_m}$ be family of pairwise independent hash functions, where the functions of $\g_k$ map strings of length at most $k$ to string of length $m$ (as in \cref{def:npRelationQ}), and let $d\in \N$ be such that $(y,w) \in \RelQ \implies
	\size{f(w)} \leq \size{y}^d$. For $n\in \N$ let
	\begin{align*}
		\RelV^{(n)} \eqdef \set{
			\bigl((y,j,g),w)\colon y \in \zn, j \in [n^d+2], g \in \g_{n^d+2}, (y,w) \in \RelQ, g(f(w))_{1\ldots j} = 0^j}
	\end{align*}
	and let $\RelV \eqdef \bigcup_{n \in \N} \RelV^{(n)}$.
\end{constructiona}

As in \cref{def:rPrime}, we assume that the tuples $(y,j,g)$'s above can be encoded such that a uniformly random string, of the right length, decodes to a uniformly random tuple in $\zn \times [n^d+2] \times \g_{n^d+2}$.

\begin{lemma}\label{lem:impagliazzolevin-part2}
	Let $\RelQ$, $f$, $d$ and $\RelV$ be as in \cref{def:rPrime}. Suppose that $(\RelQ,\cU)$ has \ioS $\frac{\qDelta^2}{50\secParQ^d}$-accessible average max-entropy at most $\Hreal_{\RelQ,\cU,f}(\secParQ) - 5\qDelta \secParQ^d$ \wrt $f$, then $(\RelV,\cU) \notin \SearchHeurBPP$.
\end{lemma}
\begin{proof}
	We assume towards a contradiction that $(\RelV,\cU) \in \SearchHeurBPP$, and show that $(\RelQ,\cU)$ has too high accessible average max-entropy.
	
	Let $\Ac$ be a randomized search heuristics for $(\RelV,\cU)$. The following algorithm $\Bc$ contradicts the assumption that $(\RelQ,\cU)$ has \ioS $\frac{\qDelta^2}{50 \secParQ^d}$-accessible average max-entropy at most $\Hreal_{\RelQ,\cU,f}(\secParQ) - 5 \qDelta \secParQ^d$ \wrt $f$.

	Let $\ell = \ell(\secParV,\vDelta)$ be an upper bound on the running time of $\Ac$ on parameters $\secParV$ and $\vDelta$. Let $\secParV (\secParQ)$ be the description length of a triplet in $\zo^\secParQ \times [\secParQ^d+2] \times \g_{\secParQ^d+2}$ and let $\vDelta(\secParQ,\qDelta) = \frac{\qDelta^2}{100\secParQ^d}$.

	\bigskip
	\noindent\framebox{
		\begin{minipage}{16cm}
			\noindent \textbf{Entropy generator $\Bc^{\Ac}$ for $(\RelQ,\cU)$}
			\medskip
			
			\hrule
			
			\medskip
			
			\textbf{Oracle:} $\Ac$\qquad\qquad /\!\!/
			Search heuristics for $(\RelV,\cU)$.
			
			\textbf{Input:} $y\in \zo^\secParQ$, $1^\secParQ$ and $\qDelta\in (0,1]$
			
			\medskip\hrule\medskip
			
			$\vDelta = \vDelta (\secParQ,\qDelta)$; $\secParV = \secParV(\secParQ)$; $\ell = \ell(\secParV,\vDelta)$\\
			$j \getsr \set{2,\ldots,\secParQ^d+2}$\\
			$g \getsr \g_{\secParQ^d+2}$\\
			$r \getsr \zl$\\
			$(b,w) = \Ac(y,j,g,1^\secParV,\vDelta;\random)$\\
			\textbf{if} $b = 1$ \textbf{return} $w$, \textbf{else} \textbf{return} $\bot$
		\end{minipage}
	}
	\bigskip
	
	It is clear that the running time of $\Bc$ is $\poly(m,1/\qDelta)$. We show that $\Bc$ achieves high max-entropy for all (except maybe finitely many) values of $\secParQ$ and any $\qDelta$. Specifically, that for all (except maybe finitely many) $(\secParQ,\qDelta)$, there exists no family $\set{\cL_\secParQ(y)}_{y\in \zo^\secParQ}$ as described in \cref{def:avgCaseAccessibleEntropy}.
	
	Fix $\secParQ$ and $\qDelta$, and let the random variables $Y$, $J$, $G$ and $\Random$ be uniformly chosen from $\zo^\secParQ \times [\secParQ^d+2]\times \g_{\secParQ^{d}+2} \times \zl$. For $y\in \zo^\secParQ$, let $\eta(y)$ be the probability that $\Ac(y,J,G,1^\secParV,\vDelta;\cdot)$ is not $\frac23$-correct. Since $\Exp[\eta(Y)] \leq \vDelta = \frac{\qDelta^2}{100\secParQ^d}$, it holds that
	\begin{align}\label{eq:etaY}
		\Pr\Bigl[\eta(Y) \leq \frac{\qDelta}{8\secParQ^d}\Bigr] \geq 1-\qDelta
	\end{align}
	Fix $y\in \zo^\secParQ$ and let $\cS(y)= \cS_{\RelQ,f}(y) = \set{ f(w) \colon (y,w) \in \RelQ}$. Intuitively, since $f(w)$ is the ``interesting part'' of a $\RelQ$-witness $w$ for $y$, $\cS(y)$ is the set of interesting parts of witnesses of $y$. 
	
	For $x \in \cS(y)$ let $\cE(y,x) \subseteq [\secParQ^d+2]\times \g_{\secParQ^{d}+2}$ be the set of pairs $(j,g)$ for which $x$ is the only
	element in $\cS(y)$ with $g(x)_{1\ldots j} = 0^j$. Going back to our intuition again, if $x$ is the ``interesting part'' of a $\RelQ$-witness for $y$, then $\cE(y,x)$ is the set of pairs $(j,g)$, such that on input $(y,j,g)$, algorithm $A$ \emph{must} output $w$ with $f(w)=x$ to be successful.
	
	By \cref{lem:valiantvazirani}, for every fixed pair $(y,x)$ with $x \in \cS(y)$:
	\begin{align}\label{eq:Eyvv}
		\Pr[(J,G)\in \cE(y,x)] \geq \frac{1}{8|\cS|\secParQ^d}
	\end{align}
	The following argument would become simpler if \cref{eq:Eyvv} would hold with equality, or if at least $\Pr[(J,G) \in \cE(y,x)]$ was the same for every $x$ with $x \in \cS(y)$. This may of course not be the case, but since all pairs $(j,g)$ have the same probability we can simply discard elements from the sets $\cE(y,x)$ until they all have the same size. We call the resulting sets $\cE'(y,x)$. Hence we can assume that $\cE'(y,x) \subseteq \cE(y,x)$ and
	$\Pr[(J,G)\in \cE'(y,x_1)] = \Pr[(J,G)\in \cE'(y,x_2)]
	\geq \frac{1}{8|\cS|\secParQ^d}$ for all $x_1,x_2 \in \cS(y)$.
	
	We next let $\cE'(y) = \bigcup_{x\in \cS} \cE'(y,x)$. Intuitively, $\cE'(y)$ are those pairs $(j,g)$ such that $A$ on input $(y,j,g)$ is forced to answer with some unique $x$. 
	Let $x_y(g,j)$ to be the unique element of $\cS(y)$ with $g(x)_{1\ldots j} = 0^j$ in case it exists, and $\bot$ otherwise, and let $X_y= x_y(G,J)$. Because of the above trickery with $\cE'$, conditioned on $(J,G) \in \cE'(y)$, the random variable $X_y$ is uniformly distributed over $\cS(y)$. Furthermore, since $\Pr[(J,G) \in \cE'(y)] \geq \frac{1}{8\secParQ^d}$, it holds that
	\begin{align}\label{eq:Ey2}
		\lefteqn{\Pr[\text{$\Ac(y,J,G,1^\secParV,\vDelta;\cdot)$ is \emph{not} $\tfrac23$-correct} \mid (J,G) \in \cE'(y)]} \\
		&\leq 8\secParQ^d\cdot \Pr[\text{$\Ac(y,J,G,1^\secParV,\vDelta;\cdot)$ is \emph{not} $\tfrac23$-correct $ \land \ (J,G) \in \cE'(y)$}]\nonumber\\
		&\leq 8\secParQ^d \cdot \eta(y). \nonumber
	\end{align}
	
	By definition, 
	\begin{align}
		\Pr[f(\Ac_2(y,j,g,1^\secParV,\vDelta;R)) =x_y(j,g)] \geq \frac23
	\end{align}
	for every $(j,g) \in \cE'(y)$ such that $\Ac(y,j,g,1^\secParV,\vDelta;\cdot)$ is $\frac23$-correct.
	
	For each $(j,g)\in \cE'(y)$, we now construct a set $\cR(j,g)$ of random
	strings as follows: first pick all random strings $r$
	for which $f(\Ac_2(y,j,g,1^\secParV,\vDelta;r)) =x_y(j,g)$ is satisfied. Afterwards, if necessary, add other random strings or discard some of the picked random strings such that $|\cR(j,g)| = \frac23 2^\ell$.
	
	Note that
	\begin{align}\label{eq:Av:3.5}
		\lefteqn{\Pr[f(\Ac_2(y,J,G,1^\secParV,\vDelta;R)) = X_y \mid (J,G) \in \cE'(y) \land R \in \cR'(J,G)]}\\
		&\geq \Pr[\text{$\Ac(y,J,G,1^\secParV,\vDelta;\cdot)$ is $\tfrac23$-correct.} \mid (J,G) \in \cE'(y)]\nonumber\\
		&\geq 1-8\secParQ^d \eta(y)\nonumber.
	\end{align}
	Hence, \cref{eq:etaY} yields that
	\begin{align}\label{eq:Av:4}
		\Pr[f(\Ac_2(y,J,G,1^n,\vDelta;R)) = X_y \mid (J,G) \in \cE'(y) \land R \in \cR'(J,G)]
		\geq 1- \qDelta
	\end{align}
	for a $(1-\qDelta)$ fraction of the $y$'s in $\zo^\secParQ$.

	It remains to show that no family of sets $\set{\cL(y)}_{y\in \zo^\secParQ}$ can be used to show that $\Ac$ has
	$\frac{\qDelta^2}{50 \secParQ^d}$-accessible max-entropy at most $\Exp[\log(\size{\cS(Y)})] - 5 \qDelta \secParQ^d$. Fix a family $\set{\cL(y)}$ with $\Exp[\log(\size{\cL(Y)})] \leq \Exp[\log(\size{\cS(Y)})] - 5 \qDelta \secParQ^d$. The following claim concludes the proof of the lemma.
	\begin{claima}\label{claim:Avg:2Ins}
		It holds that $\pr{\Gamma(Y, \Ac(Y,1^n,\qDelta;\Random)) \notin \cL(Y) \cup \set{\perp}}\geq \frac{\qDelta^2}{50 \secParQ^d}$.
	\end{claima}
	\begin{proof}
		Note that $\log(\size{\cS(y)}) - \log(\size{\cL(y)}) > 2\qDelta \secParQ^d$ holds for at least $2\qDelta$ fraction of the
		$y$'s in $\zo^\secParQ$ (otherwise, $\Exp\bigl[\log(\size{\cS(Y)}) - \log(\size{\cL(Y)})\bigr]
		\leq (1-2\qDelta) 2\qDelta \secParQ^d + 2\qDelta \cdot (\secParQ^d+1) < 5\qDelta n^d$).
		Since \cref{eq:Av:4} holds for a $(1-\qDelta)$ fraction of the $y$'s in $\zo^{\secParQ}$, there exists a set $\Good\subseteq \zo^\secParQ$ of density $\qDelta$ such that for every  $y \in \Good$, both
		\begin{align}
			\label{eq:setMy}
			\log(\size{\cL(y)}) < \log(\size{\cS(y)}) - 2 \qDelta \secParQ^d, \text{ and}\\
			\Pr[f(\Ac_2(y,J,G,1^n,\qDelta;R)) = X_y \mid (J,G) \in \cE'(y) \land R \in \cR'(J,G)]
			\geq 1-\qDelta\label{eq:fEqualsX}
		\end{align}
		hold.
		It follows that for every $y \in \Good$
		\begin{align}\label{eq:Av:5}
			\size{\cL(y)} < \size{\cS(y)} (1-1.5\qDelta)
		\end{align}
		\cref{eq:Av:5} trivially holds in case $\cL(y) = 0$, where in case $\size{\cL(y)} > 0$, it holds that
		\begin{enumerate}
			\item $\size{\cL(y)} < \size{\cS(y)} \cdot \eee^{-2\qDelta \cdot \log \size{\cS(y)}}$ (since $\size{\cS(y)} > \size{\cL(y)} \geq 1$) and
			
			\item $\eee^{-2\qDelta \cdot \log \size{\cS(y)}} \leq \size{\cS(y)} \cdot (1-1.5\qDelta)$ (since $\eee^{-\epsilon \kappa}\leq \eee^{-\epsilon} < 1-0.75\epsilon$ for $\kappa \geq 1$ and $\epsilon<\frac12$).
		\end{enumerate}
		\cref{eq:Av:5} yields that
		\begin{align}
			\pr{X _y\in \cS(y) \setminus \cL(y) | (J,G) \in \cE'(y) \land R \in \cR'(J,G)}& = \pr{X_y \in \cS(y) \setminus \cL(y) | (J,G) \in \cE'(y)}\\
			&\geq 1.5\qDelta\nonumber
		\end{align}
		for every $y \in \Good$, and therefore \cref{eq:fEqualsX} yields that
		\begin{align}
			\Pr\Bigl[f(\Ac_2(y,J,G,1^\secParV,\vDelta;\Random)) \in \cS(y) \setminus \cL(y)
			\Bigm| (J,G) \in \cE'(y) \land R \in \cR'(J,G)
			\Bigr] \geq \frac{\qDelta}{2}
		\end{align}
		for every $y \in \Good$.
		
		Since $\Pr[Y \in \Good] \geq \qDelta$ and since $\pr{(J,G) \in \cE'(y) \land R \in \cR'(J,G)} \geq \frac{2}{3}(1-\qDelta)\frac{1}{8\secParQ^d}$ for every $y \in \Good$, it follows that
		\begin{align*}
			\Pr[f(\Ac_2(Y,J,G,1^\secParV,\vDelta;\Random)) \notin \cL(y) \cup \set{\perp}]
			&\geq \qDelta\cdot \tfrac{2}{3}\cdot (1-\qDelta)\cdot \frac{1}{8\secParQ^d} \cdot \frac{\qDelta}{2} \geq \frac{\qDelta^2}{50 \secParQ^d},
		\end{align*}
		proving the claim and thus the lemma.
	\end{proof}
\end{proof}

\newcommand{\ent}{v}
\subsubsection{Putting it together}\label{sec:avgPutTOgether}

We now use \cref{lem:impagliazzolevin-part1N,lem:impagliazzolevin-part2} to prove \cref{lem:impagliazzo-levin-main} (and as we
have seen, this implies \cref{thm:impagliazzolevin}).
\begin{proof}[Proof of \cref{lem:impagliazzo-levin-main}]
	\cref{lem:impagliazzolevin-part1N} yields that $(\RelQ,\cU)$ has
	\ioS $(\qDelta/\secParQ)^{2d}$-accessible average max-entropy at most
	$\ent - \sqrt[2d]{\qDelta}\cdot\secParQ^{2d}$ \wrt $f$, for $\ent = \Exp_{I,G}[\log(\size{\cS_{\RelQ,f}(I,G)})]$. It follows that $(\RelQ,\cU)$ has \ioS $\frac{\qDelta}{50 \secParQ^d}$-accessible average max-entropy at most $\ent - 5 \qDelta \secParQ^d$ \wrt $f$, and the proof follows by \cref{lem:impagliazzolevin-part2}.
\end{proof}

\remove{
	\subsubsection{Comparison with the original construction of Impagliazzo and Levin}\label{sec:ComparisionToIL}
	
	To compare our construction more closely with the orignal one given in
	\cite{ImpLev90}, we first try to take a more global view of our proof
	of \cref{lem:impagliazzo-levin-main}. 
	
	We assume that there exists an $\NP$-search problem $(\RelR,\cD)$ with
	$\cD \in \polySamp$ such that $(\RelR,\cD) \notin \SearchHeurBPP$.
	Assuming for simplicity that we can choose $d = 1$ everywhere and
	ignoring dependencies on $n$, combining \cref{def:npRelationQ,def:rPrime}, we have
	set
	\begin{align}\label{eq:eqilfinal}
		\RelV \eqdef \Bigl\{
		\bigl((g_1, g_2, i, j),(x,w)\bigr)
		& \colon g_1,g_2 \in \g,\, i, j \in [n], y := \Dc(x), \\
		&\phantom{\colon}
		g_1(y)_{1\ldots i} = 0^i,g_2(x)_{1\ldots j} = 0^j, \nonumber \\
		&\phantom{\colon}
		(y,w)\in \RelR
		\Bigr\} \nonumber
	\end{align}
	In other words, given two hash functions $g_1$ and $g_2$, and
	integers $i$ and $j$, we are asked to find a preimage $x$ such that $g_1(y= \Dc(x))_{1\ldots
		i} = 0^i$ and $g_2(x)_{1\ldots j} = 0^j$, and  a witness $w$ for $y$ under $\RelR$.
	
	We now explain an improvement of the efficiency of the construction.
	Intuitively, when we use \cref{lem:valiantvazirani} in the proof
	of~\cref{lem:impagliazzolevin-part2} to get \cref{eq:Eyvv}, we can see
	that we are hoping for $j \approx \log_2(|\cS(y)|)$, where $\cS(y)$ is
	the set of elements $x$ in which the sampler produces $y$. Conversely,
	in the proof of \cref{lem:impagliazzolevin-part1N} we essentially hope
	for $H_{f(X)}(f(x)) \approx i$
	(cf. \cref{claim:two-universal-picking}), which means that we hope
	that $y$ has probability roughly $2^{-i}$. Since of course the
	probabililty of $y$ equals $2^n / |\cS(y)|$, it is intuitive that one
	might require $i+j = n$, instead of picking them independently and
	hoping for both events to occur at random (which we do). One would do
	this by not even encoding $j$ into the first part of $\RelV$, and
	instead only encoding $g_1$, $g_2$ and $i$, and replacing $j$ on the
	right hand side of \cref{eq:eqilfinal} with $n-i$.
	
	This indeed works, and the resulting construction is the same as the
	original one given in \cite{ImpLev90}, with the exception of some
	irrelevant choices to some encodings. We chose to split the
	construction into two parts to highlight the connection with
	computational entropy.
}

\section*{Acknowledgments}
We are thankful to Ran Raz and Chiu-Yuen Koo for useful conversations.

\bibliographystyle{abbrvnat}
\bibliography{crypto}

\end{document}